\documentclass[manyauthors]{fundam}

\usepackage{graphicx}

\usepackage{amsfonts}
\usepackage{amsmath} 
\usepackage{amssymb}

\usepackage{tikz}
\usepackage{varioref,bbm,dsfont}
\usetikzlibrary{petri}
\usetikzlibrary{decorations,decorations.pathreplacing,decorations.pathmorphing,decorations.markings}

\usepackage{tkz-graph}
\usetikzlibrary{shapes.geometric}
\usetikzlibrary{shapes.symbols}
\usetikzlibrary{arrows,arrows.meta,backgrounds,positioning,fit}
\usetikzlibrary{matrix,calc,shapes,patterns}
\usepackage[T1]{fontenc} 
\usepackage{mathbbol}
\DeclareSymbolFontAlphabet{\mathbbl}{bbold}

\tikzstyle{mypetristyle}=[
      place/.style={circle, draw=black, fill=white, thick, minimum size=4.5mm},
      transition/.style={rectangle, draw=black, fill=black!30, thick, minimum height=3mm, minimum width=3mm},
      token/.style={circle,draw=black,fill=black,inner sep=0pt,minimum size=1mm}
]

\usepackage[linesnumbered,ruled]{algorithm2e}

\def\N{\mathbb{N}}
\def\Z{\mathbb{Z}}

\def\projection#1#2{\mathchoice
              {\setbox1\hbox{${\displaystyle #1}_{\scriptstyle #2}$}
              \projectionaux{#1}{#2}}
              {\setbox1\hbox{${\textstyle #1}_{\scriptstyle #2}$}
              \projectionaux{#1}{#2}}
              {\setbox1\hbox{${\scriptstyle #1}_{\scriptscriptstyle #2}$}
              \projectionaux{#1}{#2}}
              {\setbox1\hbox{${\scriptscriptstyle #1}_{\scriptscriptstyle #2}$}
              \projectionaux{#1}{#2}}}
\def\projectionaux#1#2{{#1\,\smash{\vrule height .8\ht1 depth .85\dp1}}_{\,#2}}

\tikzstyle{ltsNode}=[circle,fill=black, inner sep=0, minimum width=4pt] 
\tikzstyle{ltsnst}=[circle,draw=black,fill=black!10, very thick, minimum width=1mm]
\tikzstyle{ltsNodePattern}=[circle,fill=black!10, very thick, minimum width=1mm]
\tikzstyle{petriNode}=[place,minimum size=3mm,very thick,fill=black!10]

\tikzset{every picture/.style={mypetristyle}}
\tikzset{every label/.style={black!90}}

\usepackage{lineno,hyperref}
\modulolinenumbers[5]

\newcommand{\ie}{i.e.\ }

\newcommand{\WMGineq}{\texorpdfstring{WMG$_\le$}}
\newcommand{\MGineq}{\texorpdfstring{MG$_\le$}} 
\newcommand{\PCMGineq}{\texorpdfstring{PCMG$_\le$}}

\newcommand{\es}{\emptyset}
\newcommand{\pminus}{
\mbox{\textrm{$-$}\!\!\!\!\!\!\!\:\,\,\raisebox{1.5mm}{$\scriptstyle \bullet$}}\:}
\newcommand{\dt}{\bullet}
\newcommand{\lbul}{^\bullet}
\newcommand{\emptyseq}{\varepsilon}
\newcommand{\support}{\mathit{supp}}

\newcommand{\Parikh}{{\mathbf P}}
\newcommand{\is}{\iota}
\newcommand{\one}{\mathbbl{1}}
\newcommand{\zero}{\mathbbl{0}}

\begin{document}
\setcounter{page}{1001}
\issue{XXI~(2001)}

\title{Checking marking reachability with the state equation in Petri net subclasses}


\author{Thomas Hujsa\thanks{Supported by the STAE foundation/project DAEDALUS, Toulouse, France.}\\  
LAAS-CNRS\\Universit\'e de Toulouse, CNRS, INSA\\Toulouse, France
({thomas.hujsa@laas.fr})
\and Bernard Berthomieu\\
LAAS-CNRS\\Universit\'e de Toulouse, CNRS, INSA\\Toulouse, France
({bernard.berthomieu@laas.fr})
\and Silvano Dal Zilio\\
LAAS-CNRS\\Universit\'e de Toulouse, CNRS, INSA\\Toulouse, France
({silvano.dalzilio@laas.fr})
\and Didier Le Botlan\\
LAAS-CNRS\\Universit\'e de Toulouse, CNRS, INSA\\Toulouse, France
({didier.le.botlan@laas.fr})
}

\maketitle

\runninghead{T. Hujsa, B. Berthomieu, S. Dal Zilio, D. Le Botlan}{Checking Petri Nets Properties Using the State Equation}


\begin{abstract}
Although decidable, the marking reachability problem for Petri nets is well-known to be intractable in general,
and a non-elementary lower bound has been recently uncovered.
In order to alleviate this difficulty,
various structural and behavioral restrictions have been considered, 
allowing to relate reachability to properties that are easier to check.
For a given initial marking,
the set of potentially reachable markings is described by the state equation solutions
and over-approximates the set of reachable markings.

In this paper, we delineate several subclasses of weighted Petri nets in which 
the set of reachable markings equals the set of potentially reachable ones, a property we call the PR-R equality.
When fulfilled, this property allows to use linear algebra to answer the reachability questions,
avoiding a brute-force analysis of the state space.
Notably, we provide conditions under which this equality holds in classes much more expressive than marked graphs, 
adding places with several ingoing and outgoing transitions, which allows to model real applications with shared buffers.
To achieve it, we investigate the relationship between liveness, reversibility, boundedness and potential reachability in Petri nets.
We also show that this equality does not hold in classes with close modeling capability when the conditions are relaxed.
\end{abstract}


\section{Introduction}\label{intro.sec}

Petri nets, or equivalently vector addition systems (VASs), have proved useful to model numerous artificial and natural systems.
Their weighted version allows weights (multiplicities) on arcs,
making possible the bulk consumption or production of tokens, 
hence a more compact representation of the systems.

For many fundamental Petri net properties, the problem of their checking is decidable although intractable.
Given a bounded Petri net, a naive analysis can be performed by constructing its finite reachability graph,
whose size may be considerably larger than the net size. 
To avoid such a costly computation, subclasses are often considered, allowing to derive efficiently their behavior from their structure only.
This approach has led to various polynomial-time checking methods dedicated to several subclasses,
the latter being defined by structural restrictions in many cases~\cite{DesEsp,STECS,LAT98,HDM2016,HD2017,HD2018}.\\

\noindent {\bf The reachability problem.}
Given a Petri net system $S$, 
the problem is to determine if a given marking 
is reachable in $S$.
In weighted Petri nets, this question was known to be EXPSPACE-hard~\cite{Lipton76};
recently, a non-elementary lower bound has been obtained~\cite{ReachNonElementary2019}.
Reachability reduces to various well-known model-checking problems~\cite{ReachNonElementary2019}, which thus inherit this lower bound.\\

\noindent {\bf Relating reachability to the state equation.} 
Consider a system $S=(N,M_0)$ with its incidence matrix $I$, where $N=(P,T,W)$ is the underlying net
($P$ being the set of places, $T$ the set of transitions and $W$ the weighting function).

The {\em state equation} associated to $S$ is expressed as $M = M_0 + I \cdot Y$,
where the variable $M$ takes its value in the set of markings 
and the variable $Y$ ranges over the set of vectors whose components are non-negative integers.

The set of markings {\em potentially reachable} in $S$ is defined as 
$PR(S) = \{ M \in \N^{\mid P\mid} \mid \exists Y \in \N^{\mid T\mid}, M = M_0 + I \cdot Y \}$;
this set is called the \emph{linearized reachability set} of $S$ in~\cite{LAT98}.
%
%
Potential reachability is a necessary condition for reachability, but it is not sufficient in general;
a Petri net satisfies the {\em PR-R equality} if its reachable markings are its potentially reachable ones.
Thus, in the subclasses that are known to fulfill the PR-R equality, 
solving the reachability problem amounts to check the existence of a solution
to an integer linear program (ILP) of polynomial size, trimming down its complexity to~NP.\\

\noindent {\bf Petri net subclasses, applications and previous studies.}
In this work, we study conditions for the PR-R equality to hold in weighted Petri nets and several of their subclasses, notably:\\
$-$ Weighted Marked Graphs with relaxed place constraint (\WMGineq{} for short), 
which force each place to have at most one input and one output, studied e.g.\ in~\cite{DH18,ArxivSSP20};\\
$-$ Augmented Marked Graphs (AMG), which are unit-weighted and allow the addition of several shared\footnote{A place is shared if it has at least two outgoing transitions.} 
places to a (unit-weighted) Marked Graph under some restrictions~\cite{ChuXie97};\\
$-$ H$1$S-\WMGineq{}, i.e.\ homogeneous\footnote{Homogeneity means that, for each shared place $p$, all the output weights of $p$ are equal.} nets having at most one shared place, the deletion of which yields a \WMGineq{}~\cite{ArxivSSP20}; this class thus contains the \WMGineq{};\\
$-$ Place-Composed Marked Graphs with relaxed place constraints (\PCMGineq{}), obtained from a given undirected graph by replacing each edge with a marked graph 
and each vertex with a shared place, which is a kind of synthesis.

As far as we know, \PCMGineq{} have not been studied until now, while the previous works on AMG and H$1$S-\WMGineq{} did not focus on the PR-R equality.
These net classes, although very restricted, can already model numerous real-world applications. Let us present some of them:

\noindent $-$ \WMGineq{} generalize the Marked Graphs (MG)/Weighted Marked Graphs (WMG)/Weighted Event Graphs (WEG)/Weighted T-systems (WTS)~\cite{MDG71,Sauer2003,March09,DesEsp,WTS92}, in which each place has exactly one input and one output.
They are a special case of persistent systems~\cite{ErgoMonoT1991}, in which no transition firing can disable any other transition.
They can model Synchronous DataFlow graphs~\cite{LeeM87}, which have been fruitfully used to design and analyze many real-world systems
such as embedded applications, notably Digital Signal Processing (DSP) applications \cite{LM87a,Emb09,Pino95}.
Various analysis and synthesis methods have been developed for \WMGineq{}~\cite{WTS92,March09,EHW18,ATAEDDEH18,DH18,DH19FI,ATAEDDEH19,ToPNoCDEH19,ArXivDEHToPNoC1920} 
and larger classes~\cite{STECS,HDM2016,ArxivSSP20},
dealing notably with reachability, liveness, boundedness and reversibility.
In the same studies, some relationships between the behavior and the state equation are also provided.

\noindent $-$ AMG, although unit-weighted, can model various manufacturing systems~\cite{ChuXie97,cheung2008augmented}, 
the {\em dining philosophers} problem~\cite{cheung2008augmented},
and
several use-cases of the Model Checking Contest\footnote{\url{https://mcc.lip6.fr/models.php}} (MCC)
such as the {\em Swimming pool} protocol, the {\em Robot manipulation} system, the {\em Client and server} protocol 
and the process-management method {\em Kanban} (under some conditions to be fulfilled by the initial marking).
AMG benefit from results on their liveness, reversibility and reachability~\cite{ChuXie97}.

\noindent $-$ H$1$S-\WMGineq{} extend both $1$S-AMG (i.e. AMG with at most one shared place) and \WMGineq{}, 
making more flexible the modeling of applications with \WMGineq{}.
The Swimming pool protocol can also be modeled with an H$1$S-\WMGineq{}~\cite{ArxivSSP20}.

\noindent $-$ The class of \PCMGineq{}, which is not included in the classes above nor contains them,
allows to define a system first in terms of its topology of shared buffers, before refining the processes that connect and use these buffers.
This model can be used for system synthesis under structural and behavioral constraints.
Previous works, described in~\cite{Jiao2004}, have proposed conditions for merging sets of places into shared places
while preserving various structural and behavioral properties, notably in unit-weighted asymmetric-choice Petri nets.
However, the subclasses studied in~\cite{Jiao2004} do not contain the \PCMGineq{}.\\

\noindent {\bf Contributions.}
In this paper, we propose new sufficient conditions ensuring the PR-R equality:
a general condition applying to all weighted Petri nets,
and other ones dedicated to the subclasses mentioned above.
When such conditions are known to be fulfilled, 
checking reachability then reduces to solving the state equation with linear algebra over the integers, trimming down the complexity to NP.
So as to obtain these results, we exploit the next notions:\\ 
$-$ {\em directedness}: a property stating the existence, for any two potentially reachable markings $M_1$ and $M_2$, 
of a marking reachable from both $M_1$ and $M_2$; \\
$-$ {\em initial directedness}: a property stating the existence, for each potentially reachable marking $M_1$,
of a marking reachable from both $M_1$ and the initial marking $M_0$;\\
$-$ {\em the reverse net}: obtained by reversing all the arcs;\\
$-$ {\em liveness}: a property stating the possibility, from each reachable marking, to fire some sequence containing all transitions;\\
$-$ {\em boundedness}: a property stating the existence of an upper bound on the number of tokens of each place over all reachable markings;\\
$-$ {\em reversibility\footnote{The different notion of \emph{reversible computation} has been investigated in~\cite{RevCompVSRev2016,RevCompVSRev2018}:
contrarily to the global property of reversibility, 
reversible computation is a local mechanism that a system can use to undo some of the executed actions~\cite{RevCompVSRev2016,RevCompVSRev2018}.}}: a property stating the possibility to reach the initial marking from each reachable marking, 
meaning the strong connectedness of the reachability graph;\\
$-$ the property $\mathcal{R}$: we introduce this property for any Petri net system, stating reversibility of both the system and its reverse.

More precisely, we exploit these notions as follows.

First, we show that combining property $\mathcal{R}$ with initial directedness is sufficient to ensure the PR-R equality in any weighted Petri net.
This new condition is not necessary is general, but we show its tightness for live and bounded \WMGineq{}.
We also present new results on liveness and deadlockability in \WMGineq{}, which help checking the precondition of liveness.

Then, we improve our understanding of the relationship between the state equation solutions 
(potential reachability), reachability, liveness and reversibility
in the mentioned generalizations of marked graphs with shared places:
in AMG, H$1$S-\WMGineq{} and \PCMGineq{}, we provide new sufficient conditions ensuring that the PR-R equality is fulfilled.
We highlight the sharpness of all conditions by providing counter-examples when only few assumptions are relaxed.
We also propose methods to check the various conditions and give insight on their complexity in the subclasses mentioned
and sometimes in larger ones.
Notably, in a subclass of \PCMGineq{}, we provide a variant of Commoner's theorem and of the Home Marking theorem
which were developed for free-choice nets to characterize the live and reversible markings in polynomial-time~\cite{DesEsp}.

This work is the sequel to our previous paper~\cite{ArxivSSP20}, in which we provided conditions for checking reachability,
liveness and reversibility more efficiently in some subclasses.
Since we often use liveness and reversibility as preconditions for the PR-R equality to hold,
these previous results can be exploited to reduce their checking complexity.\\

\noindent {\bf Organization of the paper.}
In Section~\ref{SectionDefs}, we introduce general definitions, notations and properties.
In Section~\ref{SectionSubclasses}, we define the main subclasses studied in this paper and compare their expressiveness.

In Section~\ref{PotentialReachDirected}, we define the main notions related to {\em directedness} and recall related properties.
We also recall known classes of the literature that fulfill directedness.

In Section~\ref{PRRrevDir}, we give new properties of nets and their reverse;
we provide notably a general sufficient condition for the PR-R equality,
using reversibility, initial directedness and reverse nets.
We apply this result to ensure the PR-R equality in the class of live homogeneous free-choice (HFC) nets, 
a weighted generalization of free-choice nets.

In Section~\ref{ReachWMG}, we show that live \WMGineq{} fulfill the PR-R equality,
we propose new characterizations of liveness and a new property about reachable deadlocks in this class, based on the state equation.
We also discuss methods to check the behavioral properties of interest in \WMGineq{}.

In Section~\ref{UnreachforSharedPlaces}, we construct new examples of systems that do not fulfill the PR-R equality.
 They belong to the $2$S-\WMGineq{} subclass, i.e.\ the class of nets with at most $2$ shared places, the deletion of which yields a \WMGineq{}.
On these examples, we emphasize possible causes of roadblocks to the PR-R equality.

In Sections~\ref{Sec1HEWMG}, \ref{SecAMG} and \ref{SecPCMG},
we study several subclasses of S-\WMGineq{}, namely H$1$S-\WMGineq{}, AMG and \PCMGineq{}.
We provide for them conditions that ensure the PR-R equality, exploiting the examples of Section~\ref{UnreachforSharedPlaces}.
We also discuss methods for checking these conditions in the classes studied.

In Section~\ref{SecRelatedWork}, we further discuss related works.

Finally, Section~\ref{SecConclu} presents our conclusion with perspectives.

\section{General Definitions, Notations and Properties}\label{SectionDefs}

In the following, we define formally Petri nets, related notions and properties.\\

\noindent {\bf Petri nets, incidence matrices, pre- and post-sets, shared places.}
A {\em (Petri) net} is a tuple $N=(P,T,W)$ such that 
$P$ is a finite set of {\em places}, 
$T$ is a finite set of {\em transitions}, with $P\cap T=\es$, 
and $W$ is a weight function $W\colon((P\times T)\cup(T\times P))\to\N$ 
setting the weights on the arcs.
A {\em marking} of the net $N$ 
is a mapping from $P$ to $\N$, i.e. a member of $\N^P$, 
defining the number of tokens in each place of~$N$.

A {\em (Petri net) system}
is a tuple $S=(N,M_0)$
where $N$ is a net 
and
$M_0$ is a marking, often called {\em initial marking}. 
The {\em incidence matrix} $I$ of $N$ (and $S$) is the integer place-transition matrix 
with components $I(p,t)=W(t,p)-W(p,t)$,
for each place $p$ and each transition~$t$.

The {\em post-set} $n\lbul$ and {\em pre-set} $\lbul n$ of a node $n \in P \cup T$ are defined
as $n\lbul=\{n'\in P\cup T\mid W(n,n'){>}0\}$ and $\lbul n=\{n'\in P\cup T \mid W(n',n){>}0\}$.

A place $p$ is {\em shared} if it has at least two outputs, i.e. $|p\lbul| \ge 2$.

Some of these notions are illustrated in Figure \ref{ex1.fig}.\\

\noindent {\bf Firings and reachability in Petri nets.}
Consider a system $S=(N,M_0)$ with $N=(P,T,W)$.
A transition $t$ is {\em enabled} at $M_0$ (i.e. in $S$) 
if for each $p$ in ${\lbul t}$, $M_0(p) \ge W(p,t)$, 
in which case $t$ is {\em feasible} or {\em fireable} from $M_0$. 
The firing of $t$ from $M_0$ leads to the marking $M = M_0 + I[P,t]$ 
where $I[P,t]$ is the column of $I$ associated to $t$: this is denoted by $M_0[t\rangle M$.

A finite {\em (firing) sequence} $\sigma$ of length $n \ge 0$ on the set $T$,
denoted by $\sigma = t_{1} \ldots t_{n}$ with $t_{1} \ldots t_{n} \in T$,
is a mapping $\{1, \ldots, n\} \to T$.
Infinite sequences are defined similarly as mappings $\N\setminus\{0\} \to T$.
A sequence $\sigma$ of length $n$ is {\em enabled} (or {\em feasible}, {\em fireable}) in $S$ 
if the successive states obtained,
$M_0 [t_1\rangle M_1 \ldots [t_n\rangle M_n$,
satisfy 
$M_{k-1}[t_k\rangle M_k$, for each $k$ in $\{1,\ldots, n\}$,
in which case $M_n$ is said to be {\em reachable} from $M_0$: we denote this by $M_0 [\sigma\rangle M_n$.
If $n=0$, $\sigma$ is the {\em empty sequence} $\epsilon$, implying $M_0 [\epsilon\rangle M_0$.
The set of markings reachable from $M_0$ is denoted by $R(S)$ or $[S\rangle$;
when it is clear from the context, it is also denoted by $R(M_0)$ or $[M_0\rangle$.

The {\em reachability graph} of $S$, denoted by $RG(S)$, is the rooted directed graph $(V,A,\is)$
where $V$ represents the set of vertices labeled bijectively with the markings $[M_0\rangle$,
$A$ is the set of arcs labeled with transitions of $T$
such that the arc $M \xrightarrow[]{t} M'$
belongs to $A$ if and only if
$M [t\rangle M'$ and $M \in [M_0\rangle$, and $\is$ is the root, labeled with $M_0$.\\

In Figure \ref{ex1.fig}, a weighted system is pictured on the left.
Its reachability graph is pictured on the right, where $v^T$ denotes the transpose of vector $v$.\\

\begin{figure}[htb] 
\centering
\begin{tikzpicture}[scale=0.95]

\node (p1) at (1.35,0) [place,label=below:$p_3$,tokens=4] {};
\node (p2) at (2.65,0) [place,label=below:$p_4$,tokens=3] {};
\node (p3) at (1.35,2) [place,label=above:$p_1$] {};
\node (p4) at (2.65,2) [place,label=above:$p_2$] {};

\node (t1) at (0.7,1) [transition] {};
\node (t2) at (2,1) [transition] {};
\node (t3) at (3.3,1) [transition] {};

\draw [->,thick, bend left=0] (t2) to node [below left] {$2$} (p3);
\draw [->,thick, bend left=0] (t2) to node [below right] {$1$} (p4);

\draw [->,thick, bend left=0] (p1) to node [above left] {$4$} (t2);
\draw [->,thick, bend left=0] (p2) to node [above right] {$3$} (t2);

\draw [->,thick, bend left=0] (p3) to node [above left] {$1$} (t1);
\draw [->,thick, bend left=0] (p4) to node [above right] {$1$} (t3);

\draw [->,thick, bend left=0] (t1) to node [below left] {$2$} (p1);
\draw [->,thick, bend left=0] (t3) to node [below right] {$3$} (p2);

 \node [anchor=east] at (t1.west) {$t_1$};
 \node [anchor=south] at (t2.north) {$t_2$};
 \node [anchor=west] at (t3.east) {$t_3$};

\end{tikzpicture}
\hspace*{0.35cm}
\begin{tikzpicture}[scale=0.94]

\node[](s5)at(1.2,1.5)[]{\footnotesize $(0,1,4,0)^T$};
\node[fill=black!15](s0)at(3,2.75)[]{{\footnotesize $(0,0,4,3)^T$}}; 
\node[](s1)at(3,1.5)[]{\footnotesize $(2,1,0,0)^T$}; 
\node[](s3)at(3,0.35)[]{\footnotesize $(1,1,2,0)^T$};
\node[](s2)at(5.35,1.5)[]{\footnotesize $(2,0,0,3)^T$};
\node[](s4)at(7.7,1.5)[]{\footnotesize $(1,0,2,3)^T$};

\draw(s0)[]edge[-latex,bend left=0]node[ellipse,right,inner sep=2pt,pos=0.5]{$t_2$}(s1);
\draw(s1)[]edge[-latex,bend left=0]node[ellipse,right,inner sep=2pt,pos=0.5]{$t_1$}(s3);
\draw(s1)[]edge[-latex,bend left=0]node[ellipse,above,inner sep=2pt,pos=0.5]{$t_3$}(s2);
\draw(s2)[]edge[-latex,bend left=0]node[ellipse,above,inner sep=2pt,pos=0.5]{$t_1$}(s4);
\draw(s4)[]edge[-latex,bend right=20]node[ellipse,above,inner sep=2pt,pos=0.5]{$t_1$}(s0);
\draw(s5)[]edge[-latex,bend left=0]node[ellipse,above left,inner sep=2pt,pos=0.5]{$t_3$}(s0);
\draw(s3)[]edge[-latex,bend left=0]node[ellipse,below left,inner sep=2pt,pos=0.5]{$t_1$}(s5);
\draw(s3)[]edge[-latex,bend right=20]node[ellipse,above,inner sep=2pt,pos=0.5]{$t_3$}(s4);

\end{tikzpicture}

\vspace*{5mm}

\caption{ 
A system $S=(N,M_0)$ is pictured on the left. 
The pre-set $^\bullet t_2$ of $t_2$ is $\{p_3, p_4\}$
and
the post-set $t_2^\bullet$ of $t_2$ is $\{p_1, p_2\}$.
There is no shared place.
The reachability graph $RG(S)$ of $S$ is pictured on the right.
The initial marking is the grey state in $RG(S)$.
The firing sequence $\sigma = t_2\, t_1\, t_3$ is feasible from $M_0$ and reaches the marking $(1,0,2,3)^T$.
} 
\label{ex1.fig}
\end{figure}
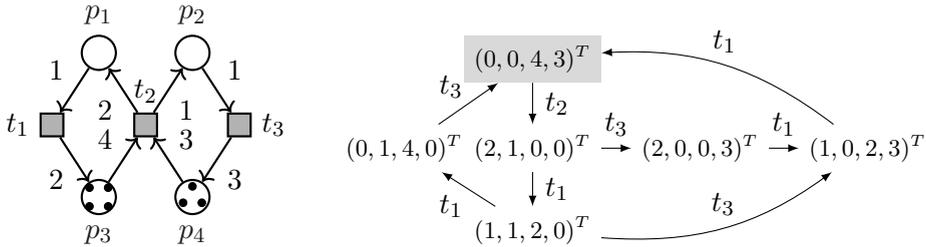

\noindent{\bf Subnets and subsystems.} 
Let $N=(P,T,W)$ and $N'=(P',T',W')$ be two nets.
$N'$ is a {\em subnet} of $N$ if $P'$ is a subset of $P$, $T'$ is a subset of $T$, 
and $W'$ is the restriction of $W$ to $(P'\times T') \cup(T'\times P')$.
$S'=(N',M_0')$ is a {\em subsystem} of $S=(N,M_0)$ if $N'$ is a subnet of $N$ 
and its initial marking $M_0'$ is the restriction of $M_0$ to $P'$, denoted by $M_0'= \projection{M_0}{P'}$.

$N'$ is a {\em P-subnet} of $N$ if $N'$ is a subnet of $N$
and $T'=\mathit{^\bullet P'} \cup P'^\bullet$,
the pre- and post-sets being taken in $N$. 
$S'=(N',M_0')$ is a {\em P-subsystem} of $S=(N,M_0)$ if $N'$ is a P-subnet of $N$ 
and $S'$ is a subsystem of $S$.
We say that $N'$ and $S'$ are {\em induced} by the subset $P'$.

Similarly, $N'$ is a {\em T-subnet} of $N$ if $N'$ is a subnet of $N$ 
and $P'=\mathit{^\bullet T'} \cup T'^\bullet$, the pre- and post-sets being taken in $N$. 
$S'=(N',M_0')$ is a {\em T-subsystem} of $S=(N,M_0)$ if $N'$ is a T-subnet of $N$ 
and $S'$ is a subsystem of $S$.
We say that $N'$ and $S'$ are {\em induced} by the subset $T'$.

Subsystems play a fundamental role in the analysis of Petri nets,
typically leading to characterizations relating the system's behavior to properties of its subsystems;
this approach yielded polynomial-time checking methods in various subclasses, e.g.~\cite{March09,HDM2016,HD2018}.
We exploit such subsystems in this paper to obtain some of our new results on reachability.

Examples are given in Figure~\ref{ex2.fig}.\\

\begin{figure}[htb] 
\centering
\begin{tikzpicture}[scale=0.95]

\node (p1) at (1.35,0) [place,label=below:$p_3$,tokens=4] {};
\node (p2) at (2.65,0) [place,label=below:$p_4$,tokens=3] {};
\node (p3) at (1.35,2) [place,label=above:$p_1$] {};
\node (p4) at (2.65,2) [place,label=above:$p_2$] {};

\node (t1) at (0.7,1) [transition] {};
\node (t2) at (2,1) [transition] {};
\node (t3) at (3.3,1) [transition] {};

\draw [->,thick, bend left=0] (t2) to node [below left] {$2$} (p3);
\draw [->,thick, bend left=0] (t2) to node [below right] {$1$} (p4);

\draw [->,thick, bend left=0] (p1) to node [above left] {$4$} (t2);
\draw [->,thick, bend left=0] (p2) to node [above right] {$3$} (t2);

\draw [->,thick, bend left=0] (p3) to node [above left] {$1$} (t1);
\draw [->,thick, bend left=0] (p4) to node [above right] {$1$} (t3);

\draw [->,thick, bend left=0] (t1) to node [below left] {$2$} (p1);
\draw [->,thick, bend left=0] (t3) to node [below right] {$3$} (p2);

 \node [anchor=east] at (t1.west) {$t_1$};
 \node [anchor=south] at (t2.north) {$t_2$};
 \node [anchor=west] at (t3.east) {$t_3$};

\end{tikzpicture}
\hspace*{0.8cm}
\begin{tikzpicture}[scale=0.95]

\node (p1) at (1.35,0) [place,label=below:$p_3$,tokens=4] {};
\node (p2) at (2.65,0) [place,label=below:$p_4$,tokens=3] {};
\node (p3) at (1.35,2) [place,label=above:$p_1$] {};
\node (p4) at (2.65,2) [place,label=above:$p_2$] {};

\node (t2) at (2,1) [transition] {};

\draw [->,thick, bend left=0] (t2) to node [below left] {$2$} (p3);
\draw [->,thick, bend left=0] (t2) to node [below right] {$1$} (p4);

\draw [->,thick, bend left=0] (p1) to node [above left] {$4$} (t2);
\draw [->,thick, bend left=0] (p2) to node [above right] {$3$} (t2);

\node [anchor=south] at (t2.north) {$t_2$};

\end{tikzpicture}
\hspace*{0.8cm}
\begin{tikzpicture}[scale=0.95]

\node (p1) at (1.35,0) [place,label=below:$p_3$,tokens=4] {};
\node (p3) at (1.35,2) [place,label=above:$p_1$] {};

\node (t1) at (0.7,1) [transition] {};
\node (t2) at (2,1) [transition] {};

\draw [->,thick, bend left=0] (t2) to node [below left] {$2$} (p3);

\draw [->,thick, bend left=0] (p1) to node [above left] {$4$} (t2);

\draw [->,thick, bend left=0] (p3) to node [above left] {$1$} (t1);

\draw [->,thick, bend left=0] (t1) to node [below left] {$2$} (p1);

\node [anchor=east] at (t1.west) {$t_1$};
\node [anchor=south] at (t2.north) {$t_2$};

\end{tikzpicture}


\caption{ 
A system $S=(N,M_0)$ is pictured on the left. 
The T-subsystem of $S$ induced by $\{t_2\}$ is pictured in the middle,
and 
the P-subsystem of $S$ induced by $\{p_1, p_3\}$ is pictured on the right.
Notice that, for any T-subsystem $S'$ of $S$, each sequence feasible in $S'$ is also feasible in $S$.
} 
\label{ex2.fig}
\end{figure}
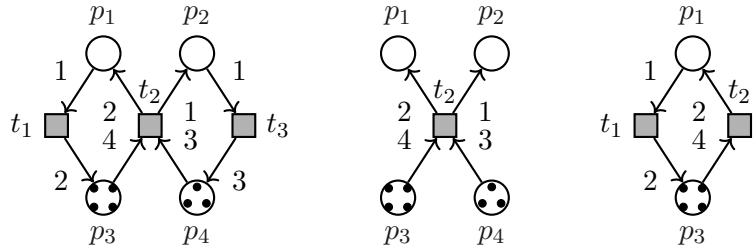

\noindent{\bf Siphons and traps.}
Consider a net $N = (P,T,W)$.
A 
subset $D \subseteq P$ of places is a {\em siphon} (sometimes also called a deadlock) 
if $\lbul D \subseteq D \lbul$.
A 
subset $Q \subseteq P$ of places is a {\em trap} if $Q\lbul \subseteq {\lbul Q}$.
Siphons and traps are most often assumed to be non-empty;
for the sake of conciseness, we allow emptiness explicitly when it is needed. 

There exist various studies relating the structure to the behavior with the help of siphons and traps.
Intuitively,
insufficiently marked siphons induce P-subsystems that cannot receive new tokens and thus block some transitions irremediably,
while marked traps always keep some token and favor the enabledness of outgoing transitions,
at least in some ordinary subclasses; see e.g.~\cite{DesEsp,STECS}.

A siphon (respectively trap) is {\em minimal} if it does not contain any proper siphon (respectively trap),
i.e. there is no subset of the same type with smaller cardinality.

In Figure~\ref{ex2.fig}:
on the left,
$\{p_1, p_2, p_3, p_4\}$ is both a siphon and a trap,
and includes smaller ones, namely $\{p_1, p_3\}$ and $\{p_2, p_4\}$,
while
$\{p_1\}$ is neither a siphon nor a trap;
in the middle, $\{p_3\}$ is a minimal siphon and is not a trap,
while $\{p_1\}$ is a minimal trap and is not a siphon.\\

\noindent {\bf Vectors, semiflows, conservativeness and consistency.} 
The {\em support} of a vector is the set of the indices of its non-null components.
Consider any net $N=(P,T,W)$ with its incidence matrix~$I$.

A {\em T-vector} (respectively P-vector) is an element of $\N^{T}$ (respectively $\N^{P}$); it is called {\em prime} if 
the greatest common divisor of its components is one 
(i.e. its components do not have a common non-unit factor). 
It is called {\em minimal} when it is prime and its support is not a proper superset of the support of any other T-vector.
The {\em cardinality} of a T-vector is the sum of its components; for instance, the cardinality of $(1,0,2,5)$ is $8$. 

The {\em Parikh vector} $\Parikh(\sigma)$ of a finite sequence $\sigma$ of transitions is the T-vector
counting the number of occurrences of each  transition in $\sigma$,
and the {\em support} of $\sigma$ is the support of its Parikh vector, 
i.e.
$\support(\sigma)=\support(\Parikh(\sigma))=\{t\in T\mid\Parikh(\sigma)(t)>0\}$.

We denote by $\zero^n$ (respectively $\one^n$) the column vector of size $n$ whose components are all equal to~$0$ (respectively $1$).
The exponent $n$ may be omitted when it is clear from the context.

A {\em T-semiflow} (respectively P-semiflow) $Y$ of the net is a non-null T-vector (respectively P-vector)
whose components are only non-negative integers (i.e. $Y\gneqq\zero$)
and such that $I \cdot Y=\zero$ (respectively $Y^T \cdot I = \zero$).

$N$ is {\em conservative}, or {\em invariant}, if a P-semiflow $X \in \N^{|P|}$ exists for $I$
such that
$X \ge \one^{|P|}$, in which case $X$ is called a {\em conservativeness vector}.
In case such a P-vector $X$ exists and, in addition, $X=\one^{|P|}$, $N$ is called {\em $1$-conservative}, or {\em $1$-invariant}.

$N$ is {\em consistent} if a T-semiflow $Y \in \N^{|T|}$ exists for $I$
such that
$Y \ge \one^{|T|}$, in which case $Y$ is called a {\em consistency vector}.

Such vectors are frequently exploited in the structural and behavioral analysis of Petri nets,
see e.g.~\cite{LAT98}.\\

\noindent {\bf State equation, potential reachability and the PR-R equality.}
Consider any system $S=(N,M_0)$ with incidence matrix $I$.
The {\em state equation} associated to $S$ is expressed as $M = M_0 + I \cdot Y$,
whose solutions are described by the variables $M$ and $Y$, denoting respectively markings and T-vectors.
The set of markings {\em potentially reachable} in $S$ 
is defined as $PR(S) = \{ M \in \N^{\mid P\mid} \mid \exists Y \in \N^{\mid T\mid}, M = M_0 + I \cdot Y \}$.
We denote by $PRG(S)$ the {\em potential reachability graph} of $S$, defined as the rooted directed graph $(V,A,\is)$
where $V$ represents the set of vertices $PR(S)$,
$A$ is the set of arcs labeled with transitions of $S$
such that, for each transition $t$, the arc $M \xrightarrow[]{t} M'$ belongs to $A$ if and only if
$M [t\rangle M'$ 
and 
$M \in PR(S)$,
and
$\is = M_0$ is the root.

A Petri net system $S$ fulfills the {\em PR-R equality} if $R(S) = PR(S)$.\\

\noindent {\bf Deadlockability, liveness, boundedness and reversibility.}
Consider any system $S=(N,M_0)$.
A transition $t$ is {\em dead} in $S$ if no marking of $[M_0\rangle$ enables $t$.
A {\em deadlock}, or {\em dead marking}, is a marking enabling no transition.
$S$ is {\em deadlock-free} if no deadlock belongs to $[M_0\rangle$; otherwise it is {\em deadlockable}. 
%

A transition $t$ is {\em live} in $S$ if for every marking $M$ in $[M_0\rangle$,
there is a marking $M' \in [M\rangle$ enabling~$t$.
$S$ is {\em live} if every transition is live in $S$.
$N$ is {\em structurally live} if a marking $M$ exists such that $(N,M)$ is live.

A marking $M$ is a {\em home state} of $S$ if it can be reached from every marking in $[M_0\rangle$. 
$S$ is {\em reversible} if its initial marking is a home state, meaning that $RG(S)$ is strongly connected.

$S$ is {\em $k$-bounded} (or {\em $k$-safe}) if an integer $k$ exists such that: for each $M$ in $[M_0\rangle$, for each place $p$, $M(p) \le k$.
It is {\em bounded} if an integer $k$ exists such that $S$ is $k$-bounded. 
$N$ is {\em structurally bounded} if $(N,M)$ is bounded for each $M$.

$N$ is {\em well-formed} if it is structurally bounded and structurally live.

The underlying net $N$ in Figure~\ref{ex1.fig} is structurally live and bounded, hence well-formed.
In the same figure, the system $S=(N,M_0)$ is live, $4$-bounded and reversible, thus non-deadlockable,
which can be checked on its finite reachability graph.

\section{Petri net subclasses}\label{SectionSubclasses}

In this section,
we define the subclasses of Petri nets studied in this paper.

\subsection{Classical restrictions on the structure}

\vspace*{2mm}

Let us define subclasses from restrictions on the structure of any net $N = (P,T,W)$.\\

\noindent {$-$ Subclasses defined by restrictions on the weights.}
%
%
$N$ is {\it ordinary} (or {\it plain}, {\it unit-weighted}) if no arc weight exceeds $1$; 
$N$ is {\em homogeneous} if for each place $p$, all outgoing weights of $p$ are equal.
In particular, ordinary nets are homogeneous.
In this paper, for any class of nets C, we denote by HC the homogeneous subclass of C.
Examples are pictured in Figures~\ref{ExSubclasses1} and~\ref{ExSubclasses2}.\\

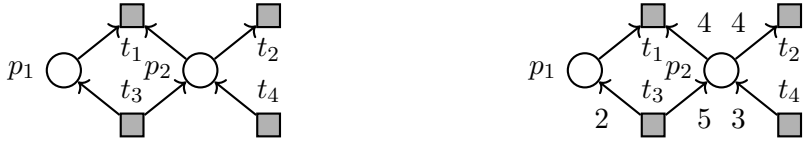
\begin{figure}[!ht]
 
\begin{tabular}{cc}
 
\begin{minipage}{0.4\linewidth}

\centering

\begin{tikzpicture}[mypetristyle,scale=0.6]

\node (p1) at (0,0) [place,thick] {};
\node (p2) at (3,0) [place,thick] {};

\node [anchor=east] at (p1.west) {$p_1$};
\node [anchor=east] at (p2.west) {$p_2$};

\node (t1) at (1.5,1.2) [transition,thick] {};
\node (t2) at (4.5,1.2) [transition,thick] {};
\node (t3) at (1.5,-1.2) [transition,thick] {};
\node (t4) at (4.5,-1.2) [transition,thick] {};

\node [anchor=north] at (t1.south) {$t_1$};
\node [anchor=north] at (t2.south) {$t_2$};
\node [anchor=south] at (t3.north) {$t_3$};
\node [anchor=south] at (t4.north) {$t_4$};

\draw [->,thick] (p1) to node [above left] {} (t1);
\draw [->,thick] (p2) to node [above right] {} (t1);
\draw [->,thick] (p2) to node [above left] {} (t2);
\draw [->,thick] (t3) to node [below left] {} (p1);
\draw [->,thick] (t3) to node [below right] {} (p2);
\draw [->,thick] (t4) to node [below left] {} (p2);

\end{tikzpicture}

\end{minipage}
&
\begin{minipage}{0.4\linewidth}

\centering

\begin{tikzpicture}[mypetristyle,scale=0.6]

\node (p1) at (0,0) [place,thick] {};
\node (p2) at (3,0) [place,thick] {};

\node [anchor=east] at (p1.west) {$p_1$};
\node [anchor=east] at (p2.west) {$p_2$};

\node (t1) at (1.5,1.2) [transition,thick] {};
\node (t2) at (4.5,1.2) [transition,thick] {};
\node (t3) at (1.5,-1.2) [transition,thick] {};
\node (t4) at (4.5,-1.2) [transition,thick] {};

\node [anchor=north] at (t1.south) {$t_1$};
\node [anchor=north] at (t2.south) {$t_2$};
\node [anchor=south] at (t3.north) {$t_3$};
\node [anchor=south] at (t4.north) {$t_4$};

\draw [->,thick] (p1) to node [above left] {} (t1);
\draw [->,thick] (p2) to node [above right] {$4$} (t1);
\draw [->,thick] (p2) to node [above left] {$4$} (t2);
\draw [->,thick] (t3) to node [below left] {$2$} (p1);
\draw [->,thick] (t3) to node [below right] {$5$} (p2);
\draw [->,thick] (t4) to node [below left] {$3$} (p2);

\end{tikzpicture}

\end{minipage}

\end{tabular}

\vspace*{4mm}

\caption{The net on the left is ordinary (i.e. unit-weighted, plain), thus in particular it is homogeneous.
The net on the right is homogeneous and is not ordinary.
Both nets have a single shared place, which is $p_2$.
}

\label{ExSubclasses1}

\end{figure}

\noindent {$-$ Subclasses without shared places.}
$N$ is {\it choice-free} (CF, also called place-output-nonbranching)
if each place has at most one output,
i.e. $\forall p\in P$, $|p^\dt|\leq 1$;
%
it is a {\em weighted marked graph with relaxed place constraints} (\WMGineq{})
if it is choice-free and, in addition, each place has at most one input,
i.e.
$\forall p\in P$, $|{}^\dt p|\leq 1$ and $|p{}^\dt|\leq 1$.
\WMGineq{} contain the {\em weighted T-systems} (WTS) of~\cite{WTS92}, also known as {\em weighted event graphs} (WEG) in~\cite{March09}
and {\em weighted marked graphs} (WMG), 
in which $\forall p\in P$, $|{}^\dt p| = 1$ and $|p{}^\dt| = 1$.
The nets of Figure~\ref{ex2.fig} are \WMGineq{}.
We denote by \MGineq{} the unit-weighted \WMGineq{}.
Well-studied ordinary subclasses are {\em marked graphs}~\cite{MDG71}, also known as {\em T-nets}~\cite{DesEsp},
which fulfill $|p^\dt|=1$ and $|{}^\dt p|=1$ for each place $p$.\\ 

\noindent {$-$ Subclasses with shared places.}
$N$ is {\em asymmetric-choice} (AC) if it satisfies 
the following condition for any two input places $p_1$, $p_2$
of each synchronization~$t$,
$p_1^\bullet \subseteq p_2^\bullet$
or
$p_2^\bullet \subseteq p_1^\bullet$.
It is {\em free-choice} (FC) if for any two input places $p_1$, $p_2$ of each synchronization $t$,
$p_1^\bullet = p_2^\bullet$.
Thus, FC nets form a subclass of AC nets.
It is a {\em state machine} if it is ordinary 
and each transition has exactly one input and one output.\\

\begin{figure}[!ht]
 
\begin{tabular}{cccc}
 
\begin{minipage}{0.189\linewidth}

\centering

\begin{tikzpicture}[mypetristyle,scale=1]

\node (p1) at (0,0) [place,thick] {};
\node (p2) at (1.25,0) [place,thick] {};

\node [anchor=east] at (p1.west) {$p_1$};
\node [anchor=west] at (p2.east) {$p_2$};

\node (t1) at (0,1.2) [transition,thick] {};
\node (t2) at (1.25,1.2) [transition,thick] {};

\node [anchor=east] at (t1.west) {$t_0$};
\node [anchor=west] at (t2.east) {$t_1$};

\draw [->,thick] (p1) to node [left] {$2$} (t1);
\draw [->,thick] (p1) to node [below,near start] {$2$} (t2.south west);
\draw [->,thick] (p2) to node [below,near start] {$3$} (t1.south east);
\draw [->,thick] (p2) to node [right] {$3$} (t2);

\end{tikzpicture}

\end{minipage}
&
\begin{minipage}{0.3\linewidth}

\centering

\begin{tikzpicture}[mypetristyle,scale=1]

\node (p1) at (0.5,0) [place,thick] {};
\node (p2) at (2,0) [place,thick] {};

\node [anchor=east] at (p1.west) {$p_1$};
\node [anchor=west] at (p2.east) {$p_2$};

\node (t0) at (0,1.2) [transition,thick] {};
\node (t1) at (1,1.2) [transition,thick] {};
\node (t2) at (2,1.2) [transition,thick] {};
\node (t3) at (3,1.2) [transition,thick] {};

\node [anchor=west] at (t1.east) {$t_1$};
\node [anchor=west] at (t2.east) {$t_2$};
\node [anchor=east] at (t0.west) {$t_0$};
\node [anchor=west] at (t3.east) {$t_3$};

\draw [->,thick] (p1) to node [below right, near start] {$2$} (t1);
\draw [->,thick] (p2) to node [above, near start] {$3$} (t1.south east);
\draw [->,thick] (p2) to node [right] {$3$} (t2);
\draw [->,thick] (p2) to node [below right, near end] {$3$} (t3);

\draw [->,thick] (p1) to node [left] {$2$} (t0);
\draw [->,thick] (p2) to node [below,near start] {$3$} (t0.east);

\end{tikzpicture}

\end{minipage}
&
\begin{minipage}{0.225\linewidth}

\centering

\begin{tikzpicture}[mypetristyle,scale=1]

\node (p1) at (0,0) [place,thick] {};
\node (p2) at (2.5,0) [place,thick] {};

\node [anchor=west] at (p1.east) {$p_1$};
\node [anchor=east] at (p2.west) {$p_2$};

\node (t0) at (0,1.2) [transition,thick] {};
\node (t1) at (1.25,1.2) [transition,thick] {};
\node (t2) at (2.5,1.2) [transition,thick] {};

\node [anchor=north] at (t1.south) {$t_1$};
\node [anchor=east] at (t2.west) {$t_2$};
\node [anchor=west] at (t0.east) {$t_0$};

\draw [->,thick] (p1) to node [above, near start] {$2$} (t1);
\draw [->,thick] (p2) to node [above, near start] {$3$} (t1);
\draw [->,thick] (p2) to node [below right, very near end] {$3$} (t2);
\draw [->,thick] (p1) to node [below left, very near end] {$2$} (t0);

\end{tikzpicture}

\end{minipage}
&
\begin{minipage}{0.165\linewidth}

\centering

\begin{tikzpicture}[mypetristyle,scale=1]

\node (p) at (0.5,0) [place,thick] {};

\node [anchor=east] at (p.west) {$p$};

\node (t0) at (-0.2,1.2) [transition,thick] {};
\node (t1) at (1.2,1.2) [transition,thick] {};

\node [anchor=west] at (t0.east) {$t_0$};
\node [anchor=east] at (t1.west) {$t_1$};

\draw [->,thick,bend right=15] (p) to node [below left, near end] {} (t0);
\draw [->,thick,bend left=15] (p) to node [below right, near end] {} (t1);
\draw [->,thick,bend right=15] (t0) to node [below left, near end] {} (p);
\draw [->,thick,bend left=15] (t1) to node [below right, near end] {} (p);

\end{tikzpicture}

\end{minipage}
\end{tabular}

\vspace*{4mm}

\caption{The net on the left is HFC, the second one is HAC.
The third net is homogeneous,
non-AC since $\mathit{\lbul t_1} = \{p_1,p_2\}$,
while 
$\mathit{p_1\lbul} \not\subseteq \mathit{p_2\lbul}$
and
$\mathit{p_2\lbul} \not\subseteq \mathit{p_1\lbul}$.
The fourth net is a state machine.
None of these nets is CF.
}

\label{ExSubclasses2}

\end{figure}
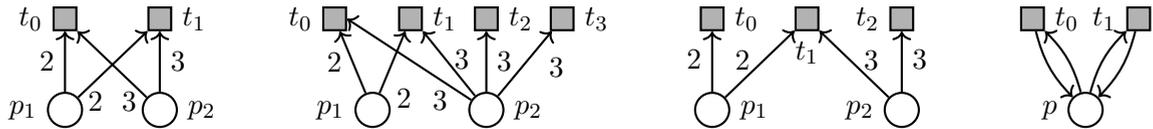

\subsection{S-\WMGineq{}}


We introduce the new classes of $k$S-\WMGineq{} and S-\WMGineq{}.

\begin{definition}[$k$S-\WMGineq{} and S-\WMGineq{}]
A net (or system) is a $k$S-\WMGineq{}  
if it has at most $k$ shared places and if the deletion of all its shared places yields a \WMGineq{}.
A net (or system) is a S-\WMGineq{} if it is a $k$S-\WMGineq{} for some positive integer~$k$.
\end{definition}

Figure~\ref{OneChoiceWMG} pictures an homogeneous $1$S-\WMGineq{} (H$1$S-\WMGineq{}) on the left.

\begin{figure}[!ht]
\centering
\begin{tikzpicture}[scale=0.75,mypetristyle]

\node (p) at (2.5,1.5) [place,tokens=3] {};
\node [anchor=north east] at (p.south west) {$p$};

\node (p1) at (1,3) [place] {};
\node [anchor=south east] at (p1.north west) {$p_1$};

\node (p2) at (4,3) [place,tokens=1] {};
\node [anchor=south] at (p2.north) {$p_2$};

\node (p3) at (5.5,1.5) [place,tokens=2] {};
\node [anchor=east] at (p3.west) {$p_3$};

\node (p4) at (1,0) [place,tokens=1] {};
\node [anchor=north east] at (p4.south west) {$p_4$};

\node (p5) at (4,0) [place] {};
\node [anchor=north] at (p5.south) {$p_5$};

\node (p6) at (6.2,1.5) [place] {};
\node [anchor=west] at (p6.east) {$p_6$};

\node (t1) at (2.5,3) [transition,thick] {};
\node (t2) at (5.5,3) [transition,thick] {};
\node (t3) at (1,1.5) [transition,thick] {};
\node (t4) at (2.5,0) [transition,thick] {};
\node (t5) at (5.5,0) [transition,thick] {};

\node [anchor=east] at (t3.west) {$t_3$};
\node [anchor=south] at (t1.north) {$t_1$};
\node [anchor=north west] at (t5.south east) {$t_5$};
\node [anchor=north] at (t4.south) {$t_4$};
\node [anchor=south west] at (t2.north east) {$t_2$};

\draw [->,thick] (t2.south east) to node [right] {$2$} (p6);
\draw [->,thick] (p6) to node [right] {$5$} (t5.north east);
\draw [->,thick] (t1) to node [left] {} (p);
\draw [->,thick] (p) to node [right] {$2$} (t4);
\draw [->,thick] (p1) to node [left] {$2$} (t3);
\draw [->,thick] (t3) to node [left] {$3$} (p4);
\draw [->,thick] (p4) to node [] {} (t4);
\draw [->,thick] (t4) to node [] {} (p5);
\draw [->,thick] (p5) to node [below] {$2$} (t5);
\draw [->,thick] (t5) to node [left] {$5$} (p3);
\draw [->,thick] (p3) to node [left] {$2$} (t2);
\draw [->,thick, bend right=20] (p) to node [above right, near start] {$2$} (t5.north west);
\draw [->,thick, bend left=20] (p) to node [below right, near start] {$2$} (t2.south west);
\draw [->,thick] (t2) to node [] {} (p2);
\draw [->,thick] (p2) to node [above] {$3$} (t1);
\draw [->,thick] (t1) to node [above] {$5$} (p1);
\draw [->,thick] (t3) to node [above] {$4$} (p);
\end{tikzpicture}
\hspace*{5mm}
\begin{tikzpicture}[scale=0.75,mypetristyle]


\node (p1) at (1,3) [place] {};
\node [anchor=south east] at (p1.north west) {$p_1$};

\node (p2) at (4,3) [place,tokens=1] {};
\node [anchor=south] at (p2.north) {$p_2$};

\node (p3) at (5.5,1.5) [place,tokens=2] {};
\node [anchor=east] at (p3.west) {$p_3$};

\node (p4) at (1,0) [place,tokens=1] {};
\node [anchor=north east] at (p4.south west) {$p_4$};

\node (p5) at (4,0) [place] {};
\node [anchor=north] at (p5.south) {$p_5$};

\node (p6) at (6.2,1.5) [place] {};
\node [anchor=west] at (p6.east) {$p_6$};

\node (t1) at (2.5,3) [transition,thick] {};
\node (t2) at (5.5,3) [transition,thick] {};
\node (t3) at (1,1.5) [transition,thick] {};
\node (t4) at (2.5,0) [transition,thick] {};
\node (t5) at (5.5,0) [transition,thick] {};

\node [anchor=east] at (t3.west) {$t_3$};
\node [anchor=south] at (t1.north) {$t_1$};
\node [anchor=north west] at (t5.south east) {$t_5$};
\node [anchor=north] at (t4.south) {$t_4$};
\node [anchor=south west] at (t2.north east) {$t_2$};

\draw [->,thick] (t2.south east) to node [right] {$2$} (p6);
\draw [->,thick] (p6) to node [right] {$5$} (t5.north east);
%
%
\draw [->,thick] (p1) to node [left] {$2$} (t3);
\draw [->,thick] (t3) to node [left] {$3$} (p4);
\draw [->,thick] (p4) to node [] {} (t4);
\draw [->,thick] (t4) to node [] {} (p5);
\draw [->,thick] (p5) to node [below] {$2$} (t5);
\draw [->,thick] (t5) to node [left] {$5$} (p3);
\draw [->,thick] (p3) to node [left] {$2$} (t2);
%
\draw [->,thick] (t2) to node [] {} (p2);
\draw [->,thick] (p2) to node [above] {$3$} (t1);
\draw [->,thick] (t1) to node [above] {$5$} (p1);
%
\end{tikzpicture}

\vspace*{4mm}

\caption{
Deleting place $p$ in the H$1$S-\WMGineq{} on the left yields the \WMGineq{} on the right.
}
\label{OneChoiceWMG}
\end{figure}
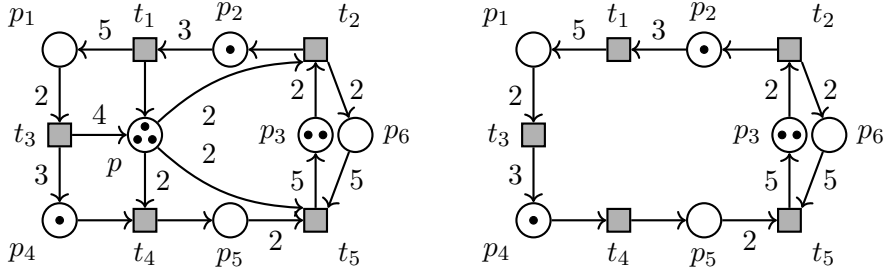

Notice that each Petri net system with $k$ places can be transformed into an S-\WMGineq{} by inserting, 
for each place $p$, a new transition $t_p$ with only input $p$ and only output $p$.
This transformation preserves numerous behavioral properties of the original system 
(e.g.\ liveness, boundedness, reversibility, the reachable markings),
hence there is no hope of reducing the checking complexity of such properties in S-\WMGineq{}.
However, as we will highlight, several intractable problems can be alleviated when $k=1$ (in the H$1$S-\WMGineq{} class).
We will show that the same methods do not work anymore when $k=2$.

\subsection{Augmented Marked Graphs (AMG)}

\vspace*{3mm}

Augmented Marked Graphs (AMG) extend marked graphs with shared places having the same number of inputs and outputs,
in addition to other constraints such as the existence of 
elementary paths connecting outputs to inputs in the underlying marked graph 
and 
restrictions on the initial marking. They form a proper subclass of S-\WMGineq{}.
We recall their most general definition next, as introduced in~\cite{ChuXie97}:

\begin{definition}[AMG \cite{ChuXie97}]
An augmented marked graph is a ordinary Petri net system composed of two distinct sets $P$ and $R$ of places 
($R$ for {\em resources}, denoting the set of shared places; $P$ denoting the other places)
and a set $T$ of transitions satisfying the following conditions:\\
$-$ (H1) The net $G$ obtained by removing the places of $R$ is a marked graph;\\
$-$ (H2) For each place $r \in R$, there exist an integer $k \ge 2$ and $k$ pairs of transitions 
described by the set
$D^r = \{(a_{r_1}, b_{r_1}), \ldots, (a_{r_k}, b_{r_k})\}$
such that
$r^\bullet = \{a_{r_1}, \ldots, a_{r_k}\}$,
$^\bullet r = \{b_{r_1}, \ldots, b_{r_k}\}$,
$a_{r_i} \neq a_{r_j}$,
$b_{r_i} \neq b_{r_j}$, $\forall i \neq j \in \{1, \ldots, k\}$,
and
for each pair $(a_{r_i}, b_{r_i}) \in D^r$ such that $a_{r_i} \neq b_{r_i}$,
there exists an elementary path 
in $G$ from $a_{r_i}$ to $b_{r_i}$;\\
$-$ (H3) Each elementary circuit in $G$ is marked by $M_0$;\\
$-$ (H4) Each place in $R$ is marked by $M_0$, and for each pair $(a_{r_i}, b_{r_i})$ in $D^r$ such that $a_{r_i} \neq b_{r_i}$,
there is an elementary path $O_{r_i}$ in $G$ 
from $a_{r_i}$ to $b_{r_i}$ that is unmarked by $M_0$.
%
%
\end{definition}

This definition allows the existence of marked elementary paths in $G$ from $a_{r_i}$ to $b_{r_i}$ for any $r_i$,
as well as the existence of other sets of pairs $D'^r$ 
containing a pair $(a_{r_i}', b_{r_i}')$ such that no unmarked elementary path exists in $G$ from $a_{r_i}'$ to $b_{r_i}'$.
Examples are provided in Figure~\ref{ExAMG} with only one resource place (i.e. a shared place),
but an arbitrary number of resource places is allowed in general (contrarily to H$1$S-\WMGineq{}).
More restricted definitions of AMG exist,
adding notably the constraint of $1$-boundedness (also called $1$-safeness, or safeness)~\cite{huang2003property}.\\

\begin{figure}[!ht]
 
\begin{tabular}{ccc}
 
\begin{minipage}{0.3\linewidth}

\centering

\begin{tikzpicture}[mypetristyle,scale=1]

\node (p1) at (3,3) [place,thick] {};
\node (p2) at (3,2) [place,thick,tokens=3] {};
\node (p3) at (1,1) [place,thick] {};
\node (p4) at (2,1) [place,thick,tokens=1] {};
\node (p5) at (3,1) [place,thick,tokens=2] {};
\node (p6) at (4,1) [place,thick] {};

\node [anchor=south] at (p1.north) {$p_1$};
\node [anchor=south] at (p2.north) {$p_2$};
\node [anchor=south] at (p3.north) {$p_3$};
\node [anchor=east] at (p4.west) {$p_4$};
\node [anchor=south] at (p5.north) {$p_5$};
\node [anchor=west] at (p6.east) {$p_6$};

\node (t1) at (2,2) [transition,thick] {};
\node (t2) at (4,2) [transition,thick] {};
\node (t3) at (2,0) [transition,thick] {};
\node (t4) at (4,0) [transition,thick] {};

\node [anchor=south] at (t1.north) {$t_1$};
\node [anchor=south] at (t2.north) {$t_2$};
\node [anchor=west] at (t3.east) {$t_3$};
\node [anchor=east] at (t4.west) {$t_4$};

\draw [->,thick] (p1) to node [above left] {} (t2);
\draw [->,thick] (t2) to node [above right] {} (p2);
\draw [->,thick] (p2) to node [above right] {} (t1);
\draw [->,thick] (t1) to node [above right] {} (p1);
\draw [->,thick] (p5) to node [above left] {} (t1);
\draw [->,thick] (p5) to node [below left] {} (t2);
\draw [->,thick] (t3) to node [below right] {} (p5);
\draw [->,thick] (t4) to node [below left] {} (p5);
\draw [->,thick] (t1) to node [above right] {} (p3);
\draw [->,thick] (t1) to node [above right] {} (p4);
\draw [->,thick] (p3) to node [above right] {} (t3);
\draw [->,thick] (p4) to node [above right] {} (t3);
\draw [->,thick] (t2) to node [above right] {} (p6);
\draw [->,thick] (p6) to node [above right] {} (t4);

\end{tikzpicture}

\end{minipage}
&
\begin{minipage}{0.3\linewidth}

\centering

\begin{tikzpicture}[mypetristyle,scale=1]

\node (p1) at (3,3) [place,thick] {};
\node (p2) at (3,2) [place,thick,tokens=3] {};
\node (p3) at (1,1) [place,thick] {};
\node (p4) at (2,1) [place,thick,tokens=1] {};
\node (p6) at (4,1) [place,thick] {};

\node [anchor=south] at (p1.north) {$p_1$};
\node [anchor=south] at (p2.north) {$p_2$};
\node [anchor=south] at (p3.north) {$p_3$};
\node [anchor=east] at (p4.west) {$p_4$};
\node [anchor=west] at (p6.east) {$p_6$};

\node (t1) at (2,2) [transition,thick] {};
\node (t2) at (4,2) [transition,thick] {};
\node (t3) at (2,0) [transition,thick] {};
\node (t4) at (4,0) [transition,thick] {};

\node [anchor=south] at (t1.north) {$t_1$};
\node [anchor=south] at (t2.north) {$t_2$};
\node [anchor=west] at (t3.east) {$t_3$};
\node [anchor=east] at (t4.west) {$t_4$};

\draw [->,thick] (p1) to node [above left] {} (t2);
\draw [->,thick] (t2) to node [above right] {} (p2);
\draw [->,thick] (p2) to node [above right] {} (t1);
\draw [->,thick] (t1) to node [above right] {} (p1);
%
%
\draw [->,thick] (t1) to node [above right] {} (p3);
\draw [->,thick] (t1) to node [above right] {} (p4);
\draw [->,thick] (p3) to node [above right] {} (t3);
\draw [->,thick] (p4) to node [above right] {} (t3);
\draw [->,thick] (t2) to node [above right] {} (p6);
\draw [->,thick] (p6) to node [above right] {} (t4);

\end{tikzpicture}

\end{minipage}
&
\begin{minipage}{0.3\linewidth}

\centering

\begin{tikzpicture}[mypetristyle,scale=1]

\node (p1) at (3,3) [place,thick] {};
\node (p2) at (3,2) [place,thick,tokens=3] {};
\node (p3) at (1,1) [place,thick,tokens=1] {};
\node (p4) at (2,1) [place,thick,tokens=1] {};
\node (p5) at (3,1) [place,thick,tokens=2] {};
\node (p6) at (4,1) [place,thick] {};

\node [anchor=south] at (p1.north) {$p_1$};
\node [anchor=south] at (p2.north) {$p_2$};
\node [anchor=south] at (p3.north) {$p_3$};
\node [anchor=east] at (p4.west) {$p_4$};
\node [anchor=south] at (p5.north) {$p_5$};
\node [anchor=west] at (p6.east) {$p_6$};

\node (t1) at (2,2) [transition,thick] {};
\node (t2) at (4,2) [transition,thick] {};
\node (t3) at (2,0) [transition,thick] {};
\node (t4) at (4,0) [transition,thick] {};

\node [anchor=south] at (t1.north) {$t_1$};
\node [anchor=south] at (t2.north) {$t_2$};
\node [anchor=west] at (t3.east) {$t_3$};
\node [anchor=east] at (t4.west) {$t_4$};

\draw [->,thick] (p1) to node [above left] {} (t2);
\draw [->,thick] (t2) to node [above right] {} (p2);
\draw [->,thick] (p2) to node [above right] {} (t1);
\draw [->,thick] (t1) to node [above right] {} (p1);
\draw [->,thick] (p5) to node [above left] {} (t1);
\draw [->,thick] (p5) to node [below left] {} (t2);
\draw [->,thick] (t3) to node [below right] {} (p5);
\draw [->,thick] (t4) to node [below left] {} (p5);
\draw [->,thick] (t1) to node [above right] {} (p3);
\draw [->,thick] (t1) to node [above right] {} (p4);
\draw [->,thick] (p3) to node [above right] {} (t3);
\draw [->,thick] (p4) to node [above right] {} (t3);
\draw [->,thick] (t2) to node [above right] {} (p6);
\draw [->,thick] (p6) to node [above right] {} (t4);

\end{tikzpicture}

\end{minipage}

\end{tabular}

\vspace*{4mm}

\caption{The system on the left is an AMG, with set of resource places $R=\{p_5\}$.
Let us choose $D^r = \{(t_1,t_3),(t_2,t_4)\}$,
$O_{r_1}$ the unmarked elementary path $t_1 p_3 t_3$
and
$O_{r_2}$ the unmarked elementary path $t_2 p_6 t_4$.
Notice that choosing $D'^r = \{(t_1,t_4),(t_2,t_3)\}$ does not permit to find an unmarked elementary path 
from~$t_2$ to~$t_3$ in $G$ since each path must visit $p_2$, which is marked.
Its underlying marked graph $G$ is pictured in the middle.
The system on the right is not an AMG. None of these systems is $1$-bounded.
}

\label{ExAMG}
\end{figure}
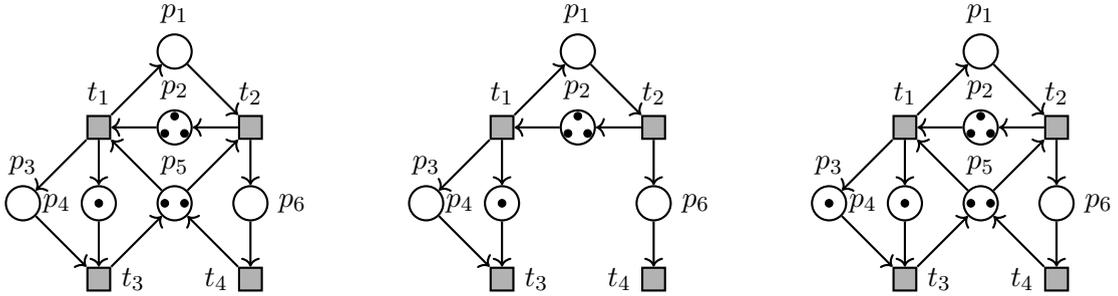

\subsection{Place-Composed Marked Graphs with relaxed place constraints (\PCMGineq{})}

\vspace*{2mm}

We introduce the new class of \PCMGineq{} and a dedicated notion of {\em well-structuredness}.
For that purpose, we need to define a place-merging operation on Petri nets.

\begin{definition}[Place-merging]
Consider a net $N=(P,T,W)$ and a subset $A$ of $2^P$ whose elements are mutually disjoint.
The net $N'=(P',T',W')$ obtained by place-merging $A$ is defined as follows:\\
$-$ $T' = T$;\\
$-$ for each element $x=\{p_1, \ldots, p_k\}$ in $A$, a place $p_x$ belongs to $P'$ such that,
for each transition $t$, $W'(p_x,t) = \sum_{p \in x} W(p,t)$ and $W'(t,p_x) = \sum_{p \in x} W(t,p)$;\\
$-$ denoting by $P''$ the set of places in $P$ that do not appear in $A$, $P' = \bigcup_{x \in A} p_x \cup P''$,
and for each place $p$ in $P''$ and each transition $t$, 
$W'(p,t) = W(p,t)$ and $W'(t,p) = W(t,p)$.\\
\end{definition}


\begin{definition}[\PCMGineq{}]
%
%
Consider any connected, undirected graph $G=(V,E)$,
where $V=\{v_1, \ldots, v_x\}$ is a finite set of vertices
and $E=\{e_1, \ldots, e_y\}$ is a finite set of edges connecting distinct vertices of $V$.
A place-composed marked graph with relaxed place constraints (\PCMGineq{}) $N=(P,T,W)$ 
is obtained from~$G$ by refining vertices with places and edges with \MGineq{} as follows:
%
%
\begin{itemize}
\item first, define a \MGineq{} $N'=(P',T',W')$ containing $y$ maximal connected components $C_1, \ldots, C_y$,
each of which contains at least two places;
denote by $C$ the set of these components;

\item then, define a bijective mapping $\beta: E \mapsto C$ that associates to each edge a component,
and a mapping $\gamma: E \mapsto P' \times P'$ that associates to each edge $e_i=\{v_a,v_b\}$, $a<b$, 
a pair of distinct places $(p_{i,a}, p_{i,b})$ in the component $\beta(e_i)$;

\item finally, for each vertex $v_j$ of $V$, denote by $A_j$ the set of all the places associated to $v_j$ through~$\gamma$;
denote by $A$ the set $\{A_j \mid v_j \in V\}$;

\item $N$ is obtained by place-merging $A$.

\end{itemize}
\end{definition}

Examples are pictured in Figure \ref{Fig-PCMG}.

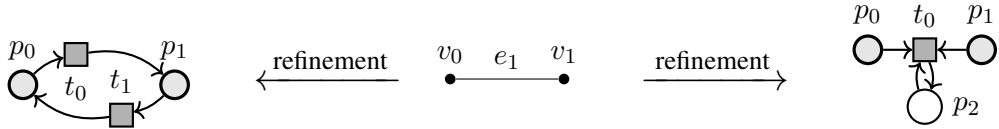
\begin{figure}[!ht]
\centering
\begin{minipage}{0.25\linewidth}
\centering
\raisebox{1mm}{
\begin{tikzpicture}[mypetristyle]

\node (p0) at (0,0) [petriNode] {};
\node (p1) at (2,0) [petriNode] {};

\node [anchor=south] at (p0.north) {$p_0$};
\node [anchor=south] at (p1.north) {$p_1$};

\node (t0) at (0.7,0.4) [transition,thick] {};
\node (t1) at (1.3,-0.4) [transition,thick] {};

\node [anchor=north] at (t0.south) {$t_0$};
\node [anchor=south] at (t1.north) {$t_1$};

\draw [->,thick, bend left=15] (p0) to node [] {} (t0);
\draw [->,thick, bend left=25] (t0) to node [] {} (p1);
\draw [->,thick, bend left=15] (p1) to node [] {} (t1);
\draw [->,thick, bend left=25] (t1) to node [] {} (p0);

\end{tikzpicture}
}
\end{minipage}
{\Large $\xleftarrow{\textrm{~refinement~}}{}{}$}
\begin{minipage}{0.175\linewidth}
\raisebox{6mm}{
\begin{tikzpicture}[scale=1,mypetristyle]
\node[ltsNode](n0)at(0,0){};
\node[ltsNode](n1)at(1.5,0){};
\draw[-](n0)to node[auto,swap,above]{$e_1$}(n1);
\node [anchor=south] at (n0.north) {$v_0$};
\node [anchor=south] at (n1.north) {$v_1$};
\end{tikzpicture}
}
\end{minipage}
{\Large $\,\xrightarrow{\textrm{~refinement~}}{}{}$}
\begin{minipage}{0.2\linewidth}
\centering
\raisebox{6mm}{
\begin{tikzpicture}[mypetristyle]

\node (p0) at (0,0) [petriNode] {};
\node (p1) at (1.5,0) [petriNode] {};
\node (p2) at (0.75,-0.75) [place] {};

\node [anchor=south] at (p0.north) {$p_0$};
\node [anchor=south] at (p1.north) {$p_1$};
\node [anchor=west] at (p2.east) {$p_2$};

\node (t0) at (0.75,0) [transition,thick] {};

\node [anchor=south] at (t0.north) {$t_0$};

\draw [->,thick, bend left=0] (p0) to node [] {} (t0);
\draw [->,thick, bend left=0] (p1) to node [] {} (t0);
\draw [->,thick, bend left=20] (t0) to node [] {} (p2);
\draw [->,thick, bend left=20] (p2) to node [] {} (t0);

\end{tikzpicture}
}
\end{minipage}



\caption{
In the middle, a simple undirected graph $G=(V,E)$ is pictured, where $V=\{v_0,v_1\}$ and $E=~\{e_1\}$.
The \MGineq{} on the left is obtained from $G$ by identifying place $p_0$ to vertex $v_0$ and place $p_1$ to vertex $v_1$,
i.e. $\beta(e_1) = C_1$, $\gamma(e_1) = (p_0,p_1)$.
On the right, $e_1$ is replaced by a component $C_1$ formed of a single transition~$t_0$
synchronizing its input places $p_0$, $p_1$ and $p_2$ and writing in $p_2$.
}

\label{Fig-PCMG}

\end{figure}

In the design phase, \PCMGineq{} allow to define first the topology of shared places,
i.e. the communication links between the buffers of the system.
The processes reading and writing the buffers can then be defined through refinement.\\

In the sequel, we focus mainly on well-formed MG subnets, i.e. MG subnets that are structurally live (meaning that a live marking exists)
and structurally bounded (i.e. for each initial marking, the system is bounded).
We define next a notion of well-structuredness for \PCMGineq{}.

\begin{definition}[Well-structured \PCMGineq{}]
A \PCMGineq{} $S=(N,M_0)$, where $N=(P,T,W)$, obtained from a graph $G=(V,E)$, is well-structured if 
each component of $C$ is a strongly connected and well-formed MG.
\end{definition}

The Petri net obtained on the left of Figure \ref{Fig-PCMG} is a well-structured \PCMGineq{},
while the net on the right is not, since the unique MG component is not structurally live.
Another well-structured \PCMGineq{}, with shared places, is given in the middle of Figure~\ref{AMGvsPCMG}.

In the definition of \PCMGineq{}, the undirected graph representing the topology stems from the fundamental behavioral properties
fulfilled by this class. 
Indeed, we provide in Section~\ref{SecPCMG} a characterization of reversibility,
together with a sufficient condition of PR-R equality,
and show they are no more valid when the underlying undirected graph topology of \PCMGineq{} is relaxed.
Thus, \PCMGineq{} are defined so as to benefit from stronger conditions ensuring liveness, boundedness, reversibility
and to fulfill the PR-R equality, while extending the expressiveness of marked graphs and state machines.

\subsection{Expressiveness comparison}

\vspace*{3mm}


The AMG class contains all the MG but not all WMG (nor the \WMGineq{}), since the former are unit-weighted and the latter have arbitrary weights.
Since AMG allow shared places, \WMGineq{} do not contain them all either.
AMG and \PCMGineq{} are incomparable (see Figure~\ref{AMGvsPCMG}).
H$1$S-\WMGineq{} are not included in \PCMGineq{} nor in AMG.

\begin{figure}[!ht]
\centering
\begin{minipage}{0.2\linewidth}
\centering
\raisebox{4mm}{
\begin{tikzpicture}[scale=0.8,mypetristyle]
\node[ltsNode](n0)at(1.5,2){};
\node[ltsNode](n1)at(0,0){};
\node[ltsNode](n2)at(3,0){};
\draw[-](n0)to node[auto,swap]{}(n1);
\draw[-](n1)to node[auto,swap]{}(n2);
\draw[-](n2)to node[auto,swap]{}(n0);
\node [anchor=south] at (n0.north) {$v_{p_0}$};
\node [anchor=north] at (n1.south) {$v_{p_1}$};
\node [anchor=north] at (n2.south) {$v_{p_2}$};
\end{tikzpicture}
}
\end{minipage}
\raisebox{4mm}{
$\xrightarrow{\textrm{refinement}}{}{}$
}
\begin{minipage}{0.25\linewidth}
\centering

\begin{tikzpicture}[scale=0.8,mypetristyle]

\node (p0) at (1.5,2) [petriNode,tokens=1] {};
\node (p1) at (0,0) [petriNode,tokens=1] {};
\node (p2) at (3,0) [petriNode,tokens=1] {};

\node [anchor=south] at (p0.north) {$p_0$};
\node [anchor=north] at (p1.south) {$p_1$};
\node [anchor=north] at (p2.south) {$p_2$};

\node (t0) at (0.75,1) [transition,thick] {};
\node (t1) at (1.5,0) [transition,thick] {};
\node (t2) at (2.25,1) [transition,thick] {};
\node (t3) at (0.15,1.45) [transition,thick] {};
\node (t4) at (2.85,1.45) [transition,thick] {};
\node (t5) at (1.5,-0.7) [transition,thick] {};

\node [anchor=west] at (t0.east) {$t_0$};
\node [anchor=south] at (t1.north) {$t_1$};
\node [anchor=east] at (t2.west) {$t_2$};
\node [anchor=south] at (t3.north) {$t_3$};
\node [anchor=south] at (t4.north) {$t_4$};
\node [anchor=north] at (t5.south) {$t_5$};

\draw [->,thick, bend left=0] (p0) to node [] {} (t0);
\draw [->,thick, bend left=0] (t0) to node [] {} (p1);
\draw [->,thick, bend left=0] (p1) to node [] {} (t1);
\draw [->,thick, bend left=0] (t1) to node [] {} (p2);
\draw [->,thick, bend left=0] (p2) to node [] {} (t2);
\draw [->,thick, bend left=0] (t2) to node [] {} (p0);

\draw [->,thick, bend left=25] (p0) to node [] {} (t4);
\draw [->,thick, bend left=25] (t4) to node [] {} (p2);
\draw [->,thick, bend left=25] (p2) to node [] {} (t5);
\draw [->,thick, bend left=25] (t5) to node [] {} (p1);
\draw [->,thick, bend left=25] (p1) to node [] {} (t3);
\draw [->,thick, bend left=25] (t3) to node [] {} (p0);

\end{tikzpicture}

\end{minipage}
\hspace*{10mm}
\begin{minipage}{0.23\linewidth}
\centering
\raisebox{1mm}{
\begin{tikzpicture}[scale=0.8,mypetristyle]

\node (p0) at (1.5,2) [place,tokens=1] {};
\node (p1) at (0,0) [place,tokens=1] {};
\node (p2) at (3,0) [place,tokens=1] {};

\node [anchor=south] at (p0.north) {$p_0$};
\node [anchor=north] at (p1.south) {$p_1$};
\node [anchor=north] at (p2.south) {$p_2$};

\node (t0) at (0.4,1.2) [transition,thick] {};
\node (t1) at (1.5,-0.6) [transition,thick] {};
\node (t2) at (2.6,1.2) [transition,thick] {};

\node [anchor=south] at (t0.north) {$t_0$};
\node [anchor=north] at (t1.south) {$t_1$};
\node [anchor=south] at (t2.north) {$t_2$};

\draw [->,thick, bend left=0] (p0) to node [] {} (t0);
\draw [->,thick, bend left=0] (t0) to node [] {} (p1);
\draw [->,thick, bend left=0] (p1) to node [] {} (t1);
\draw [->,thick, bend right=0] (t1) to node [] {} (p2);
\draw [->,thick, bend left=0] (p2) to node [] {} (t2);
\draw [->,thick, bend left=0] (t2) to node [] {} (p0);

\draw [->,thick, bend left=0] (t0) to node [] {} (p0);
\draw [->,thick, bend left=0] (p1) to node [] {} (t0);
\draw [->,thick, bend left=0] (t1) to node [] {} (p1);
\draw [->,thick, bend left=0] (p2) to node [] {} (t1);
\draw [->,thick, bend left=0] (t2) to node [] {} (p2);
\draw [->,thick, bend left=0] (p0) to node [] {} (t2);

\draw [->,thick, bend left=0] (t0) to node [] {} (p2);
\draw [->,thick, bend left=0] (p2) to node [] {} (t0);

\end{tikzpicture}
}
\end{minipage}



\caption{
The system in the middle is a \PCMGineq{} but not an AMG: in the underlying marked graph $G$, there is no path connecting the transitions.
The system on the right is an AMG but not a \PCMGineq{}: 
$t_0$ connects $3$ shared places, thus the system cannot be obtained 
from an undirected graph as in the definition of \PCMGineq{}.
}

\label{AMGvsPCMG}

\end{figure}
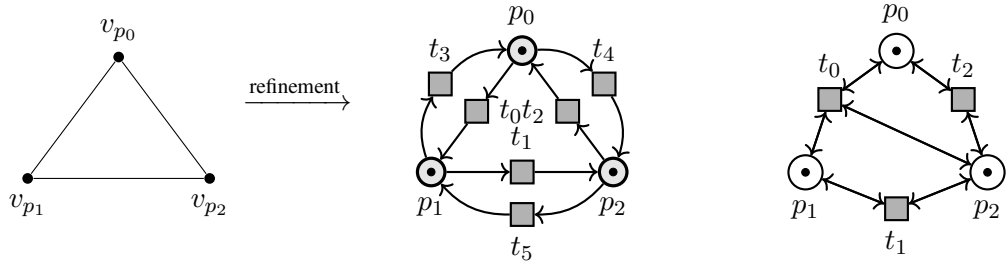

Deterministically synchronized sequential processes (DSSP), introduced in~\cite{DSSP98}, 
aim at modelling several agents that cooperate through asynchronous message passing,
in a modular way, each module representing an agent.
%
%
We do not study DSSP in this work, whose formal definition, together with examples, can be found in~\cite{DSSP98}.
This unit-weighted class does not contain all state machines 
even with a single shared place,
nor all the AMG, the H$1$S-\WMGineq{} and the \PCMGineq{}.
However, it is worth mentioning that each HFC system can be transformed into a DSSP system
with the same set of feasible sequences, as shown in~\cite{DSSP98}.

Examples are pictured in Figure~\ref{FigComparison} and the inclusion relation between the main subclasses
studied in this paper is depicted in Figure~\ref{DAGclasses}.\\

\begin{figure}[!h]
 
\begin{tabular}{cccc}
 
\begin{minipage}[b]{0.32\linewidth}

\centering

\begin{tikzpicture}[mypetristyle,scale=0.7]

\node (p1) at (0,0) [place,thick] {};
\node (p2) at (2,0) [place,thick] {};
\node (p3) at (4,0) [place,thick] {};

\node [anchor=west] at (p1.east) {$p_1$};
\node [anchor=north] at (p2.south) {$p_2$};
\node [anchor=east] at (p3.west) {$p_3$};

\node (t1) at (0,1.2) [transition,thick] {};
\node (t2) at (4,1.2) [transition,thick] {};
\node (t3) at (0,-1.2) [transition,thick] {};
\node (t4) at (4,-1.2) [transition,thick] {};

\node [anchor=south] at (t1.north) {$t_1$};
\node [anchor=south] at (t2.north) {$t_2$};
\node [anchor=north] at (t3.south) {$t_3$};
\node [anchor=north] at (t4.south) {$t_4$};

\draw [->,thick] (p1) to node [left] {} (t1);
\draw [->,thick] (p2) to node [above right] {$2$} (t1);
\draw [->,thick] (p2) to node [above left] {$2$} (t2);
\draw [->,thick] (p3) to node [right] {} (t2);

\draw [->,thick] (p3) to node [right] {} (t2);
\draw [->,thick] (t4) to node [left] {$3$} (p3);
\draw [->,thick] (t3) to node [right] {} (p1);
\draw [->,thick] (t3) to node [right] {} (p2);

\end{tikzpicture}\\
\centering
\footnotesize
H$1$S-\WMGineq{}\\
~
\end{minipage}
~
&
~~
\begin{minipage}[b]{0.12\linewidth}

\centering

\begin{tikzpicture}[mypetristyle,scale=0.7]

\node (p) at (1,1) [place,thick] {};

\node (t1) at (0,2) [transition,thick] {};
\node (t2) at (2,2) [transition,thick] {};
\node (t3) at (0,0) [transition,thick] {};
\node (t4) at (2,0) [transition,thick] {};

\node [anchor=south] at (t1.north) {$t_1$};
\node [anchor=south] at (t2.north) {$t_2$};
\node [anchor=north] at (t3.south) {$t_3$};
\node [anchor=north] at (t4.south) {$t_4$};

\draw [->,thick] (p) to node [below, near end] {$2$} (t1);
\draw [->,thick] (p) to node [below, near end] {$2$} (t2);
\draw [->,thick] (t3) to node [left] {} (p);
\draw [->,thick] (t4) to node [right] {} (p);

\end{tikzpicture}\\
\centering
\footnotesize
~~~HFC\\
~
\end{minipage}
&
\begin{minipage}[b]{0.13\linewidth}
\raisebox{1.7cm}{ 
\begin{tikzpicture}[decoration={snake,amplitude=.3mm,segment length=2mm,post length=1mm}]
\draw[->,thick,decorate] (0,-1cm) to node [above] {\footnotesize simulated by~\,}  ++(1.8,0);
\end{tikzpicture}
}
\end{minipage}
&
\raisebox{3.5mm}{
\begin{minipage}[b]{0.22\linewidth}

\centering

\begin{tikzpicture}[mypetristyle,scale=0.7]

\node (p) at (1,1) [place,thick] {};
\node (p1) at (1,2) [place,thick,tokens=1] {};
\node (p2) at (-1,0) [place,thick,tokens=1] {};
\node (p3) at (3,0) [place,thick,tokens=1] {};

\node [anchor=north] at (p.south) {$p$};
\node [anchor=south] at (p1.north) {$p_1$};
\node [anchor=north] at (p2.south) {$p_2$};
\node [anchor=north] at (p3.south) {$p_3$};

\node (t1) at (0,2) [transition,thick] {};
\node (t2) at (2,2) [transition,thick] {};
\node (t3) at (0,0) [transition,thick] {};
\node (t4) at (2,0) [transition,thick] {};

\node [anchor=south] at (t1.north) {$t_1$};
\node [anchor=south] at (t2.north) {$t_2$};
\node [anchor=north] at (t3.south) {$t_3$};
\node [anchor=north] at (t4.south) {$t_4$};

\draw [->,thick] (p) to node [below, near end] {$2$} (t1);
\draw [->,thick] (p) to node [below, near end] {$2$} (t2);
\draw [->,thick] (t3) to node [left] {} (p);
\draw [->,thick] (t4) to node [right] {} (p);

\draw [->,thick,bend left=0] (p1) to node [left] {} (t1);
\draw [->,thick,bend left=0] (t1) to node [left] {} (p1);

\draw [->,thick,bend left=0] (p1) to node [left] {} (t2);
\draw [->,thick,bend left=0] (t2) to node [left] {} (p1);

\draw [->,thick,bend left=0] (p2) to node [left] {} (t3);
\draw [->,thick,bend left=0] (t3) to node [left] {} (p2);

\draw [->,thick,bend left=0] (p3) to node [left] {} (t4);
\draw [->,thick,bend left=0] (t4) to node [left] {} (p3);

\end{tikzpicture}\\
\centering
\footnotesize
~~~~DSSP
\end{minipage}
}

\end{tabular}

\caption{The net on the left is an H$1$S-\WMGineq{}, the one in the middle is an HFC and simulated by the one on the right, which is a DSSP.}

\label{FigComparison}

\end{figure}
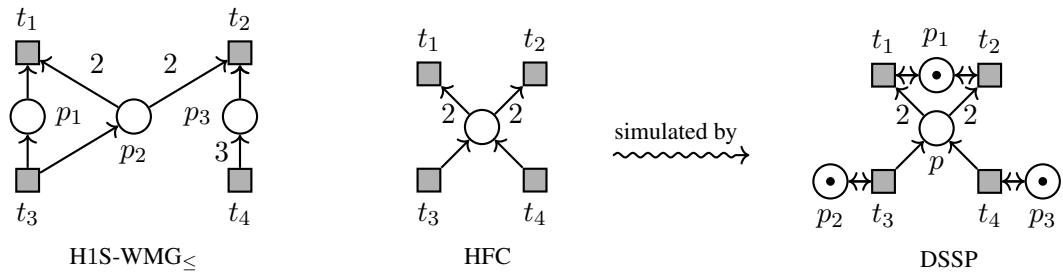

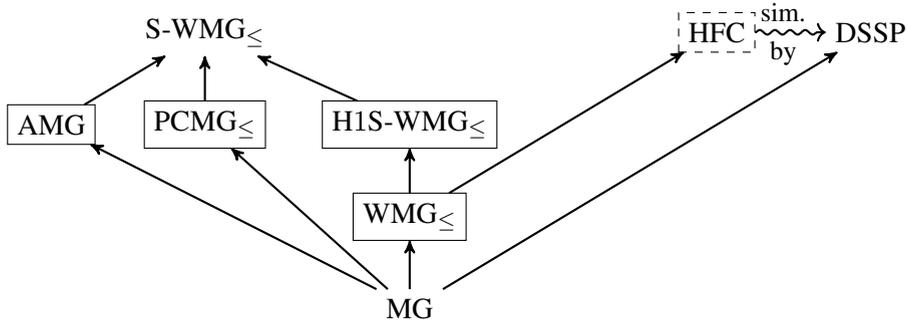
\begin{figure}[htb] 
\centering

\tikzstyle{VertexStyle} = []
\tikzstyle{EdgeStyle}   = [->,>=stealth']      

\begin{tikzpicture}[scale=1.35]
    \Vertex[x=3,y=0]{MG}
    { \tikzstyle{VertexStyle} = [shape = rectangle, draw]  
    \Vertex[x=-0.5,y=1.8]{AMG} 
    \Vertex[x=3,y=0.9,L=\WMGineq{}]{WMG}
    \Vertex[x=3,y=1.8,L=H$1$S-\WMGineq{}]{1-HEWMG}
    \Vertex[x=1,y=1.8,L=\PCMGineq{}]{PCMG}
    }
    { \tikzstyle{VertexStyle} = [shape = rectangle, draw, dashed]  
    \Vertex[x=6,y=2.7]{HFC}     
    }
    \Vertex[x=1,y=2.7,L=S-\WMGineq{}]{EWMG} 
    \Vertex[x=7.5,y=2.7]{DSSP}
    \Edges(MG.north west,AMG,EWMG)
    \Edges(MG,PCMG,EWMG)
    \Edges(MG,WMG,1-HEWMG,EWMG) 
    \Edges(WMG,HFC)
    \Edges(MG.north east,DSSP)
    \tikzstyle{EdgeStyle}  = [->, inner sep=3pt, decorate, decoration={snake,amplitude=.3mm,segment length=2mm,post length=1mm}] 
    \Edges(HFC,DSSP)
    \node[](simby)at(6.65,2.9){\small sim.};
    \node[](simby2)at(6.65,2.5){\small by};
\end{tikzpicture}

\vspace*{3mm}

\caption{
Inclusion of several classes mentioned in this paper.
We focus on the boxed classes, for which we develop new results;
the HFC class is also studied to a smaller extent (dashed box).
Straight arrows represent the inclusion relation: MG are included in AMG, \PCMGineq{}, \WMGineq{} and DSSP,  
\WMGineq{} form a subclass of H$1$S-\WMGineq{} and HFC nets, 
while S-\WMGineq{}, HFC and DSSP nets are incomparable,
and 
AMG, \PCMGineq{} and H$1$S-\WMGineq{} are also incomparable.
%
However,
the wavy arrow represents the possibility of transforming each HFC into a DSSP with the same
set of transitions and feasible sequences. 
We do not depict the transitive closure of the inclusion relation for the sake of readability.
} 
\label{DAGclasses}
\end{figure}

In the following sections, we develop new conditions ensuring the PR-R equality in weighted Petri nets
and some of their subclasses mentioned above, namely H$1$S-\WMGineq{}, AMG and \PCMGineq{};
we also consider the HFC class to a smaller extent.
To achieve it,
we introduce the notion of {\em directedness} of the potential reachability graph, together with variants, in the next section.

Previous works exist that study properties related to the state equation in other classes,
which do not contain our classes or do not tackle the PR-R equality problem,
as summarized in the related work at the end of this paper.

\section{Directedness}\label{PotentialReachDirected}

In this section, we first introduce the notion of {\em directedness} of the potential reachability graph,
with variants, extracted from~\cite{STECS,DSSP98,ArxivSSP20}.
%
%
Then, we present an overview of the classes from the literature that benefit from directedness,
including the persistent class for which a stronger form of directedness exists, embodied by Keller's theorem.

\subsection{Directedness and variants}

\vspace*{3mm}

\begin{definition}[Directedness of the potential reachability graph]
Let us consider any system $S = (N,M_0)$ and its potential reachability graph $PRG(S)$:\\
$-$ $PRG(S)$ is \emph{directed} if every two potentially reachable markings have a common reachable marking.
More formally: $\forall M_1, M_2 \in PR(S)$: $R((N,M_1)) \cap R((N,M_2)) \neq \emptyset$.\\
$-$ $PRG(S)$ is \emph{initially directed} if $\forall M_1 \in PR(S): R(S) \cap R((N,M_1)) \neq \emptyset$.
%
\end{definition}

The directedness of $PRG(S)$ is called {\em structural directedness} in~\cite{LenderProc2006}.

We shall also consider the particular case of directedness restricted to the reachability graph,
i.e. when every two reachable markings have a common reachable marking.

Figure~\ref{FigDefDirectedness} illustrates these properties.

\begin{figure}[!ht]
 \centering

\begin{tabular}{ccc}
 
\begin{minipage}{0.3\linewidth}

\centering

\begin{tikzpicture}[scale=0.75]
\node[ltsNode,label=below:$M_0$](n0)at(1,0){};
\node[ltsNode,label=left:$M_1$](n1)at(0,1){};
\node[ltsNode,label=right:$M_2$](n2)at(2,1){};
\node[ltsNode,label=above:$M$](n3)at(1,2){};
\draw[-{>[scale=2.5,length=2,width=2]},dashed](n0)to node[auto,swap]{}(n1);
\draw[-{>[scale=2.5,length=2,width=2]},dashed](n0)to node[auto,swap]{}(n2);
\draw[-{>[scale=2.5,length=2,width=2]}](n1)to node[auto,swap]{}(n3);
\draw[-{>[scale=2.5,length=2,width=2]}](n2)to node[auto,swap]{}(n3);
\end{tikzpicture}\\
\centering
\footnotesize
Directedness of $PRG(S)$
\end{minipage}
&
%
\begin{minipage}{0.3\linewidth}
 \centering

\begin{tikzpicture}[scale=0.75]
\node[ltsNode,label=below:$M_0$](n0)at(1,0){};
\node[ltsNode,label=left:$M_1$](n1)at(0,1){};
\node[ltsNode,label=right:$M_2$](n2)at(2,1){};
\node[ltsNode,label=above:$M$](n3)at(1,2){};
\draw[-{>[scale=2.5,length=2,width=2]}](n0)to node[auto,swap]{}(n1);
\draw[-{>[scale=2.5,length=2,width=2]}](n0)to node[auto,swap]{}(n2);
\draw[-{>[scale=2.5,length=2,width=2]}](n1)to node[auto,swap]{}(n3);
\draw[-{>[scale=2.5,length=2,width=2]}](n2)to node[auto,swap]{}(n3);
\end{tikzpicture}\\
\centering
\footnotesize
Directedness of $RG(S)$
\end{minipage}
&
\begin{minipage}{0.3\linewidth}

\begin{tikzpicture}[scale=0.75]
\node[ltsNode,label=below:$M_0$](n0)at(1,0){};
\node[ltsNode,label=left:$M_1$](n1)at(0,1){};
\node[ltsNode,label=above:$M$](n3)at(1,2){};
\draw[-{>[scale=2.5,length=2,width=2]},dashed](n0)to node[auto,swap]{}(n1);
\draw[-{>[scale=2.5,length=2,width=2]}](n0)to node[auto,swap]{}(n3);
\draw[-{>[scale=2.5,length=2,width=2]}](n1)to node[auto,swap]{}(n3);
\end{tikzpicture}\\
\centering
\footnotesize
Initial directedness of $PRG(S)$
\end{minipage}
\end{tabular}

\vspace*{3mm}

\caption{
Variants of directedness.
}
\label{FigDefDirectedness}
\end{figure}
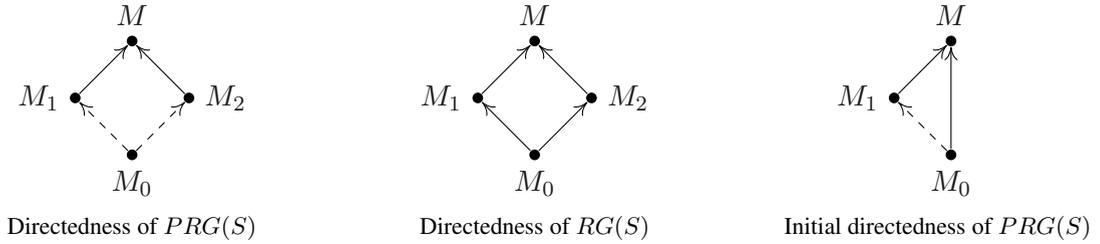

Notice that initial directedness does not imply directedness, as examplified by Figure~\ref{InitDirNotDir}.\\

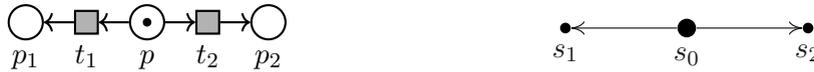
\begin{figure}[!ht]
\centering
\vspace*{1cm}
\begin{minipage}{0.45\linewidth}
\centering
\raisebox{0mm}{
\begin{tikzpicture}[scale=0.8,mypetristyle]

\node (p) at (0,0) [place,tokens=1] {};
\node (p1) at (-2,0) [place] {};
\node (p2) at (2,0) [place] {};

\node [anchor=north] at (p.south) {$p$};
\node [anchor=north] at (p1.south) {$p_1$};
\node [anchor=north] at (p2.south) {$p_2$};

\node (t1) at (-1,0) [transition,thick] {};
\node (t2) at (1,0) [transition,thick] {};

\node [anchor=north] at (t1.south) {$t_1$};
\node [anchor=north] at (t2.south) {$t_2$};

\draw [->,thick, bend left=0] (p) to node [] {} (t1);
\draw [->,thick, bend left=0] (t1) to node [] {} (p1);
\draw [->,thick, bend left=0] (p) to node [] {} (t2);
\draw [->,thick, bend right=0] (t2) to node [] {} (p2);

\end{tikzpicture}
}
\end{minipage}
\begin{minipage}{0.45\linewidth}
\centering

\begin{tikzpicture}[scale=0.8]
\node[ltsNode,minimum width=7pt,label=below:$s_0$](s0)at(0,0){};
\node[ltsNode,label=below:$s_1$](s1)at(-2,0){};
\node[ltsNode,label=below:$s_2$](s2)at(2,0){};
\draw[-{>[scale=2.5,length=2,width=2]}](s0)to node[auto,swap]{}(s1);
\draw[-{>[scale=2.5,length=2,width=2]}](s0)to node[auto,swap]{}(s2);
\end{tikzpicture}

\end{minipage}


\caption{
On the left, a system $S$. The LTS on the right represents both $RG(S)$ and $PRG(S)$.
The latter is initially directed but not directed.
}

\label{InitDirNotDir}

\end{figure}

\subsection{Initial directedness and strong liveness}

\vspace*{3mm}

A system $(N,M_0)$ is {\em strongly live} if, for each potentially reachable marking $M$, $(N,M)$ is live.
We~recall next lemma. 

\begin{lemma}[Strong liveness \cite{ArxivSSP20}]\label{LiveAndDirected}
Consider a live system $S$. 
If $PRG(S)$ is initially directed, then $S$ is strongly live.
\end{lemma}

We know that the converse of Lemma~\ref{LiveAndDirected} does not hold, 
even in the class of ordinary asymmetric-choice Petri nets~\cite{ArxivSSP20}.

\subsection{Known classes with directed reachability graph}

\vspace*{3mm}

The potential reachability graph of live HFC systems is known to be directed (Theorem 12 in \cite{STECS});
the reachability graph of live DSSP is directed as well (Theorem 4 in \cite{DSSP98}).\\

Persistent systems have a strong restriction on their bahavior: no transition firing can disable any other transition.
They do not include all HFC nor all DSSP systems,
and their reachability graph is directed in a stronger form, as expressed by Keller's theorem below.
We first need to recall the notion of residues, on which this theorem is based.

\begin{definition}[(Left) Residue]\label{residues.def}
Let $T$ be a set of labels (typically, transitions)
and 
$\tau,\sigma\in T^*$
two sequences over this set.
The {\em (left) residue of $\tau$ with respect to $\sigma$}, 
denoted by $\tau\pminus\sigma$,
arises from cancelling successively
in $\tau$ the leftmost occurrences of all symbols from $\sigma$, read from left to right.
Inductively:
$\tau\pminus\emptyseq=\tau$;
$\tau\pminus t=\tau$ if $t\notin\support(\tau)$;
$\tau\pminus t$ is the sequence obtained 
by erasing the leftmost $t$ in $\tau$ if $t\in\support(\tau)$;
and 
$\tau\pminus(t\sigma)=(\tau\pminus t)\pminus\sigma$.

For example, 
$acbcacbc\pminus abbcb=cacc$ 
and 
$abbcb\pminus acbcacbc=b$.

Residues naturally extend to T-vectors as follows:
for any sequence $\sigma$ and T-vector $Y$,
$\sigma \pminus Y$ is $\sigma$ in which,
for each transition $t$ in $\support(Y)$,
the $\min\{P(\sigma)(t),Y(t)\}$ leftmost occurrences of $t$ have been removed.
\end{definition}

\begin{theorem}[Keller \cite{keller}]\label{kellersTheorem}
Let $S$ be a persistent system. 
Let $\tau$ and $\sigma$ be two sequences feasible in $S$.
Then $\tau(\sigma\pminus\tau)$ and $\sigma(\tau\pminus\sigma)$ are both feasible in $S$ and lead to the same marking.
\end{theorem}

Keller's theorem applies to \WMGineq{} and the larger class of CF nets,
since they are structurally persistent (each place having at most one output).\\

In the next section, we exploit directedness to develop our first general condition ensuring the PR-R equality
in weighted Petri nets.

\section{Reverse nets, properties and the PR-R equality}\label{PRRrevDir}

In order to study the relationship between reachability and potential reachability,
we introduce the notion of reverse nets and sequences.
We also introduce related notation and behavioral properties, and develop new relations between these properties.

Then, we relate reversibility of a system and of its reverse to initial directedness,
yielding a new general sufficient condition of PR-R equality for weighted Petri nets.
We recall methods checking its reversibility assumption in weighted subclasses of Petri nets.

We deduce a sufficient condition of PR-R equality for the live HFC subclass and a polynomial-time variant of it.
Finally, we recall a liveness characterization for CF nets,
which will prove useful in the study of \WMGineq{}.

\subsection{Reverse nets and properties}

\vspace*{2mm}

\begin{definition}[Reverse nets, systems and sequences]\label{rev.def}
The {\em reverse} of a net $N$, denoted by $-N$, is obtained from $N$ by reversing all the arcs while keeping the weights.
The {\em reverse} of a system $S=(N,M_0)$, denoted by $-S$, is the system $(-N,M_0)$.
We denote by $\sigma^{\triangleleft}$ the sequence $\sigma$ followed in reverse order, called its {\em reverse}.
For example, if $\sigma = t_1 t_2 t_2 t_3$, then $\sigma^{\triangleleft} = t_3 t_2 t_2 t_1$. 
\end{definition}

The notation $-N$ stems from the fact that the incidence matrix of the reverse of $N$
is the opposite $-I$ of the incidence matrix $I$ of $N$,
so that $-I + I$ is null.

\begin{definition}[Properties $\mathcal{L}$, $\mathcal{R}$ and $\mathcal{B}$]
A system $S$ fulfills property $\mathcal{L}$ if $S$ and $-S$ are live;
it fulfills property $\mathcal{R}$ if $S$ and $-S$ are reversible;
it fulfills property $\mathcal{B}$ if $S$ and $-S$ are bounded.
\end{definition}

We assume that each Petri net has at least one transition. 
Next lemmas relate properties of a system to the same properties in its reverse,
and will prove useful in the study of subclasses.

\begin{lemma}[Properties $\mathcal{R}$ and $\mathcal{B}$]\label{RevReverseIsBounded}
Let us suppose that a system $S$ fulfills property~$\mathcal{R}$.
Then: $S$ is bounded iff $-S$ is bounded.
\end{lemma}

\begin{proof}
If $-S$ is reversible and unbounded, consider an unbounded place $p$ in $-A$:
for each $k$, there exists a sequence $\sigma_k$ in $-A$ that visits a marking $M_k$ such that $M_k(p) \ge k$  
and comes back to $M_0$,
thus $\sigma^{\triangleleft}_k$ is feasible in $A$, visiting the same marking $M_k$.
Thus, $-A$ is bounded.
\end{proof}

We obtain next result.

\begin{lemma}[$LR$ PR markings and property $\mathcal{R}$]\label{LRPRrevReverse}  
Consider a system $S$. Suppose that every potentially reachable marking of $S$ is live and reversible ($LR$).
Then $-S$ is live and reversible. 
\end{lemma}

\begin{proof}
Denote by $I$ the incidence matrix of $S$.
Suppose that $-S=(-N,M_0)$ is not reversible: 
consider a marking $M$ reachable in $-S$ with some sequence $\sigma^\triangleleft$
such that $M_0$ is not reachable in $(-N,M)$.
Since $S$ is live and reversible, a sequence $\alpha$ is feasible in $S$
that contains all transitions and leads back to $M_0$, 
hence there exists some positive integer $k$ such that
$M = M_0 + I \cdot Y$ with $Y = k \cdot \Parikh(\alpha) - \Parikh(\sigma) \ge 0$.
Thus, $M$ is potentially reachable in $S$,
so that $(N,M)$ is live and reversible.
Consequently, since $M_0$ is reached from $(N,M)$ by firing $\sigma$,
a sequence $\tau$ is feasible in $(N,M_0)$ that leads to $M$.
We deduce that the sequence $\tau^\triangleleft$ is feasible in $(-N,M)$ and leads to $(-N,M_0)$, contradiction.
Thus $-S$ is reversible. Moreover, $\alpha^\triangleleft$ is feasible in $-S$, which is consequently live.
Hence the claim.
\end{proof}

\subsection{Ensuring the PR-R equality from reversibility and initial directedness}

\vspace*{2mm}

We obtain the next sufficient condition of reachability for the markings in $PR(S)$.
Its proof is illustrated in Figure~\ref{FigProofTh}. 
This new result will be exploited in the sequel to ensure the PR-R equality in Petri net subclasses.

\begin{theorem}[Combining initial directedness with property $\mathcal{R}$]\label{PReachDirect}
Consider a Petri net system $S=(N,M_0)$ satisfying property $\mathcal{R}$
and
such that $PRG(S)$ is initially directed.
Then $R(S) = PR(S)$.
\end{theorem}

\begin{proof}
The proof is illustrated in Figure~\ref{FigProofTh}.
Consider any marking $M$ potentially reachable from $M_0$.
By initial directedness,
there exists $M' \in R((N,M_0)) \cap R((N,M))$,
with 
feasible sequences $M_0 \xrightarrow[]{\sigma_0} M'$ and $M \xrightarrow[]{\sigma_1} M'$.
The marking $M_0$ is reachable from $M'$ with some sequence $\sigma_2$
since the system is reversible.
Now, let us consider the reverse of these sequences.
In particular, the sequence $\sigma_2^{\triangleleft} \sigma_1^{\triangleleft}$ leads to $M$ in the reverse system $-S$;
since this system is also reversible,
a sequence $\sigma_3$ exists that leads to the initial marking.
In $S$,
$\sigma_3^{\triangleleft}$ leads to $M$, which is thus reachable.
\end{proof}

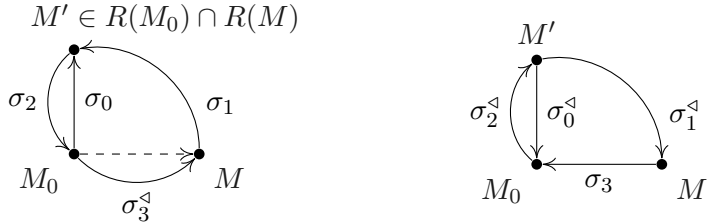
\begin{figure}[!ht]
 \centering
 
\begin{minipage}{0.5\linewidth}

\centering

\begin{tikzpicture}[scale=0.55,mypetristyle]
%
%
\node[ltsNode,label=below left:$M_0$](n0)at(0,0){};
\node[ltsNode,label=below right:$M$](n1)at(3,0){};
\node[ltsNode,label=above:$~~~~~~~~~~~~~~~~~~~~~~~~~M' \in R(M_0) \cap R(M)$](n2)at(0,2.5){};
\draw[-{>[scale=2.5,length=2,width=2]},dashed](n0)to node[auto,swap]{}(n1);
\draw[-{>[scale=2.5,length=2,width=2]},bend right=45](n0)to node[below]{$\sigma_3^{\triangleleft}$}(n1);
\draw[-{>[scale=2.5,length=2,width=2]}](n0)to node[auto,swap,right]{$\sigma_0$}(n2);
\draw[-{>[scale=2.5,length=2,width=2]},bend right=50](n2)to node[auto,swap,left]{$\sigma_2$}(n0);
\draw[-{>[scale=2.5,length=2,width=2]},bend right=50](n1)to node[right, near start]{~$\sigma_1$}(n2);
\end{tikzpicture}

\end{minipage}
\hspace*{10mm}
\begin{minipage}{0.38\linewidth}

\begin{tikzpicture}[scale=0.55,mypetristyle]
\node[ltsNode,label=below left:$M_0$](n0)at(0,0){};
\node[ltsNode,label=below right:$M$](n1)at(3,0){};
\node[ltsNode,label=above:$M'$](n2)at(0,2.5){};
\draw[-{>[scale=2.5,length=2,width=2]}](n1)to node[auto,swap,below]{$\sigma_3$}(n0);
\draw[-{>[scale=2.5,length=2,width=2]}](n2)to node[auto,swap,right]{$\sigma_0^{\triangleleft}$}(n0);
\draw[-{>[scale=2.5,length=2,width=2]},bend left=50](n0)to node[auto,swap,left]{$\sigma_2^{\triangleleft}$}(n2);
\draw[-{>[scale=2.5,length=2,width=2]},bend left=50](n2)to node[right, near end]{~$\sigma_1^{\triangleleft}$}(n1);
\end{tikzpicture}

\end{minipage}


\vspace*{3mm}

\caption{
Illustration of the proof of Theorem~\ref{PReachDirect}.
Part of the reachability graph of $S$ is depicted on the left.
On the right, sequences in $-S$ are considered.
}
\label{FigProofTh}
\end{figure}

\subsection{Checking reversibility}

\vspace*{2mm}

The reversibility checking problem is PSPACE-hard~\cite{esparza1996decidability}.
However, under the liveness assumption, characterizations of reversibility 
exist for HFC nets that often avoid to explore the reachability graph exhaustively.
Polynomial-time sufficient conditions of liveness and reversibility also exist 
for well-formed HFC nets and join-free (JF) nets (i.e. without synchronizations)~\cite{March09,HDM2016,HD2018}.

We recall the notion of a T-sequence and its importance for reversibility.

\begin{definition}[T-sequence \cite{Hujsa2015,HDM2016}]
Consider a system $S$ whose set of transitions is $T$ and denote by $I$ its incidence matrix.
A firing sequence $\sigma$ of $S$ is a T-sequence 
if it contains all transitions of $T$ (i.e. $\support(\sigma) = T$) and $I \cdot \Parikh(\sigma) = 0$ 
(i.e. $\Parikh(\sigma)$ is a consistency vector).
\end{definition}

In all weighted Petri nets, the existence of a feasible T-sequence is a known necessary condition 
of liveness and reversibility, taken together~\cite{Hujsa2015}.
It has also been proven sufficient for reversibility in live HFC systems, also called Equal-Conflict systems~\cite{HDM2016},
in a proper subclass of the live join-free systems (with an additional constraint on the reachable markings)~\cite{HD2018}
and
in live H$1$S systems~\cite{ArxivSSP20}.

\subsection{Directedness and PR-R equality in HFC systems}\label{SubSecReachHFC}

\vspace*{2mm}

We now consider the special case of the HFC subclass.

\begin{proposition}[Directedness of live HFC systems (Theorem~12 in \cite{STECS})]
Consider any HFC system $S$. If $S$ is live, then $PRG(S)$ is directed.
\end{proposition}

Applying Theorem~\ref{PReachDirect}, we deduce next result.

\begin{corollary}[PR-R equality in live HFC systems]\label{PRRHFC}
Consider a system $S$ satisfying property $\mathcal{R}$.
If $S$ is live and HFC, then $R(S) = PR(S)$.
\end{corollary}

Consider any well-formed HFC system $S$ whose reverse is also a well-formed HFC.
Theorem $28$ in \cite{STECS} provides a polynomial-time charaterization of well-formedness for HFC nets.
A wide-ranging linear-time sufficient condition of liveness and reversibility 
in well-formed HFC systems is given by Theorem $6.6$ in~\cite{HDM2016}.
Thus, in this subclass, we deduce a polynomial-time sufficient condition of PR-R equality.

\subsection{Liveness of CF systems}

\vspace*{2mm}

We recall the next characterization of liveness for weighted choice-free systems (which form a subclass of the HFC systems)
given as Corollary $4$ in \cite{tcs97}.

\begin{proposition}[Liveness of choice-free systems \cite{tcs97}]\label{LivenessOfCFnets}
Let $(N,M_0)$ be a choice-free system with incidence matrix $I$. 
It is live iff there exist a marking $M \in R((N,M_0))$
and
a firing sequence $\sigma \in L(N,M)$ such that $\Parikh(\sigma) \ge \one$ and $I \cdot \Parikh(\sigma) \ge 0$.
\end{proposition}

This result will prove useful in the study of \WMGineq{} in the next section.

\section{Reachability properties of \WMGineq{}}\label{ReachWMG}

Live Weighted T-Systems (WTS), in which each place has exactly one input and one output, fulfill the PR-R equality~\cite{WTS92}.
In this section, we extend this result to the live \WMGineq{}, which allow places without inputs and places without output.
With the aim of checking the liveness of a \WMGineq{} (as a precondition),
we also provide new characterizations of liveness for \WMGineq{} and their circuit subclass, as well as properties on their deadlocks.
Finally, we recall several other ways of checking liveness, as well as reversibility and boundedness in \WMGineq{}.

\subsection{Liveness, deadlockability and PR-R equality in \WMGineq{}}

\vspace*{2mm}

To obtain the PR-R equality result, we need the following proposition, which recalls Corollary~$1$ of~\cite{DH18}.

\begin{proposition}[Fireable T-vectors in \WMGineq{} \cite{DH18}]\label{RealizableTvectors}
Let $N = (P,T,W)$ be a \WMGineq{} with incidence matrix $I$. 
Let $M_0$ be any marking and $Y \in \N^T$ be a T-vector such that $M = M_0 + C \cdot Y \ge 0$. 
Let $\sigma$ be a transition sequence such that $Y \le \Parikh(\sigma)$.
Then, if $M_0 [\sigma\rangle$, 
there is a firing sequence $M_0 [\sigma'\rangle M$ such that $\Parikh(\sigma') = Y$.
\end{proposition}

We deduce that liveness is sufficient for ensuring the PR-R equality in this class.

\begin{corollary}\label{PRRinLiveWMG}
Every live \WMGineq{} fulfills the PR-R equality.
\end{corollary}

\begin{proof}
For each solution $(M,Y)$ of the state equation, since the system is live,
there exists a feasible sequence $\sigma$ whose Parikh vector is greater than or equal to $Y$,
hence Proposition~\ref{RealizableTvectors} applies
and 
$M$ is reachable.
\end{proof}

Notice that the live \WMGineq{} form a subclass of live HFC systems, hence their potential reachability graph is directed 
(as recalled for the HFC class in Section~\ref{PotentialReachDirected}).
Moreover, every live and bounded \WMGineq{} is reversible and fulfills liveness, boundedness and reversibility in its reverse~\cite{WTS92};
when boundedness is dropped, \WMGineq{} may be live without being reversible, and their reverse are not always live,
so that Theorem~\ref{PReachDirect} cannot be applied to derive the PR-R equality.
When the liveness assumption is dropped, examples not fulfilling the PR-R equality are easily built, 
as the one in Figure~\ref{DeadWmg}.\\

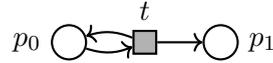
\begin{figure}[!h]
 

\centering

\begin{tikzpicture}[mypetristyle,scale=1]

\node (p0) at (0,0) [place,thick] {};
\node (p1) at (2,0) [place,thick] {};

\node [anchor=east] at (p0.west) {$p_0$};
\node [anchor=west] at (p1.east) {$p_1$};

\node (t) at (1,0) [transition,thick] {};
\node [anchor=south] at (t.north) {$t$};

\draw [->,thick,bend right=20] (p0) to node [below] {} (t);
\draw [->,thick,bend right=20] (t) to node [below,bend right=20] {} (p0);
\draw [->,thick] (t) to node [left] {} (p1);

\end{tikzpicture}


\caption{A deadlocked \MGineq{} (i.e. a unit-weighted \WMGineq{}) with marking $(0,0)$. For each integer $k > 0$, the marking $(0,k)$ is potentially reachable but not reachable.
}

\label{DeadWmg}

\end{figure}

So as to check liveness, we introduce several characterizations.
The next characterization of liveness, expressed in terms of the liveness of circuit subsystems, is extracted from~\cite{WTS92}.

\begin{proposition}[Liveness of WTS \cite{WTS92}]\label{LiveWTS}
A WTS $S=(N,M_0)$ is live iff every elementary circuit P-subsystem $C$ of $S$ is live.
\end{proposition}

This proposition readily extends to \WMGineq{} without source places (i.e. places with no input):

\begin{corollary}[Extension of Theorem 4.12 in \cite{WTS92}]\label{CorWTS}
A \WMGineq{} without source places $S=(N,M_0)$ is live iff every elementary circuit P-subsystem $C$ of $S$ is live.
\end{corollary}

We now introduce property $E$, which we use to obtain a variant of Corollary~\ref{CorWTS}.

\begin{definition}[Property $E$]
A Petri net system $S=(N,M_0)$ has the property $E$ if, for each solution $(M,Y)$ of its state equation,
$M$ enables at least one transition.
\end{definition}

\begin{theorem}
A \WMGineq{} without source places is live iff property $E$ is true in each elementary circuit P-subsystem.
\end{theorem}

\begin{proof}
We show the first direction $(\Rightarrow)$.
Consider any elementary circuit P-subsystem $C=(\projection{N}{P'},\projection{M_0}{P'})$ of the live \WMGineq{},
where $P'$ is the set of places of $C$:
applying Corollary~\ref{CorWTS}, $C$ is live.
Denote by $I_C$ the incidence matrix of the circuit $C$.
By Proposition~\ref{RealizableTvectors},
each solution $(M_c,Y)$ of the state equation $M_c = I_C \cdot Y + \projection{M_0}{P'}$
corresponds to a sequence feasible and a marking reachable in $C$,
in particular each solution enables at least one transition of $C$.
Thus, each circuit P-subsystem satisfies $E$.

Let us now consider the other direction $(\Leftarrow)$.
If $E$ is true in each elementary circuit P-subsystem $C$, 
then in particular each marking reachable in $C$ enables some transition,
$C$ is thus deadlock-free, hence live.
By Corollary~\ref{CorWTS}, we deduce the \WMGineq{} to be live.
\end{proof}

Let us now characterize the set of reachable markings of the non-live, connected \WMGineq{} without source places.

\begin{theorem}\label{ReachDeadlockWMG}
Consider a connected \WMGineq{} without source places $S=(N,M_0)$ with incidence matrix $I$.
If it is non-live, there exists a unique T-vector $Y_d$ of smallest cardinality such that $M_d = M_0 + I \cdot Y_d$ is a deadlock.
Moreover,
a sequence $\sigma_d$ exists that is feasible in $S$ and leads to $M_d$ such that $\Parikh(\sigma_d) = Y_d$.
\end{theorem}

\begin{proof}
Since \WMGineq{} are persistent, Keller's theorem applies (Theorem \ref{kellersTheorem}) and there is only one reachable deadlock, denoted by $M_d$.
Denote by $\sigma_d$ one of the sequences leading to $M_d$ with minimal length: 
there might be several ones with minimal length, but we show below that there is only one.
Denote by~$Y_d$ the Parikh vector of $\sigma_d$.

Suppose there is some T-vector $Y_d'$ defining a potentially reachable deadlock $M_d'$ that is not reachable,
such that $Y_d' \not\gneqq Y_d$, i.e. either $Y_d' \le Y_d$, or both are incomparable.

Let us prove by induction on the length $n$ of $\sigma_d$ that $Y_d' = Y_d$.
If $n = 0$, it is clear since we assumed $Y_d' \not\gneqq Y_d$.
If $n > 0$, let us write $\sigma_d = t_0 \tau$, where $t_0$ leads to a marking $M$.
We have two cases: $Y_d'(t_0) = 0$ and $Y_d'(t_0) > 0$.
In the first case, by structural persistence,
$t_0$ is still enabled at $M_d'$, a contradiction.
In the second case,
given that $M_d' = M + I \cdot (Y_d' - \one_{t_0})$ (where $\one_{t_0}$ denotes the T-vector whose only non-null component equals $1$ 
and has index $t_0$) 
applying the induction hypothesis to $M$, $(Y_d' - \one_{t_0})$ and $\tau$ (whose length is $n-1$ and Parikh vector is $Y_d - \one_{t_0}$)
yields 
$(Y_d - \one_{t_0}) = (Y_d' - \one_{t_0})$,
thus
$Y_d = Y_d'$.

Hence, there is a unique minimal $Y_d$,
and
for each sequence $\sigma_d$ feasible in $S$ that leads to $M_d$, we have $\Parikh(\sigma_d) = Y_d$.
\end{proof}


We denote by $DEAD$ 
the predicate on nets and markings such that $DEAD(N,M) = true$ 
iff the marking $M$ is a deadlock for the net $N=(P,T,W)$.

We obtain next theorem for checking liveness in a weighted circuit,
relaxing the non-negativity constraint on the components of the potentially reachable markings.

\begin{theorem}[Checking liveness of weighted circuits]\label{CheckLivenessCircuit}
A circuit system $S=(N,M_0)$ with incidence matrix $I$ is live 
iff 
the following system has no solution $(M_d,Y) \in \Z^{|P|} \times \N^{|T|}$:
$$\begin{cases}
M_d = M_0 + I \cdot Y \\
DEAD(M_d) \\
\end{cases}$$
%
\end{theorem}

\begin{proof}
\noindent $(\Leftarrow)$ If the system has no such solution, then in particular it has no solution $(M_d,Y) \in \N^{|P|} \times \N^{|T|}$,
hence no feasible sequence leads to a deadlock, thus $S$ is live.

$(\Rightarrow)$ Suppose that $S$ is live and that some solution $(M_d,Y) \in \Z^{|P|} \times \N^{|T|}$ to the system exists.
Denote by $\sigma'$ the sequence of maximal length that is feasible in $S$ and such that $\Parikh(\sigma') = Y' \le Y$,
leading to $M' = M_0 + I \cdot Y'$.
Let us define $Y'' = Y - Y'$, which has a non-empty support since $Y' \lneqq Y$. 
By definition of $Y'$, $M'$ does not enable any transition in the support of $Y''$.
Since $S$ is live, $M'$ enables some transition $t$ not in $\support(Y'')$ 
with unique input place $p$.
Since $M_d$ is a deadlock, $M_d(p) < W(p,t)$.
By definition of $t$ and $Y''$, $M'(p) \le M_d(p) < W(p,t)$, thus $t$ is not enabled at $M'$, a contradiction.

We deduce the claim.
\end{proof}

We deduce next corollary when conservativeness is assumed, 
allowing to relax the non-negativity constraint not only on the components of potentially reachable markings,
but also on the components of the T-vectors.

\begin{corollary}[Liveness of conservative weighted circuits]
A conservative circuit system $S=(N,M_0)$ with incidence matrix $I$ is live 
iff 
the following system has no solution $(M_d,Y) \in \Z^{|P|} \times \Z^{|T|}$:
$$\begin{cases}
M_d = M_0 + I \cdot Y \\ 
DEAD(M_d) \\ 
\end{cases}$$
%
\end{corollary}

\begin{proof}
\noindent $(\Leftarrow)$ This direction is obtained as in the proof of Theorem~\ref{CheckLivenessCircuit}.

\noindent $(\Rightarrow)$ Suppose that $S$ is live and that some solution $(M_d,Y) \in \Z^{|P|} \times \Z^{|T|}$ to the system exists.
If $(M_d,Y) \in \N^{|P|} \times \N^{|T|}$, applying Proposition~\ref{RealizableTvectors},
$M_d$ is reachable, contradicting liveness.
Thus, let us suppose that negative components appear in $M_d$ or $Y$.

Since $S$ is live, conservative (hence bounded) and strongly connected, it has a minimal T-semiflow~$Y$ with support $T$ 
(by consistency and Theorem~8 in \cite{tcs97}).
If $Y$ has some negative component,
then there exists a positive integer $k$ such that $Y_k = Y + k \cdot Y$ has only positive components
and such that $M_d = M_0 + I \cdot Y_k$.
Hence we suppose without loss of generality that $(M_d,Y) \in \Z^{|P|} \times \N^{|T|}$,
where $M_d$ has at least one negative component and $Y$ is a T-vector.
The rest of the proof is the same as in the proof of Theorem~\ref{CheckLivenessCircuit}.
\end{proof}

\subsection{Checking Properties of \WMGineq{}}

\vspace*{2mm}

Structural boundedness means boundedness for each marking,
while bounded systems are not always structurally bounded, even when they are live~\cite{LAT98}.
If a system is unbounded, the underlying net is not structurally bounded. 
We recall the next characterization for this property, which appears in various studies, e.g. in~\cite{LAT98}.

\begin{proposition}[Corollary 16 in \cite{LAT98}]\label{PropCharStructBound}
A net with incidence matrix $I$ is not structurally bounded iff there exists a T-vector $Y \gneqq 0$ such that $I \cdot Y \gneqq 0$.
\end{proposition}

Consequently, structural boundedness can be checked in polynomial time with linear programming 
(over the rationals, obtaining an integer-valued solution from a rational-valued one) for any weighted Petri net,
contrarily to boundedness, whose checking problem is EXPSPACE-complete.

To check liveness of \WMGineq{}, either Proposition~\ref{LivenessOfCFnets} can be used directly 
(since it is a liveness characterization for CF nets, which contain the \WMGineq{}),
or 
we use the results of the previous subsection;
since using these conditions in a naive way is generally costly,
one can try first the polynomial-time sufficient (and non-necessary) conditions of~\cite{March09}.

To check boundedness of \WMGineq{} using liveness: if it is live, it is bounded iff it is structurally bounded \cite{WTS92,tcs97,STECS}, 
which can be checked in polynomial time via Proposition~\ref{PropCharStructBound}.
Otherwise, suppose it is connected, without source places and non-live: then, by Theorem~\ref{ReachDeadlockWMG}, 
no infinite sequence is feasible
and all feasible sequences of maximal length lead to the same deadlock, thus the \WMGineq{} is necessarily bounded.
However, \WMGineq{} exist that are bounded, and not live nor structurally bounded (e.g. Figure~\ref{DeadWmg}). 

To check reversibility using liveness: suppose it is live. Then if it is bounded, it is reversible~\cite{WTS92}.
Otherwise, it is unbounded, and it is reversible iff a T-sequence (i.e. a sequence containing all transitions and getting back to the same marking) 
is feasible from the initial marking;
thus consistency is a necessary condition for reversibility under the liveness assumption.
More precisely, in live and connected \WMGineq{},
it is necessary and sufficient to check the existence of a T-sequence 
whose Parikh vector equals the unique minimal T-semiflow of the incidence matrix~\cite{WTS92,tcs97}.
By Theorem~{4.10} in~\cite{WTS92}, if $S$ is a consistent, non-deadlocked WTS, then reversibility and liveness of $S$ are equivalent;
this extends readily to \WMGineq{}, since consistency implies that the \WMGineq{} is a WTS.

\section{Systems with shared places not fulfilling the PR-R equality}\label{UnreachforSharedPlaces}

In this section, we provide various examples of systems belonging to the subclasses studied in this paper and
that do not fulfill the PR-R equality, while several other structural and behavioral properties are ensured.
These examples will prove useful to obtain our sufficient conditions of PR-R equality in the subsequent sections
and to show the sharpness of their assumptions.
Notably, they permit to highlight the importance of siphon properties.\\

\noindent {\bf Examples of AMG and H$1$S-\WMGineq{}.} In Figure \ref{CEPRAMG}, 
we provide three examples of AMG satisfying specific conditions, together with two of their P-subsystems:
\begin{itemize}
\item The AMG on the left, which is also a H$1$S-\WMGineq{}, shows that being live, reversible and bounded ($LRB$) 
with only one shared place, while not having a reversible reverse, is not sufficient for ensuring the PR-R equality.

\item The other two AMG in the middle have two shared places (hence are H$2$S-\WMGineq{}): 
each one is the reverse of the other one, both are $LRB$, 
thus in particular fulfill property $\mathcal{R}$, but do not satisfy the PR-R equality.

\item On the right, two live, unbounded P-subsystems of the AMG with two shared places, induced by a minimal siphon, are pictured,
and do not fulfill the PR-R equality.
\end{itemize}

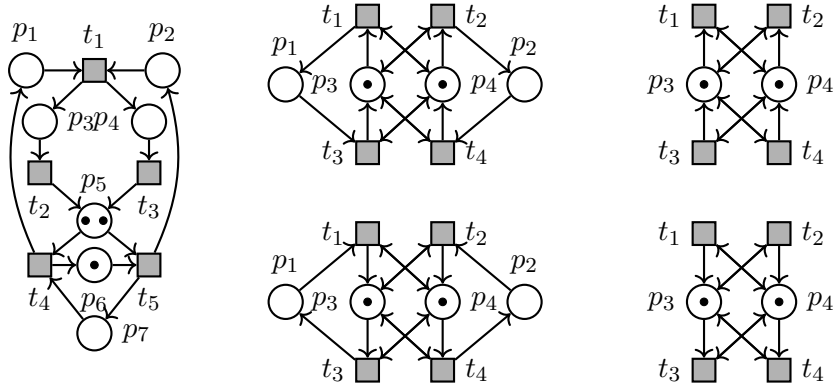
\begin{figure}[!ht]
\centering
\hspace{10mm}
\begin{tikzpicture}[scale=0.9,mypetristyle]
\node (p1) at (1,6) [place] {};
\node [anchor=south] at (p1.north) {$p_1$};
\node (p2) at (3,6) [place] {};
\node [anchor=south] at (p2.north) {$p_2$};
\node (p3) at (1.2,5.25) [place] {};
\node [anchor=west] at (p3.east) {$p_3$};
\node (p4) at (2.8,5.25) [place] {};
\node [anchor=east] at (p4.west) {$p_4$};
\node (p5) at (2,3.8) [place,tokens=2] {};
\node [anchor=south] at (p5.north) {$p_5$};

%

\node (p6) at (2,3.15) [place,tokens=1] {};
\node [anchor=north] at (p6.south) {$p_6$};

\node (p7) at (2,2.15) [place] {};
\node [anchor=west] at (p7.east) {$p_7$};

\node (t1) at (2,6) [transition,thick] {};
\node (t2) at (1.2,4.5) [transition,thick] {};
\node (t3) at (2.8,4.5) [transition,thick] {};
\node (t4) at (1.2,3.15) [transition,thick] {};
\node (t5) at (2.8,3.15) [transition,thick] {};

\node [anchor=south] at (t1.north) {$t_1$};
\node [anchor=north] at (t2.south) {$t_2$};
\node [anchor=north] at (t3.south) {$t_3$};
\node [anchor=north] at (t4.south) {$t_4$};
\node [anchor=north] at (t5.south) {$t_5$};

\draw [->,thick] (t1) to node [] {} (p3);
\draw [->,thick] (t1) to node [] {} (p4);
\draw [->,thick] (p3) to node [] {} (t2);
\draw [->,thick] (p4) to node [] {} (t3);
\draw [->,thick] (t2) to node [] {} (p5);
\draw [->,thick] (t3) to node [] {} (p5);
\draw [->,thick] (p5) to node [] {} (t4);
\draw [->,thick] (p5) to node [] {} (t5);
\draw [->,thick] (t4) to node [] {} (p6);
\draw [->,thick] (p6) to node [] {} (t5);
\draw [->,thick] (t5) to node [] {} (p7);
\draw [->,thick] (p7) to node [] {} (t4);

\draw [->,thick,bend left=20] (t4) to node [] {} (p1);
\draw [->,thick,bend right=20] (t5) to node [] {} (p2);

\draw [->,thick] (p1) to node [] {} (t1);
\draw [->,thick] (p2) to node [] {} (t1);

%

\end{tikzpicture}
\hspace{7mm}
\raisebox{20mm}{
\begin{minipage}{0.6\linewidth}
\begin{tikzpicture}[scale=0.9,mypetristyle]
\node (p1) at (0,0) [place] {};
\node [anchor=south] at (p1.north) {$p_1$};
\node (p2) at (3.5,0) [place] {};
\node [anchor=south] at (p2.north) {$p_2$};
\node (p3) at (1.2,0) [place,tokens=1] {};
\node [anchor=east] at (p3.west) {$p_3$};
\node (p4) at (2.3,0) [place,tokens=1] {};
\node [anchor=west] at (p4.east) {$p_4$};

\node (t1) at (1.2,1) [transition,thick] {};
\node (t2) at (2.3,1) [transition,thick] {};
\node (t3) at (1.2,-1) [transition,thick] {};
\node (t4) at (2.3,-1) [transition,thick] {};

\node [anchor=east] at (t1.west) {$t_1$};
\node [anchor=west] at (t2.east) {$t_2$};
\node [anchor=east] at (t3.west) {$t_3$};
\node [anchor=west] at (t4.east) {$t_4$};

\draw [->,thick] (t1) to node [] {} (p1);
\draw [->,thick] (p1) to node [] {} (t3);
\draw [->,thick] (t3) to node [] {} (p3);
\draw [->,thick] (p3) to node [] {} (t1);
\draw [->,thick] (p4) to node [] {} (t2);
\draw [->,thick] (t2) to node [] {} (p2);
\draw [->,thick] (p2) to node [] {} (t4);
\draw [->,thick] (t4) to node [] {} (p4);
\draw [->,thick,bend left=0] (p3) to node [] {} (t2);
\draw [->,thick,bend left=0] (t2) to node [] {} (p3);
\draw [->,thick,bend left=0] (p4) to node [] {} (t1);
\draw [->,thick,bend left=0] (t1) to node [] {} (p4);
\draw [->,thick] (p3) to node [] {} (t4);
\draw [->,thick] (t4) to node [] {} (p3);
\draw [->,thick] (p4) to node [] {} (t3);
\draw [->,thick] (t3) to node [] {} (p4);

\end{tikzpicture}
%
%
\hspace{10mm}
%
\begin{tikzpicture}[scale=0.9,mypetristyle]

\node (p3) at (1.2,0) [place,tokens=1] {};
\node [anchor=east] at (p3.west) {$p_3$};
\node (p4) at (2.3,0) [place,tokens=1] {};
\node [anchor=west] at (p4.east) {$p_4$};

\node (t1) at (1.2,1) [transition,thick] {};
\node (t2) at (2.3,1) [transition,thick] {};
\node (t3) at (1.2,-1) [transition,thick] {};
\node (t4) at (2.3,-1) [transition,thick] {};

\node [anchor=east] at (t1.west) {$t_1$};
\node [anchor=west] at (t2.east) {$t_2$};
\node [anchor=east] at (t3.west) {$t_3$};
\node [anchor=west] at (t4.east) {$t_4$};

\draw [->,thick] (t3) to node [] {} (p3);
\draw [->,thick] (p3) to node [] {} (t1);
\draw [->,thick] (p4) to node [] {} (t2);
\draw [->,thick] (t4) to node [] {} (p4);
\draw [->,thick,bend left=0] (p3) to node [] {} (t2);
\draw [->,thick,bend left=0] (t2) to node [] {} (p3);
\draw [->,thick,bend left=0] (p4) to node [] {} (t1);
\draw [->,thick,bend left=0] (t1) to node [] {} (p4);
\draw [->,thick] (p3) to node [] {} (t4);
\draw [->,thick] (t4) to node [] {} (p3);
\draw [->,thick] (p4) to node [] {} (t3);
\draw [->,thick] (t3) to node [] {} (p4);
\end{tikzpicture}\\
~ \\
%
%
%
\begin{tikzpicture}[scale=0.9,mypetristyle]
\node (p1) at (0,0) [place] {};
\node [anchor=south] at (p1.north) {$p_1$};
\node (p2) at (3.5,0) [place] {};
\node [anchor=south] at (p2.north) {$p_2$};
\node (p3) at (1.2,0) [place,tokens=1] {};
\node [anchor=east] at (p3.west) {$p_3$};
\node (p4) at (2.3,0) [place,tokens=1] {};
\node [anchor=west] at (p4.east) {$p_4$};

\node (t1) at (1.2,1) [transition,thick] {};
\node (t2) at (2.3,1) [transition,thick] {};
\node (t3) at (1.2,-1) [transition,thick] {};
\node (t4) at (2.3,-1) [transition,thick] {};

\node [anchor=east] at (t1.west) {$t_1$};
\node [anchor=west] at (t2.east) {$t_2$};
\node [anchor=east] at (t3.west) {$t_3$};
\node [anchor=west] at (t4.east) {$t_4$};

\draw [->,thick] (p1) to node [] {} (t1);
\draw [->,thick] (t3) to node [] {} (p1);
\draw [->,thick] (p3) to node [] {} (t3);
\draw [->,thick] (t1) to node [] {} (p3);
\draw [->,thick] (t2) to node [] {} (p4);
\draw [->,thick] (p2) to node [] {} (t2);
\draw [->,thick] (t4) to node [] {} (p2);
\draw [->,thick] (p4) to node [] {} (t4);
\draw [->,thick,bend left=0] (t2) to node [] {} (p3);
\draw [->,thick,bend left=0] (p3) to node [] {} (t2);
\draw [->,thick,bend left=0] (t1) to node [] {} (p4);
\draw [->,thick,bend left=0] (p4) to node [] {} (t1);
\draw [->,thick] (p3) to node [] {} (t4);
\draw [->,thick] (t4) to node [] {} (p3);
\draw [->,thick] (p4) to node [] {} (t3);
\draw [->,thick] (t3) to node [] {} (p4);

\end{tikzpicture}
%
\hspace{10mm}
%
%
\begin{tikzpicture}[scale=0.9,mypetristyle]

\node (p3) at (1.2,0) [place,tokens=1] {};
\node [anchor=east] at (p3.west) {$p_3$};
\node (p4) at (2.3,0) [place,tokens=1] {};
\node [anchor=west] at (p4.east) {$p_4$};

\node (t1) at (1.2,1) [transition,thick] {};
\node (t2) at (2.3,1) [transition,thick] {};
\node (t3) at (1.2,-1) [transition,thick] {};
\node (t4) at (2.3,-1) [transition,thick] {};

\node [anchor=east] at (t1.west) {$t_1$};
\node [anchor=west] at (t2.east) {$t_2$};
\node [anchor=east] at (t3.west) {$t_3$};
\node [anchor=west] at (t4.east) {$t_4$};

\draw [->,thick] (p3) to node [] {} (t3);
\draw [->,thick] (t1) to node [] {} (p3);
\draw [->,thick] (t2) to node [] {} (p4);
\draw [->,thick] (p4) to node [] {} (t4);
\draw [->,thick,bend left=0] (t2) to node [] {} (p3);
\draw [->,thick,bend left=0] (p3) to node [] {} (t2);
\draw [->,thick,bend left=0] (t1) to node [] {} (p4);
\draw [->,thick,bend left=0] (p4) to node [] {} (t1);
\draw [->,thick] (p3) to node [] {} (t4);
\draw [->,thick] (t4) to node [] {} (p3);
\draw [->,thick] (p4) to node [] {} (t3);
\draw [->,thick] (t3) to node [] {} (p4);
\end{tikzpicture}
\end{minipage}
}
%

\vspace*{4mm}

\caption{
On the left,
the AMG system $S=(N,M_0)$ is live, reversible and bounded ($LRB$),
with only one shared place.
Its reverse is $\overline{LR}B$ (where $\overline{Q}$ denotes the negation of property $Q$),
hence does not fulfill property~$\mathcal{R}$.
Indeed, in the reverse $-S$, $t_2$ can be fired two times, leading to the deadlock $M_D=(0,0,2,0,0,1,0)$. 
This marking is potentially reachable in $S$ through the Parikh vector $Y=(2,0,2,2,2)$,
although not reachable.
On the top middle,
the AMG $S'$ satisfies property $\mathcal{R}$: it is $LRB$ and its reverse $-S'$, pictured on the bottom middle, is also $LRB$.
$S'$ contains the minimal siphon $D'=\{p_3,p_4\}$ which induces the strongly connected, 
$LR\overline{B}$ subsystem $S_{D'}$
on the top right. On the bottom right, $-S_{D'}$ is $LR\overline{B}$.  
In $S'$, the marking $(1,1,0,0)$ is a potentially reachable deadlock obtained with the Parikh vector $Y=(1,1,0,0)$
that is not reachable.
}
\label{CEPRAMG}
\end{figure}

On the left of Figure \ref{CE2choiceplaces}, we provide an example of an AMG that is H$2$S-\WMGineq{}, 
live and fulfills property $\mathcal{R}$, whose minimal siphons have at most one shared place,
but does not fulfill the PR-R equality.\\

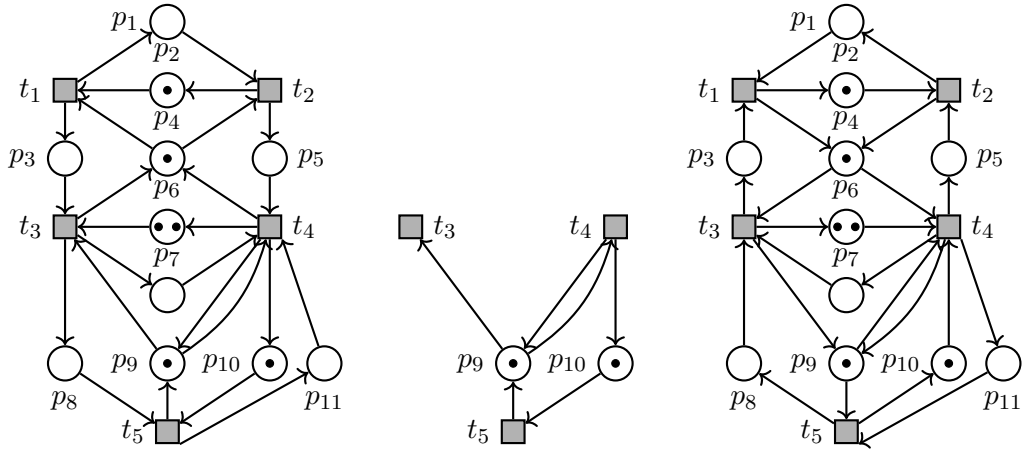
\begin{figure}[!ht]
\centering
\begin{tikzpicture}[scale=0.9,mypetristyle]

\node (p1) at (1.5,6) [place] {};
\node [anchor=east] at (p1.west) {$p_1$};

\node (p2) at (1.5,5) [place,tokens=1] {};
\node [anchor=south] at (p2.north) {$p_2$};

\node (p3) at (0,4) [place] {};
\node [anchor=east] at (p3.west) {$p_3$};

\node (p4) at (1.5,4) [place,tokens=1] {};
\node [anchor=south] at (p4.north) {$p_4$};

\node (p5) at (3,4) [place] {};
\node [anchor=west] at (p5.east) {$p_5$};

\node (p6) at (1.5,3) [place,tokens=2] {};
\node [anchor=south] at (p6.north) {$p_6$};

\node (p7) at (1.5,2) [place] {};
\node [anchor=south] at (p7.north) {$p_7$};

\node (p8) at (0,1) [place] {};
\node [anchor=north] at (p8.south) {$p_8$};

\node (p9) at (1.5,1) [place,tokens=1] {};
\node [anchor=east] at (p9.west) {$p_9$};

\node (p10) at (3,1) [place,tokens=1] {};
\node [anchor=east] at (p10.west) {$p_{10}$};

\node (p11) at (3.8,1) [place] {};
\node [anchor=north] at (p11.south) {$p_{11}$};

\node (t1) at (0,5) [transition,thick] {};
\node (t2) at (3,5) [transition,thick] {};
\node (t3) at (0,3) [transition,thick] {};
\node (t4) at (3,3) [transition,thick] {};
\node (t5) at (1.5,0) [transition,thick] {};

\node [anchor=east] at (t1.west) {$t_1$};
\node [anchor=west] at (t2.east) {$t_2$};
\node [anchor=east] at (t3.west) {$t_3$};
\node [anchor=west] at (t4.east) {$t_4$};
\node [anchor=east] at (t5.west) {$t_5$};

\draw [->,thick] (t1) to node [] {} (p1);
\draw [->,thick] (p1) to node [] {} (t2);
\draw [->,thick] (t2) to node [] {} (p2);
\draw [->,thick] (p2) to node [] {} (t1);
\draw [->,thick] (p5) to node [] {} (t4);
\draw [->,thick] (t4) to node [] {} (p4);
\draw [->,thick] (p4) to node [] {} (t2);
\draw [->,thick] (t2) to node [] {} (p5);
\draw [->,thick] (p3) to node [] {} (t3);
\draw [->,thick] (t3) to node [] {} (p4);
\draw [->,thick] (p4) to node [] {} (t1);
\draw [->,thick] (t1) to node [] {} (p3);
\draw [->,thick] (p6) to node [] {} (t3);
\draw [->,thick] (t3) to node [] {} (p7);
\draw [->,thick] (t4) to node [] {} (p6);
\draw [->,thick] (p7) to node [] {} (t4);
\draw [->,thick] (t4) to node [] {} (p10);
\draw [->,thick] (p10) to node [] {} (t5);
\draw [->,thick] (t5.south east) to node [] {} (p11);
\draw [->,thick] (p11) to node [] {} (t4.south east);
\draw [->,thick] (t5) to node [] {} (p9);
\draw [->,thick] (p9) to node [] {} (t3);
\draw [->,thick] (t3) to node [] {} (p8);
\draw [->,thick] (p8) to node [] {} (t5);
\draw [->,thick,bend right=20] (p9) to node [] {} (t4);
\draw [->,thick] (t4) to node [] {} (p9);
\end{tikzpicture}
\hspace{4mm}
\begin{tikzpicture}[scale=0.9,mypetristyle]

\node (p9) at (1.5,1) [place,tokens=1] {};
\node [anchor=east] at (p9.west) {$p_9$};

\node (p10) at (3,1) [place,tokens=1] {};
\node [anchor=east] at (p10.west) {$p_{10}$};

\node (t3) at (0,3) [transition,thick] {};
\node (t4) at (3,3) [transition,thick] {};
\node (t5) at (1.5,0) [transition,thick] {};

\node [anchor=west] at (t3.east) {$t_3$};
\node [anchor=east] at (t4.west) {$t_4$};
\node [anchor=east] at (t5.west) {$t_5$};

\draw [->,thick] (t4) to node [] {} (p10);
\draw [->,thick] (p10) to node [] {} (t5);

\draw [->,thick] (t5) to node [] {} (p9);
\draw [->,thick] (p9) to node [] {} (t3);
\draw [->,thick,bend right=20] (p9) to node [] {} (t4);
\draw [->,thick] (t4) to node [] {} (p9);
\end{tikzpicture}
\hspace{4mm}
\begin{tikzpicture}[scale=0.9,mypetristyle]

\node (p1) at (1.5,6) [place] {};
\node [anchor=east] at (p1.west) {$p_1$};

\node (p2) at (1.5,5) [place,tokens=1] {};
\node [anchor=south] at (p2.north) {$p_2$};

\node (p3) at (0,4) [place] {};
\node [anchor=east] at (p3.west) {$p_3$};

\node (p4) at (1.5,4) [place,tokens=1] {};
\node [anchor=south] at (p4.north) {$p_4$};

\node (p5) at (3,4) [place] {};
\node [anchor=west] at (p5.east) {$p_5$};

\node (p6) at (1.5,3) [place,tokens=2] {};
\node [anchor=south] at (p6.north) {$p_6$};

\node (p7) at (1.5,2) [place] {};
\node [anchor=south] at (p7.north) {$p_7$};

\node (p8) at (0,1) [place] {};
\node [anchor=north] at (p8.south) {$p_8$};

\node (p9) at (1.5,1) [place,tokens=1] {};
\node [anchor=east] at (p9.west) {$p_9$};

\node (p10) at (3,1) [place,tokens=1] {};
\node [anchor=east] at (p10.west) {$p_{10}$};

\node (p11) at (3.8,1) [place] {};
\node [anchor=north] at (p11.south) {$p_{11}$};

\node (t1) at (0,5) [transition,thick] {};
\node (t2) at (3,5) [transition,thick] {};
\node (t3) at (0,3) [transition,thick] {};
\node (t4) at (3,3) [transition,thick] {};
\node (t5) at (1.5,0) [transition,thick] {};

\node [anchor=east] at (t1.west) {$t_1$};
\node [anchor=west] at (t2.east) {$t_2$};
\node [anchor=east] at (t3.west) {$t_3$};
\node [anchor=west] at (t4.east) {$t_4$};
\node [anchor=east] at (t5.west) {$t_5$};

\draw [->,thick] (p1) to node [] {} (t1);
\draw [->,thick] (t2) to node [] {} (p1);
\draw [->,thick] (p2) to node [] {} (t2);
\draw [->,thick] (t1) to node [] {} (p2);
\draw [->,thick] (t4) to node [] {} (p5);
\draw [->,thick] (p4) to node [] {} (t4);
\draw [->,thick] (t2) to node [] {} (p4);
\draw [->,thick] (p5) to node [] {} (t2);
\draw [->,thick] (t3) to node [] {} (p3);
\draw [->,thick] (p4) to node [] {} (t3);
\draw [->,thick] (t1) to node [] {} (p4);
\draw [->,thick] (p3) to node [] {} (t1);
\draw [->,thick] (t3) to node [] {} (p6);
\draw [->,thick] (p7) to node [] {} (t3);
\draw [->,thick] (p6) to node [] {} (t4);
\draw [->,thick] (t4) to node [] {} (p7);
\draw [->,thick] (p10) to node [] {} (t4);
\draw [->,thick] (t5) to node [] {} (p10);
\draw [->,thick] (p11) to node [] {} (t5.south east);
\draw [->,thick] (t4.south east) to node [] {} (p11);
\draw [->,thick] (p9) to node [] {} (t5);
\draw [->,thick] (t3) to node [] {} (p9);
\draw [->,thick] (p8) to node [] {} (t3);
\draw [->,thick] (t5) to node [] {} (p8);
\draw [->,thick,bend left=20] (t4) to node [] {} (p9);
\draw [->,thick] (p9) to node [] {} (t4);
\end{tikzpicture}

\vspace*{4mm}

\caption{
On the left, a $LRB$ AMG $S=(N,M_0)$ with two shared places.
It contains the ill-formed (\ie non well-formed) P-subnet induced by the minimal siphon $\{p_9, p_{10}\}$,
depicted in the middle.
This siphon cannot become unmarked through any feasible firing sequence,
however it is unmarked at the potentially reachable marking obtained with the Parikh vector $(2,1,2,0,1)$.
It is not reachable in $S$, but becomes reachable if one initial token is added to $p_3$,
through the sequence $t_3 \, t_1 \, t_5 \, t_3  \, t_2 \, t_1$.
Thus, $S$ is live but not m-live (i.e. does not always remain live upon any addition of initial tokens). 
The reverse of $S$ is pictured on the right:
it is $LRB$. Thus $S$ fulfills property $\mathcal{R}$,
but does not fulfill the PR-R equality.
}
\label{CE2choiceplaces}
\end{figure}

\noindent {\bf Importance of the structure of P-subsystems induced by minimal siphons.}
In the middle of Figure~\ref{CE2choiceplaces}, we exhibit a minimal siphon that induces
an ill-formed (i.e. non well-formed) P-subsystem of the live AMG system on the left,
explaining the non-reachability of some potentially reachable marking that empties the siphon.
Indeed, no siphon can become unmarked at any reachable marking in live AMG~\cite{ChuXie97}. 

Going back to the examples in the middle of Figure \ref{CEPRAMG},
the unique minimal siphon with shared places contains $2$ shared places
and
induces a live, strongly connected P-subsystem fulfilling property~$\mathcal{R}$.
The two other minimal siphons induce two strongly connected, ($1$-)safe, live and reversible state machines
with one shared place, their reverse having the same properties.
Notice also that the set of siphons of~$S'$ is the same as in $-S'$ in this example,
although it is not the case in general.
Thus,
assuming each minimal siphon of $S'$ and $-S'$ to induce a strongly connected, live and reversible P-subsystem
is not sufficient for ensuring the PR-R equality in AMG, even with only two shared places (resource places),
hence in $2$S-\WMGineq{}.

Besides, we show in Figure \ref{2EWMGnotR} a non-homogeneous $4$S-\WMGineq{} in which each minimal siphon induces a well-formed,
live and reversible P-subsystem
and that does not satisfy the PR-R equality.\\

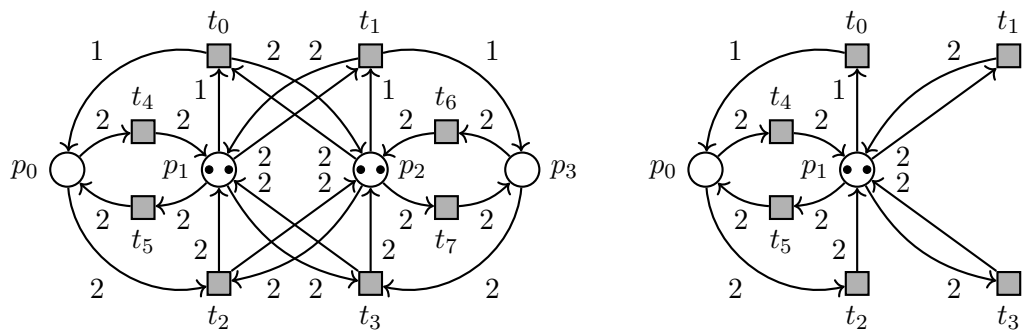
\begin{figure}[!ht]
\centering
\begin{tikzpicture}[scale=1,mypetristyle]

\node (p0) at (0,0) [place] {};
\node [anchor=east] at (p0.west) {$p_0$};

\node (p1) at (2,0) [place,tokens=2] {};
\node [anchor=east] at (p1.west) {$p_1$};

\node (p2) at (4,0) [place,tokens=2] {};
\node [anchor=west] at (p2.east) {$p_2$};

\node (p3) at (6,0) [place] {};
\node [anchor=west] at (p3.east) {$p_3$};

\node (t0) at (2,1.5) [transition,thick] {};
\node (t1) at (4,1.5) [transition,thick] {};
\node (t2) at (2,-1.5) [transition,thick] {};
\node (t3) at (4,-1.5) [transition,thick] {};
\node (t4) at (1,0.5) [transition,thick] {};
\node (t5) at (1,-0.5) [transition,thick] {};
\node (t6) at (5,0.5) [transition,thick] {};
\node (t7) at (5,-0.5) [transition,thick] {};

\node [anchor=south] at (t0.north) {$t_0$};
\node [anchor=south] at (t1.north) {$t_1$};
\node [anchor=north] at (t2.south) {$t_2$};
\node [anchor=north] at (t3.south) {$t_3$};
\node [anchor=south] at (t4.north) {$t_4$};
\node [anchor=north] at (t5.south) {$t_5$};
\node [anchor=south] at (t6.north) {$t_6$};
\node [anchor=north] at (t7.south) {$t_7$};

\draw [->,thick, bend left=20] (p0) to node [above] {$2$} (t4);
\draw [->,thick, bend right=50] (p0) to node [below left] {$2$} (t2);

\draw [->,thick] (p1) to node [above left] {$1$} (t0);
\draw [->,thick, bend left=20] (p1) to node [below] {$2$} (t5);
\draw [->,thick] (p1) to node [below, near start] {$2$} (t1.south west);
\draw [->,thick, bend right=25] (p1) to node [below, near end] {$2$} (t3);

\draw [->,thick] (p2) to node [above right] {$1$} (t1);
\draw [->,thick] (p2) to node [below, near start] {$2$} (t0.south east);
\draw [->,thick, bend left=25] (p2) to node [below, near end] {$2$} (t2);
\draw [->,thick, bend right=20] (p2) to node [below] {$2$} (t7);

\draw [->,thick, bend left=50] (p3) to node [below right] {$2$} (t3);
\draw [->,thick, bend right=20] (p3) to node [above] {$2$} (t6);

\draw [->,thick, bend right=50] (t0) to node [above left] {$1$} (p0);
\draw [->,thick, bend left=25] (t0) to node [above, near start] {$2$} (p2);

\draw [->,thick, bend right=25] (t1) to node [above, near start] {$2$} (p1);
\draw [->,thick, bend left=50] (t1) to node [above right] {$1$} (p3);

\draw [->,thick] (t2) to node [below left] {$2$} (p1);
\draw [->,thick] (t2.north east) to node [above, near end] {$2$} (p2);
\draw [->,thick] (t3.north west) to node [above, near end] {$2$} (p1);
\draw [->,thick] (t3) to node [below right] {$2$} (p2);
\draw [->,thick, bend left=20] (t4) to node [above] {$2$} (p1);
\draw [->,thick, bend left=20] (t5) to node [below] {$2$} (p0);
\draw [->,thick, bend right=20] (t6) to node [above] {$2$} (p2);
\draw [->,thick, bend right=20] (t7) to node [below] {$2$} (p3);

\end{tikzpicture}
\hspace*{5mm}
\begin{tikzpicture}[scale=1,mypetristyle]

\node (p0) at (0,0) [place] {};
\node [anchor=east] at (p0.west) {$p_0$};

\node (p1) at (2,0) [place,tokens=2] {};
\node [anchor=east] at (p1.west) {$p_1$};

\node (t0) at (2,1.5) [transition,thick] {};
\node (t1) at (4,1.5) [transition,thick] {};
\node (t2) at (2,-1.5) [transition,thick] {};
\node (t3) at (4,-1.5) [transition,thick] {};
\node (t4) at (1,0.5) [transition,thick] {};
\node (t5) at (1,-0.5) [transition,thick] {};

\node [anchor=south] at (t0.north) {$t_0$};
\node [anchor=south] at (t1.north) {$t_1$};
\node [anchor=north] at (t2.south) {$t_2$};
\node [anchor=north] at (t3.south) {$t_3$};
\node [anchor=south] at (t4.north) {$t_4$};
\node [anchor=north] at (t5.south) {$t_5$};

\draw [->,thick, bend left=20] (p0) to node [above] {$2$} (t4);
\draw [->,thick, bend right=50] (p0) to node [below left] {$2$} (t2);

\draw [->,thick] (p1) to node [above left] {$1$} (t0);
\draw [->,thick, bend left=20] (p1) to node [below] {$2$} (t5);
\draw [->,thick] (p1) to node [below, near start] {$2$} (t1.south west);
\draw [->,thick, bend right=25] (p1) to node [below, near end] {$2$} (t3);

\draw [->,thick, bend right=50] (t0) to node [above left] {$1$} (p0);

\draw [->,thick, bend right=25] (t1) to node [above, near start] {$2$} (p1);

\draw [->,thick] (t2) to node [below left] {$2$} (p1);
\draw [->,thick] (t3.north west) to node [above, near end] {$2$} (p1);
\draw [->,thick, bend left=20] (t4) to node [above] {$2$} (p1);
\draw [->,thick, bend left=20] (t5) to node [below] {$2$} (p0);

\end{tikzpicture}

\caption{
On the left, 
a non-homogeneous $4$S-\WMGineq{} in which the marking $(1,1,1,1)$ is potentially reachable although not reachable.
It satisfies properties $\mathcal{L}$ and $\mathcal{R}$.
On the right,
the P-subsystem induced by the minimal siphon $\{p_0,p_1\}$ is pictured.
Each minimal siphon that contains $p_0$ must contain $p_1$ and reciprocally.
By symmetry,
each minimal siphon that contains $p_2$ must contain $p_3$ and reciprocally.
Thus, the only minimal siphons of the system are $\{p_0,p_1\}$ and $\{p_2,p_3\}$,
which cover the set of places
and
induce strongly connected, live, reversible and conservative P-subsystems.
}
\label{2EWMGnotR}
\end{figure}

Next, in Section \ref{Sec1HEWMG}, we provide a sufficient condition of PR-R equality in H$1$S-\WMGineq{}, without assumptions on siphons.
Later on, in Section \ref{SecAMG}, 
we take inspiration from these examples to derive a sufficient condition of PR-R equality
in AMG with an arbitrary number of shared places, using properties of the minimal siphons.
We investigate reachability for the new class of \PCMGineq{} in Section~\ref{SecPCMG}.\\

\section{Ensuring the PR-R Equality in H$1$S-\WMGineq{}}\label{Sec1HEWMG}

\vspace*{2mm}

In this section, we first recall a result on the potential reachability graph of H$1$S-\WMGineq{},
and deduce a sufficient condition of PR-R equality in this class.
Then, we recall ways of checking its assumptions of liveness and reversibility in the larger class of H$1$S systems.

\subsection{Potential reachability in H$1$S-\WMGineq{}}

We recall next theorem, illustrated in Figure~\ref{FigOnePlaceCommonReachabilityTheorem},
which applies to H$1$S-\WMGineq{} and exploits liveness, but does not need the reversibility nor boundedness assumptions.

\begin{theorem}[Properties of the potential reachability graph in H$1$S-\WMGineq{}~\cite{ArxivSSP20}]\label{OnePlaceCommonReachability}
Consider a live H$1$S-\WMGineq{} $S=(N,M_0)$. 
For any Parikh vector $Y$ and marking $M$ such that $M_0 + I \cdot Y = M$,
there exists a firing sequence $M_0 \xrightarrow{\sigma}{}{} M'$
such that $M'$ is also reached by firing $\sigma \pminus Y$ from $M$, where $\Parikh(\sigma) \ge Y$.
Consequently, $PRG(S)$ is initially directed and $(N,M)$ is live.
\end{theorem}

\begin{figure}[!ht]

\centering


\hspace*{1cm}
\begin{tikzpicture}[scale=0.55,mypetristyle]
\node[ltsNode,label=below left:Live~$M_0$](n0)at(0,0){};
\node[ltsNode,label=below right:$M$~is~thus~live](n1)at(3,0){};
\node[ltsNode,label=above:$M'$](n2)at(1.5,2.5){};
\draw[-{>[scale=2.5,length=2,width=2]},bend left=0](n0)to node[auto,swap,above left,near start]{$\sigma$}(n2);
\draw[-{>[scale=2.5,length=2,width=2]},bend left=0](n1)to node[auto,swap,above right,near start]{$\sigma \pminus Y$}(n2);
\draw[-{>[scale=2.5,length=2,width=2]},dashed](n0)to node[auto,swap,below]{$Y$}(n1);
\end{tikzpicture}


\vspace*{2mm}

\caption{
Illustration of the claim of Theorem~\ref{OnePlaceCommonReachability}: for any such T-vector $Y$, there exists $\sigma$ leading to some $M'$ such that $\Parikh(\sigma) \ge Y$
and $\sigma \pminus Y$ leads to $M'$ from $M$.
}
\label{FigOnePlaceCommonReachabilityTheorem}
\end{figure}
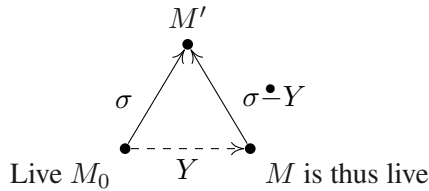

We deduce next result.

\begin{corollary}[Sufficient condition for PR-R equality in H$1$S-\WMGineq{}~\cite{ArxivSSP20}]\label{PRequivRoneChoice}
Let $S=(N,M_0)$ be a live, H$1$S-\WMGineq{} satisfying property $\mathcal{R}$.
Then $R(S) = PR(S)$.
\end{corollary}

\begin{proof}
Since $S$ fulfills the conditions of Theorem~\ref{OnePlaceCommonReachability}, $PRG(S)$ is initially directed.
Moreover, $S$ satisfies property $\mathcal{R}$.
Hence, Theorem~\ref{PReachDirect} applies.
\end{proof}

This corollary is not true in the class of live, H$2$S-\WMGineq{} satisfying property $\mathcal{R}$, 
as examplified in the middle of Figure~\ref{CEPRAMG}.\\

\subsection{Checking liveness and reversibility in H$1$S systems} 

\vspace*{2mm}

Let us recall characterizations of liveness and reversibility for the H$1$S systems, which contain the H$1$S-\WMGineq{}.

A siphon $D$ of a system $S=(N,M_0)$ is said to be {\em deadlocked} if, for each place $p$ in $D$, for each $t \in p^\bullet$, $M_0(p) < W(p,t)$.

\begin{theorem}[Liveness of H$1$S systems~\cite{ArxivSSP20}]\label{LivenessOf1EWMG}
Consider an H$1$S system.
It is live iff no minimal siphon is deadlocked at any reachable marking.
\end{theorem}

\begin{theorem}[Reversibility of live H$1$S systems~\cite{ArxivSSP20}]\label{ReversibilityOf1EWMG}
Consider a live H$1$S system $S=(N,M_0)$. Then $S$ is reversible iff $S$ enables a T-sequence.
\end{theorem}

This theorem is no more true in the case of two shared places, as shown in Figure~\ref{NonRevTwoPlacesBisFirst}. 

\begin{figure}[!ht]
\vspace*{1cm}
\centering

\begin{tikzpicture}[scale=1,mypetristyle]

\node (p0) at (-0.5,2.8) [place,tokens=1] {};
\node (p1) at (2.4,2) [place,tokens=1] {};
\node (p2) at (1,0.5) [place,tokens=0] {};
\node (p3) at (3.9,-0.3) [place,tokens=0] {};
\node (p4) at (0.3,1) [place,tokens=1] {};
\node (p5) at (0,2) [place,tokens=1] {};
\node (p6) at (3.4,0.5) [place,tokens=1] {};
\node (p7) at (3,1.5) [place,tokens=1] {};

\node [anchor=east] at (p0.west) {$p_0$};
\node [anchor=south] at (p1.north) {$p_1$};
\node [anchor=north] at (p2.south) {$p_2$};
\node [anchor=west] at (p3.east) {$p_3$};
\node [anchor=south] at (p4.north) {$p_4$};
\node [anchor=south] at (p5.north) {$p_5$};
\node [anchor=north] at (p6.south) {$p_6$};
\node [anchor=north] at (p7.south) {$~~~p_7$};

\node (t0) at (3.9,2.8) [transition] {};
\node (t1) at (1,2) [transition] {};
\node (t2) at (-0.5,-0.3) [transition] {};
\node (t3) at (2.4,0.5) [transition] {};

\node [anchor=west] at (t0.east) {$t_0$};
\node [anchor=south] at (t1.north) {$t_1$};
\node [anchor=east] at (t2.west) {$t_2$};
\node [anchor=north] at (t3.south) {$t_3$};

\draw [->,thick,bend left=0] (t2) to node {} (p2);
\draw [->,thick,bend left=0] (p2) to node {} (t2);

\draw [->,thick] (p2) to node {} (t1);
\draw [->,thick] (t3) to node {} (p2);
\draw [->,thick] (p1) to node {} (t3);
\draw [->,thick] (t1) to node {} (p1);

\draw [->,thick] (t0) to node {} (p7);
\draw [->,thick] (p7) to node {} (t3);

\draw [->,thick,bend right=0] (t0.west) to node {} (p1);
\draw [->,thick,bend left=0] (p1) to node {} (t0.west);

\draw [->,thick, bend right=0] (p0.north east) to node {} (t0.north west);
\draw [->,thick] (t2.north west) to node {} (p0.south west);


\draw [->,thick,bend left=0] (t0.south east) to node {} (p3.north east);
\draw [->,thick,bend left=0] (p3.south west) to node {} (t2.south east);

\draw [->,thick] (t2) to node {} (p4);
\draw [->,thick] (p4) to node {} (t1);

\draw [->,thick,bend right=0] (t1) to node {} (p5);
\draw [->,thick,bend right=0] (p5) to node {} (t2.north);

\draw [->,thick,bend right=0] (t3) to node {} (p6);
\draw [->,thick,bend right=0] (p6) to node {} (t0);

\end{tikzpicture}
\hspace*{4mm}
\raisebox{3mm}{
\begin{tikzpicture}[scale=1,mypetristyle]
\node[ltsNode,label=above:$s_0$,minimum width=7pt](s0)at(1,2){};
\node[ltsNode,label=above:$s_1$](s1)at(3,2){};
\node[ltsNode,label=left:$s_2$](s2)at(1,1){};
\node[ltsNode,label=right:$s_3$](s3)at(3,1){};
\node[ltsNode,label=below:$s_4$](s4)at(1,0){};
\node[ltsNode,label=below:$s_5$](s5)at(3,0){};
\node[ltsNode,label=left:$s_6$](s6)at(0,1){};
\node[ltsNode,label=right:$s_7$](s7)at(4,1){};
\draw[-{>[scale=2.5,length=2,width=2]}](s0)to node[above]{$t_3$}(s1);
\draw[-{>[scale=2.5,length=2,width=2]}](s1)to node[left]{$t_1$}(s3);
\draw[-{>[scale=2.5,length=2,width=2]}](s3)to node[left]{$t_0$}(s5);
\draw[-{>[scale=2.5,length=2,width=2]}](s5)to node[below right]{$t_3$}(s7);
\draw[-{>[scale=2.5,length=2,width=2]}](s0)to node[right]{$t_0$}(s2);
\draw[-{>[scale=2.5,length=2,width=2]}](s2)to node[right]{$t_3$}(s4);
\draw[-{>[scale=2.5,length=2,width=2]}](s4)to node[above]{$t_1$}(s5);
\draw[-{>[scale=2.5,length=2,width=2]}](s4)to node[below left]{$t_2$}(s6);
\draw[-{>[scale=2.5,length=2,width=2]}](s6)to node[above left]{$t_1$}(s0);
\draw[-{>[scale=2.5,length=2,width=2]}](s7)to node[above right]{$t_2$}(s1);
\end{tikzpicture}
}
\vspace*{5mm}

\caption{On the left, a unit-weighted, live, structurally bounded system with only two shared places, namely $p_1$ and $p_2$.
The system enables the T-sequence $t_0 \, t_3 \, t_2 \, t_1$ but is not reversible.
On the right, its non strongly connected reachability graph is pictured, with initial state $s_0$.
}
\label{NonRevTwoPlacesBisFirst}
\vspace*{2mm}
\end{figure}
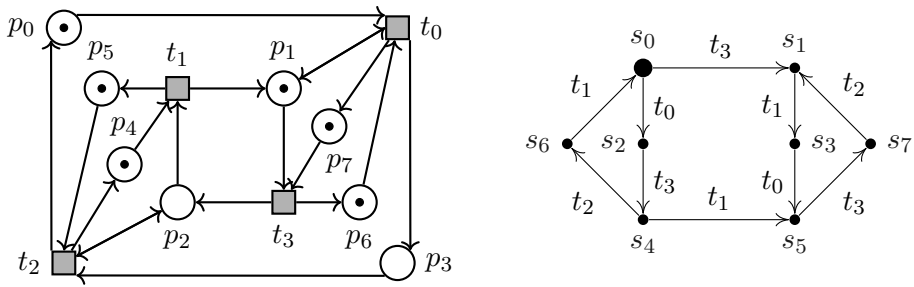

Once such a system, with one shared place, is known to be live, 
checking reversibility thus amounts to checking the existence of a feasible T-sequence.
It applies in particular to H$1$S-\WMGineq{} systems.

However, these characterizations of liveness and reversibility do not provide checking algorithms.
We leave their design as future work.

In the following sections, we study other classes of systems with shared places, namely AMG and \PCMGineq{}.
We investigate conditions inducing the PR-R equality in these systems.

\section{Ensuring the PR-R Equality in Augmented Marked Graphs}\label{SecAMG}

In this section,
we first recall reachability properties of AMG developed in \cite{ChuXie97}.
Then,
we develop results leading to a sufficient condition of PR-R equality.
We also discuss methods for checking the latter condition.

\subsection{Previous results on AMG}

Let us recall some of the known results about AMG.

\begin{proposition}[Invariant number of tokens, Property $24$ in \cite{ChuXie97}]\label{Prop24}
Each resource (i.e. shared) place $r \in R$ together with the places in paths $O_{r_i}$ induces a P-semiflow 
and
for all $M \in R(M_0)$,
$M(r) + \sum_{i \in N_r} \sum_{p \in O_{r_i}} M(p) = M_0(r) + \sum_{i \in N_r} \sum_{p \in O_{r_i}} M_0(p)$.
\end{proposition}

\begin{proposition}[Liveness implies reversibility, Property $25$ in \cite{ChuXie97}]\label{Prop25}
An augmented marked graph is reversible if it is live.
\end{proposition}

\begin{proposition}[Property $26$ in \cite{ChuXie97}]\label{Prop26}
Let $(N,M_0)$ be a live Petri net system with incidence matrix $I$
and
satisfying the assumptions $H_1$, $H_2$ and $H_3$ of the AMG.
Let $M^* \ge 0$ be a marking satisfying $(C1)$: $\exists Y \in \N^{|T|}$
such that 
$M^* = M_0 + I \cdot Y$,
and
$(C2)$: no place in paths $O_{r_i}$ is marked by $M^*$.
Then $M^*$ is reachable from $M_0$ and $M^*(r) > 0$, $\forall r \in R$. 
\end{proposition}

\begin{proposition}[Liveness and siphons, Property $27$ in \cite{ChuXie97}]\label{Prop27}
An augmented marked graph is live iff it cannot reach any marking at which some siphon is unmarked. 
\end{proposition}

\begin{proposition}[Liveness, siphons and home states, Property $29$ in \cite{ChuXie97}]\label{Prop29}
Let $(N,M_0)$ be a Petri net satisfying assumptions $H_1$, $H_2$ and $H_3$. 
If there exists a marking $M^*$ satisfying conditions $C1$ and $C2$ of Proposition~\ref{Prop26},
then $(N,M_0)$ is live iff no siphon is unmarked at any reachable marking. 
Furthermore, $M^*$ is a home state.
\end{proposition}

\subsection{New results on directedness, strong liveness and the PR-R equality in AMG}

We obtain next lemma, which proves the converse of Lemma~\ref{LiveAndDirected} in the live AMG class,
and even a stronger version since we get directedness instead of initial directedness.

\begin{lemma}[Directedness in strongly live AMG]\label{PRLiveAMGinitDir}
If an AMG $S$ is strongly live, then $PRG(S)$ is directed.
\end{lemma}

\begin{proof}
The proof is illustrated in Figure~\ref{FigProofPRLiveAMGinitDir}.
By Proposition \ref{Prop25}, $S=(N,M_0)$ is reversible. Consider any solution $(M,Y)$ of the state equation of $S$.
By liveness and reversibility, there exists a sequence $\sigma$ feasible in $S$ such that $\Parikh(\sigma) \ge Y$
and
$\sigma$ leads to the initial marking $M_0$.
By assumption, $(N,M)$ is live, and from the above we deduce
$M_0 = M + I \cdot (\Parikh(\sigma) - Y)$,
i.e. $M_0$ is potentially reachable from $M$ with the vector $\Parikh(\sigma) - Y$.
Now, Proposition~\ref{Prop26} applies by renaming $M_0$ as $M^*$
and
$M$ as $M_0$:
indeed,
$S$ is supposed to be an AMG,
hence fulfills all the conditions $H_1$ to $H_4$,
and
$(N,M)$ is live and fulfills the conditions $H_1$ to $H_3$.
Condition $C1$ is fulfilled by $S$ since $M_0 = M + I \cdot (\Parikh(\sigma) - Y)$,
and
condition $C2$ is fulfilled by $S$ too since it is an AMG.
We deduce that $M_0$ is reachable from $M$ with a firing sequence $\tau$.
Consequently,
for all pairs of potentially reachable markings $(M,M')$, $M_0$ is a marking reachable from both $M$ and $M'$,
implying that $PRG(S)$ is directed. 
\end{proof}

\begin{figure}[!ht]

\centering


%
\begin{tikzpicture}[scale=0.5,mypetristyle]
\node[ltsNode,label=below:Live~$M_0$](n0)at(0,0){};
\node[ltsNode,label=above left:Live~$M$](n1)at(-6,2.5){};
\node[ltsNode,label=above right:Live~$M'$](n2)at(6,2.5){};
\draw[-{>[scale=2.5,length=2,width=2]},dashed](n0)to node[above right, near end]{$~Y$}(n1);
\draw[-{>[scale=2.5,length=2,width=2]},dashed](n0)to node[above left, near end]{$Y'$}(n2);
\draw[-{>[scale=2.5,length=2,width=2]},bend right=30](n1)to node[above]{$\tau$}(n0);
\draw[-{>[scale=2.5,length=2,width=2]},bend left=30](n2)to node[above]{$\tau'$}(n0);
\draw[-{>[scale=2.5,length=2,width=2]},loop above,looseness=100](n0)to node[above]{$\sigma \ge_\Parikh Y,Y'$}(n0);
\end{tikzpicture}


\vspace*{5mm}

\caption{
Illustration of the proof of Lemma~\ref{PRLiveAMGinitDir}: for any solution $(M,Y)$ to the state equation associated to $S=(N,M_0)$, 
some sequence $\sigma$ is feasible at $M_0$ and leads to $M_0$ such that $\Parikh(\sigma) \ge Y$;
using $\sigma$ and $Y$,
we deduce that some sequence $\tau$ is feasible at $M$ and leads to $M_0$.
Thus, $M_0$ is reachable from each potentially reachable marking, hence the directedness of $PRG(S)$, 
illustrated for two solutions $(M,Y)$ and $(M',Y')$ to the state equation.
}
\label{FigProofPRLiveAMGinitDir}
\end{figure}
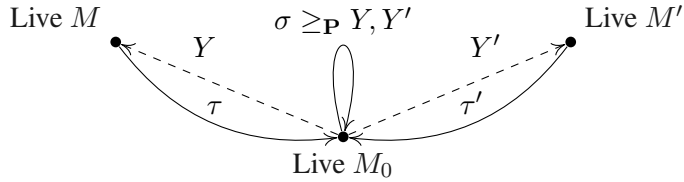

\begin{lemma}\label{AMGconsPRisLive}
Consider a live AMG $S=(N,M_0)$ in which each minimal siphon induces a conservative P-subnet,
and
consider any potentially reachable marking $M$.
Then, $(N,M)$ is live, thus $S$ is strongly live.
\end{lemma}

\begin{proof}
Since $S$ is live, each minimal siphon is initially marked.
Besides, it is reversible.
By conservativeness, each marking that is potentially reachable in the P-subsystem induced by any minimal siphon
necessarily marks this siphon, hence also each non-minimal siphon that contains it.
In particular, each marking that is potentially reachable in $S$ marks each siphon.
By liveness and reversibility,
$M_0$ is potentially reachable from $M$.
Now, in Proposition~\ref{Prop29}, let us rename $M_0$ as $M$ and $M^*$ as $M_0$.
Applying it,
since no marking reachable from $(N,M)$ empties any siphon, $(N,M)$ is live.
By Lemma~\ref{PRLiveAMGinitDir}, $PRG(S)$ is directed. Lemma~\ref{LiveAndDirected} applies and $S$ is strongly live.
\end{proof}

We deduce next theorem.

\begin{theorem}[Potential reachability in AMG]\label{PRinAMG}
Let $S$ be a live AMG system satisfying property $\mathcal{R}$,
in which each minimal siphon induces a conservative P-subnet.
Then $PR(S) = R(S)$.
\end{theorem}

\begin{proof}
Applying Lemma~\ref{AMGconsPRisLive}, for each potentially reachable marking $M$, $(N,M)$ is live.
Now, Lemma~\ref{PRLiveAMGinitDir} applies: $PRG(S)$ is directed, thus also initially directed.
By Theorem~\ref{PReachDirect}, $PR(S) = R(S)$.
\end{proof}

\noindent Figure~\ref{CE2choiceplaces} provides a counter-example when a minimal siphon induces a non-conservative P-subsystem, which is the only assumption relaxed.

\subsection{Checking Properties of Augmented Marked Graphs}

\vspace*{2mm}


Proposition \ref{Prop27} states that an AMG~$S$ is live iff it cannot reach any marking at which some siphon is unmarked.
Now, if each minimal siphon of $S$ and of $-S$ induces a conservative P-subsystem,
$S$~fulfills property~$\mathcal{L}$, since both $S$ and $-S$ are AMG and each siphon is initially marked.
Applying Proposition~\ref{Prop25}, $S$~also fulfills property $\mathcal{R}$.
Then, Theorem~\ref{PRinAMG} can be exploited.

In addition to these remarks, we obtain below a result relating the behavior of an AMG to the behavior of its reverse.

\begin{lemma}[Properties of the reverse AMG] 
Let us suppose that the AMG $A$ is live, bounded and reversible ($LBR$).
Then:
\begin{enumerate}
\item If $-A$ is live, then $-A$ is $LBR$.
\item If the underlying MG $G$ of $A$ is bounded, then $-A$ is bounded.
%

%
\end{enumerate}

\end{lemma}

\begin{proof}

\noindent (1.) Applying Proposition~\ref{Prop25}, and noticing that the reverse of an AMG is also an AMG, $-A$ is reversible. By Lemma~\ref{RevReverseIsBounded}, $-A$ is bounded, hence it is $LBR$.\\

\noindent (2.) Since $A$ is live, $G$ is live \cite{WTS92,tcs97}. 
Since $G$ is also bounded, it is conservative and structurally bounded \cite{WTS92,tcs97}.
By Proposition~\ref{Prop24}, each resource place of $-A$ (which is also an AMG) is bounded.
Since the other places belong to $-G$, which is structurally bounded since its reverse is a conservative MG,
each place of $-A$ is bounded, hence the claim.
\end{proof}

\section{Reachability in \PCMGineq{}}\label{SecPCMG}

In this section, we focus on well-structured \PCMGineq{}.
First, in this class,
we exhibit the structure of minimal siphons when the undirected graph $G$ is acyclic;
under the same constraint, we develop a characterization of liveness in terms of marked siphons.
Then, we develop a characterization of reversibility under the liveness assumption, without the acyclicity constraint.
Finally,
assuming that a live \PCMGineq{} system $S$ is obtained from an acyclic undirected graph $G$, 
we show that $S$ is reversible and fulfills the PR-R equality.

\subsection{Structure of siphon-induced P-subnets in well-structured \PCMGineq{} obtained from an acyclic graph $G$}

\vspace*{2mm}

Next theorem highlights the state machine structure of siphon-induced P-subnets when $G$ is acyclic.

\begin{theorem}[Structure of minimal siphons]\label{SiphonsOfPCMG}
Consider a well-structured \PCMGineq{} $N$ obtained from an acyclic, connected, undirected graph $G$.
Then each place belongs to a minimal siphon and a minimal trap of $N$, and each minimal siphon,
each minimal trap, induces a (strongly) connected state machine P-subnet of $N$.
\end{theorem}

\begin{proof}
Consider any P-subnet $N_D$ induced by a minimal siphon $D$, then $N_D$ is connected, since otherwise $D$ would not be minimal.
We prove the claim by strong induction on the number $n$ of shared places in $N$, for $n \ge 0$. 

\noindent $-$ Base case: $n=0$, $N$ is either empty, or an isolated place, or a strongly connected and well-formed marked graph.
In the first two cases, the claim is trivially true. 
In the third case, each place belongs to an elementary circuit P-subnet, hence the claim.

\noindent $-$ Inductive case: $n \ge 1$, $N$ has some shared place. 
If $D$ does not contain any shared place, it induces an elementary circuit P-subnet of $N$.
Otherwise, $D$ contains at least one shared place $p$ corresponding to a vertex $v$ of $G$. 
Consider any edge $e$ having $v$ as an extremity.

In the subgraph of $G$ obtained by deleting all the edges adjacent to $v$ except $e$,
denote by $G'$ the maximal connected component containing $e$.
The associated \PCMGineq{} $\projection{N}{G'}$ is strongly connected, well-structured and has a strictly smaller number of shared places.
The subset $D'$ of $D$ obtained by projection on $\projection{N}{G'}$ is a siphon,
and it is minimal in $\projection{N}{G'}$ since otherwise $D$ would not be minimal in~$N$.

Hence, the inductive hypothesis applies to $\projection{N}{G'}$:
the siphon $D'$ induces a strongly connected state machine P-subnet $N_{D'}$ of $\projection{N}{G'}$.

Now, we apply the same reasoning to every other edge adjacent to $v$.
We deduce that $D$ induces a strongly connected state machine P-subnet of $N$.

We proved the base case and the inductive case: the property is true for each $n \ge 0$. Hence the claim for the siphon case.

Using the same reasoning for traps or noticing that the reverse of such a well-structured \PCMGineq{} is also a well-structured \PCMGineq{}
obtained from the same undirected, acyclic graph $G$,
we deduce the claim for the trap case.
\end{proof}

If the acyclicity assumption is dropped, this result is no more true, as examplified in Figure~\ref{nonStrucLiveTriangle}.\\

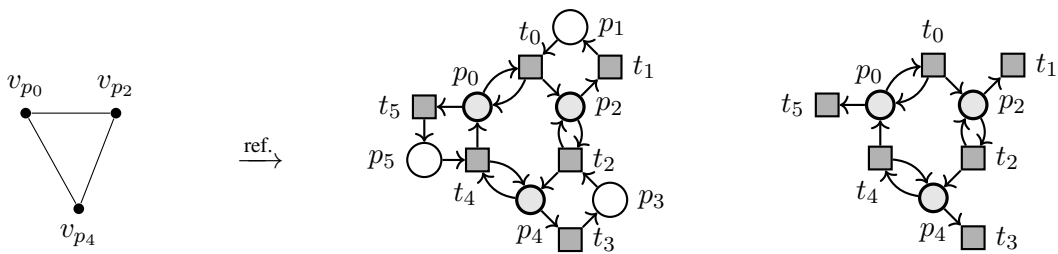
\begin{figure}[!ht]
 \centering
\begin{tabular}{cccc}
 
\begin{minipage}{0.2\linewidth}
\centering
\begin{tikzpicture}[scale=0.7,mypetristyle]
\node[ltsNode](n0)at(0,1.2){};
\node[ltsNode](n1)at(1.7,1.2){};
\node[ltsNode](n2)at(1,-0.6){};
\draw[-](n0)to node[auto,swap]{}(n1);
\draw[-](n1)to node[auto,swap]{}(n2);
\draw[-](n2)to node[auto,swap]{}(n0);
\node [anchor=south] at (n0.north) {$v_{p_0}$};
\node [anchor=south] at (n1.north) {$v_{p_2}$};
\node [anchor=north] at (n2.south) {$v_{p_4}$};
\end{tikzpicture}

\end{minipage}
&
$\xrightarrow{\textrm{ref.}}{}{}$
&
\begin{minipage}{0.3\linewidth}
\centering
\raisebox{7mm}{
\begin{tikzpicture}[scale=0.7,mypetristyle]

\node (p0) at (0,0) [petriNode] {};
\node (p1) at (1.75,1.5) [place,thick] {};
\node (p2) at (1.75,0) [petriNode] {};
\node (p3) at (2.5,-1.75) [place,thick] {};
\node (p4) at (1,-1.75) [petriNode] {};
\node (p5) at (-1,-1) [place,thick] {};

\node [anchor=south] at (p0.north west) {$p_0$};
\node [anchor=west] at (p1.east) {$p_1$};
\node [anchor=west] at (p2.east) {$p_2$};
\node [anchor=west] at (p3.east) {$p_3$};
\node [anchor=north] at (p4.south) {$p_4$};
\node [anchor=east] at (p5.west) {$p_5$};

\node (t0) at (1,0.75) [transition,thick] {};
\node (t1) at (2.5,0.75) [transition,thick] {};
\node (t2) at (1.75,-1) [transition,thick] {};
\node (t3) at (1.75,-2.5) [transition,thick] {};
\node (t4) at (0,-1) [transition,thick] {};
\node (t5) at (-1,0) [transition,thick] {};

\node [anchor=south] at (t0.north) {$t_0$};
\node [anchor=west] at (t1.east) {$t_1$};
\node [anchor=west] at (t2.east) {$t_2$};
\node [anchor=west] at (t3.east) {$t_3$};
\node [anchor=north] at (t4.south west) {$t_4$};
\node [anchor=east] at (t5.west) {$t_5$};

\draw [->,thick, bend left=30] (p0) to node [] {} (t0);
\draw [->,thick, bend left=30] (t0) to node [] {} (p0);
\draw [->,thick, bend left=0] (p1) to node [] {} (t0);
\draw [->,thick, bend left=0] (t1) to node [] {} (p1);
\draw [->,thick, bend left=0] (p2) to node [] {} (t1);
\draw [->,thick, bend left=0] (t0) to node [] {} (p2);

\draw [->,thick, bend left=30] (p2) to node [] {} (t2);
\draw [->,thick, bend left=30] (t2) to node [] {} (p2);
\draw [->,thick, bend left=0] (p3) to node [] {} (t2);
\draw [->,thick, bend left=0] (t3) to node [] {} (p3);
\draw [->,thick, bend left=0] (p4) to node [] {} (t3);
\draw [->,thick, bend left=0] (t2) to node [] {} (p4);

\draw [->,thick, bend left=30] (p4) to node [] {} (t4);
\draw [->,thick, bend left=30] (t4) to node [] {} (p4);
\draw [->,thick, bend left=0] (p5) to node [] {} (t4);
\draw [->,thick, bend left=0] (t5) to node [] {} (p5);
\draw [->,thick, bend left=0] (p0) to node [] {} (t5);
\draw [->,thick, bend left=0] (t4) to node [] {} (p0);

\end{tikzpicture}
}
\end{minipage}
&
\begin{minipage}{0.3\linewidth}
\centering
\raisebox{7mm}{
\begin{tikzpicture}[scale=0.7,mypetristyle]

\node (p0) at (0,0) [petriNode] {};
\node (p2) at (1.75,0) [petriNode] {};
\node (p4) at (1,-1.75) [petriNode] {};

\node [anchor=south] at (p0.north west) {$p_0$};
\node [anchor=west] at (p2.east) {$p_2$};
\node [anchor=north] at (p4.south) {$p_4$};

\node (t0) at (1,0.75) [transition,thick] {};
\node (t1) at (2.5,0.75) [transition,thick] {};
\node (t2) at (1.75,-1) [transition,thick] {};
\node (t3) at (1.75,-2.5) [transition,thick] {};
\node (t4) at (0,-1) [transition,thick] {};
\node (t5) at (-1,0) [transition,thick] {};

\node [anchor=south] at (t0.north) {$t_0$};
\node [anchor=west] at (t1.east) {$t_1$};
\node [anchor=west] at (t2.east) {$t_2$};
\node [anchor=west] at (t3.east) {$t_3$};
\node [anchor=north] at (t4.south west) {$t_4$};
\node [anchor=east] at (t5.west) {$t_5$};

\draw [->,thick, bend left=30] (p0) to node [] {} (t0);
\draw [->,thick, bend left=30] (t0) to node [] {} (p0);
\draw [->,thick, bend left=0] (p2) to node [] {} (t1);
\draw [->,thick, bend left=0] (t0) to node [] {} (p2);

\draw [->,thick, bend left=30] (p2) to node [] {} (t2);
\draw [->,thick, bend left=30] (t2) to node [] {} (p2);
\draw [->,thick, bend left=0] (p4) to node [] {} (t3);
\draw [->,thick, bend left=0] (t2) to node [] {} (p4);

\draw [->,thick, bend left=30] (p4) to node [] {} (t4);
\draw [->,thick, bend left=30] (t4) to node [] {} (p4);
\draw [->,thick, bend left=0] (p0) to node [] {} (t5);
\draw [->,thick, bend left=0] (t4) to node [] {} (p0);

\end{tikzpicture}
}
\end{minipage}

\end{tabular}

\vspace*{-5mm}

\caption{An undirected graph $G$ with place labels on the left, from which the \PCMGineq{} in the middle is derived by refinement.
This \PCMGineq{} is well-structured and not structurally live.
It contains the minimal siphon $\{p_0, p_2, p_4\}$ inducing the P-subnet depicted on the right,
which is not a state machine.
In the reverse net, the minimal siphon becomes a minimal trap.
}

\label{nonStrucLiveTriangle}

\end{figure}

\subsection{Liveness of well-structured \PCMGineq{} obtained from an acyclic graph $G$, in PTIME}

\vspace*{2mm}

Next theorem provides a characterization of liveness when $G$ is acyclic.
Noticing that each minimal siphon of the nets considered is a trap (by Theorem~\ref{SiphonsOfPCMG}),
the result can be seen as a variant of Commoner's theorem 
and the Home Marking theorem developed for free-choice nets (see e.g. \cite{DesEsp}).
We derive its polynomial-time complexity as a corollary.

\begin{theorem}[Liveness of well-structured \PCMGineq{} with acyclic graph]\label{LiveWellPCMG}
Consider a well-structured \PCMGineq{} system $S=(N,M_0)$ obtained from 
an acyclic undirected graph $G$ and having at least one transition. 
Then $S$ is live iff each minimal siphon--equivalently each minimal trap--of $N$ is marked by $M_0$.
\end{theorem}

\begin{proof}
By Theorem~\ref{SiphonsOfPCMG}, each minimal siphon $D$ induces a strongly connected state machine P-subnet $N_D$, 
and each place belongs to some minimal siphon.
If a siphon is not initially marked, then $S$ cannot be live. Let us prove the other direction.

We prove the claim by strong induction on the number $n$ of edges, $n \ge 1$ since we assumed $N$ to have at least one transition. 

Base cases: $n=1$, $N$ is a well-formed marked graph with at least one transition,
in which each elementary circuit P-subnet is initially marked, from which liveness is derived (by Corollary~\ref{CorWTS}).

Inductive case: $n>1$. We suppose the claim to be true for each $n' < n$.

Assume that each $D$ is initially marked and that a transition $t$ is dead at a reachable marking $M_t$.
Denote by $N_t=(P_t,T_e,W_t)$ the well-formed MG T-subnet containing $t$ and induced by the associated edge $e_t$.
Since $t$ is dead at $M_t$, $(N_t,\projection{M_t}{P_t})$ is deadlockable,
so that at least one (elementary) circuit P-subnet $C_t$ of $N_t$ is unmarked by $M_t$.
%
%
Since each minimal siphon of $N$ is initially marked, $C_t$ contains one or two shared places of $N$.

If $C_t$ contains exactly one shared place $p$, associated to vertex $v$ in $G$, consider the subgraph of $G$ obtained by deleting $e_t$
and 
denote by $G'$ the maximal connected component containing~$v$.
Let $D$ be a minimal siphon of $N$ containing the places $P_C$ of $C$, then $D \setminus (P_C \setminus \{p\})$ 
is a minimal siphon of $\projection{N}{G'}$ that is marked by $M_t$.
Consider any minimal siphon $D'$ of $\projection{N}{G'}$ containing $p$:
$D'$ is marked by $M_t$, 
since otherwise $D' \cup P_C$ is a minimal siphon of $N$ unmarked by $M_t$, which is impossible (by Theorem~\ref{SiphonsOfPCMG}).
Each other minimal siphon of $\projection{N}{G'}$ not containing $p$ is also a minimal siphon of $N$ and is marked by $M_t$.
Applying the inductive hypothesis, the T-subsystem $(\projection{N}{G'}, \projection{M_t}{G'})$ of $(N,M_t)$ is live 
and enables a sequence sending a token to $p$.

This reasoning applies symmetrically to the other extremity of $e_t$.

Now, if $C_t$ contains exactly two shared places $p$ and $p'$, associated to $v$ and $v'$ in $G$, 
consider the subgraph of $G$ obtained by deleting $e_t$
and 
denote by $G_v$ the maximal connected component containing~$v$, by $G_{v'}$ the maximal connected component containing $v'$.
If $D_v$ is a minimal siphon of $G_v$ unmarked by $M_t$ and containing $v$,
and if $D_{v'}$ is a minimal siphon of $G_{v'}$ unmarked by $M_t$ and containing $v'$ (they exist by Theorem~\ref{SiphonsOfPCMG}),
then $D_v \cup D_{v'} \cup P_C$ is a minimal siphon of $N$ unmarked by $M_t$, which is impossible.
Consequently, $D_v$ of $D_{v'}$ is marked by $M_t$, 
and $(\projection{N}{G_v}, \projection{M_t}{G_v})$ or $(\projection{N}{G_{v'}}, \projection{M_t}{G_{v'}})$
is a live T-subsystem of $(N,M_t)$ in which
a token can be sent to $p$ or $p'$.

We deduce that a marking $M_t'$ is reachable from $M_t$ such that $(N_t,\projection{M_t'}{P_t})$ is live.
This contradicts the fact that $t$ is dead at $M_t$. Thus, $S$ is live.

We proved the base cases and the inductive case, so that the claim is true for each $n$.
Hence the result for the siphon case.

The trap case is derived directly from the above, since, in the class of nets considered, 
each minimal siphon is a minimal trap and reciprocally, by Theorem~\ref{SiphonsOfPCMG}.
\end{proof}

\begin{corollary}[Polynomial-time complexity of Theorem~\ref{LiveWellPCMG}]\label{PCMGlivePTIME}
Checking the liveness of a well-structured, connected \PCMGineq{} obtained from an acyclic graph, is a polynomial-time problem.
\end{corollary}

\begin{proof}
Denote by $Q$ the set of all the places that are not marked by the initial marking $M_0$.
Computing the unique maximal trap or siphon $Q_{\mathrm{max}}$ included in $Q$ is done in polynomial-time, 
as detailed in the proof of Theorem 8.12 in~\cite{DesEsp}.
Using Theorem~\ref{LiveWellPCMG}, we have to check that each minimal siphon and trap is initially marked.
If $Q_{\mathrm{max}}$ is empty, then each of them is marked and the system is live.
Otherwise, either $Q_{\mathrm{max}}$ is minimal and the system is not live,
or it is not minimal, meaning that it contains a proper, minimal, non-empty and unmarked siphon or trap, 
implying the system is not live: 
the result is the same in both cases, so that we do not have to compute any minimal siphon or trap.
\end{proof}

\subsection{Live and well-structured \PCMGineq{} are not always reversible}

\vspace*{2mm}

The existence of a feasible T-sequence is yet not known to be sufficient for reversibility 
in the class of live and well-structured \PCMGineq{}, which are not included in the class of HFC nets.
Also, the class of AMG 
benefits from conditions ensuring liveness, boundedness and reversibility~\cite{chu1997deadlock},
but does not contain the \PCMGineq{}, hence the answer cannot be deduced directly from them.

In the class of well-structured \PCMGineq{},
reversibility is not necessarily deduced from liveness,
as shown in Figure~\ref{nonRevTriangle}.\\

\begin{figure}[!ht]
 \centering
\begin{tabular}{ccc}
 
\begin{minipage}{0.2\linewidth}
\centering
\begin{tikzpicture}[scale=0.7,mypetristyle]
\node[ltsNode](n0)at(0,1.2){};
\node[ltsNode](n1)at(1.7,1.2){};
\node[ltsNode](n2)at(1,-0.6){};
\draw[-](n0)to node[auto,swap]{}(n1);
\draw[-](n1)to node[auto,swap]{}(n2);
\draw[-](n2)to node[auto,swap]{}(n0);
\node [anchor=south] at (n0.north) {$v_{p_0}$};
\node [anchor=south] at (n1.north) {$v_{p_2}$};
\node [anchor=north] at (n2.south) {$v_{p_4}$};
\end{tikzpicture}

\end{minipage}
&
$\xrightarrow{\textrm{refinement}}{}{}$
&
\begin{minipage}{0.3\linewidth}
\centering
\raisebox{7mm}{
\begin{tikzpicture}[scale=0.7,mypetristyle]

\node (p0) at (0,0) [petriNode] {};
\node (p1) at (1.75,1.5) [place,thick,tokens=1] {};
\node (p2) at (1.75,0) [petriNode] {};
\node (p3) at (2.5,-1.75) [place,thick,tokens=1] {};
\node (p4) at (1,-1.75) [petriNode] {};
\node (p5) at (-1,-1) [place,thick,tokens=1] {};

\node [anchor=south] at (p0.north west) {$p_0$};
\node [anchor=west] at (p1.east) {$p_1$};
\node [anchor=west] at (p2.east) {$p_2$};
\node [anchor=west] at (p3.east) {$p_3$};
\node [anchor=north] at (p4.south) {$p_4$};
\node [anchor=east] at (p5.west) {$p_5$};

\node (t0) at (1,0.75) [transition,thick] {};
\node (t1) at (2.5,0.75) [transition,thick] {};
\node (t2) at (1.75,-1) [transition,thick] {};
\node (t3) at (1.75,-2.5) [transition,thick] {};
\node (t4) at (0,-1) [transition,thick] {};
\node (t5) at (-1,0) [transition,thick] {};

\node [anchor=south] at (t0.north) {$t_0$};
\node [anchor=west] at (t1.east) {$t_1$};
\node [anchor=west] at (t2.east) {$t_2$};
\node [anchor=west] at (t3.east) {$t_3$};
\node [anchor=north] at (t4.south west) {$t_4$};
\node [anchor=east] at (t5.west) {$t_5$};

\draw [->,thick, bend left=30] (p0) to node [] {} (t0);
\draw [->,thick, bend left=30] (t0) to node [] {} (p0);
\draw [->,thick, bend left=0] (t0) to node [] {} (p1);
\draw [->,thick, bend left=0] (p1) to node [] {} (t1);
\draw [->,thick, bend left=0] (t1) to node [] {} (p2);
\draw [->,thick, bend left=0] (p2) to node [] {} (t0);

\draw [->,thick, bend left=30] (p2) to node [] {} (t2);
\draw [->,thick, bend left=30] (t2) to node [] {} (p2);
\draw [->,thick, bend left=0] (t2) to node [] {} (p3);
\draw [->,thick, bend left=0] (p3) to node [] {} (t3);
\draw [->,thick, bend left=0] (t3) to node [] {} (p4);
\draw [->,thick, bend left=0] (p4) to node [] {} (t2);

\draw [->,thick, bend left=30] (p4) to node [] {} (t4);
\draw [->,thick, bend left=30] (t4) to node [] {} (p4);
\draw [->,thick, bend left=0] (t4) to node [] {} (p5);
\draw [->,thick, bend left=0] (p5) to node [] {} (t5);
\draw [->,thick, bend left=0] (t5) to node [] {} (p0);
\draw [->,thick, bend left=0] (p0) to node [] {} (t4);

\end{tikzpicture}
}
\end{minipage}

\end{tabular}

\vspace*{-5mm}

\caption{An undirected graph $G$ on the left from which the \PCMGineq{} system on the right is derived by refinement.
The latter is well-structured, live and non-reversible.
It is not an AMG since $p_1$, $p_3$ and $p_5$ are initially marked,
hence condition ($H_4$) is not fulfilled.
}

\label{nonRevTriangle}

\end{figure}
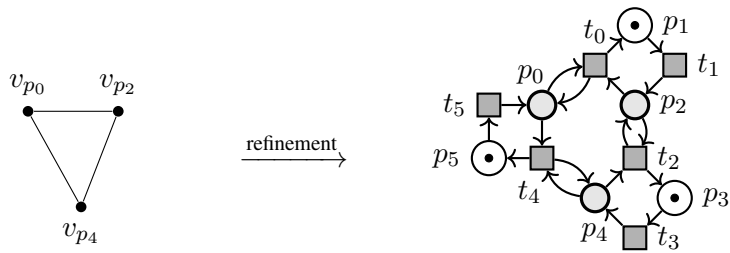

\subsection{A characterization of reversibility for well-structured and live \PCMGineq{}}

\vspace*{2mm}

We provide the characterization through the next theorem and the subsequent corollary.
The theorem studies the reversing of a single transition firing,
and
the corollary generalizes the result to finite sequences of arbitrary length.

\begin{theorem}[Reversing the action of a single transition]\label{ReverseOneTransition}
Consider a live and well-structured \PCMGineq{} system $S=(N,M_0)$
in which a T-sequence is feasible.
After the firing of any single transition in $S$, a feasible sequence exists that leads to $M_0$.
\end{theorem}

\begin{proof}
The proof is illustrated in Figure~\ref{FigProofReverseOneTransition}.
Denote by $G$ the undirected graph from which $S$ is obtained.
Denote by $\alpha$ some T-sequence feasible in $S$.
Consider the firing of some transition $t$ from $M_0$ leading to the marking $M_t$.

If no input place $p$ of $t$ is a shared place,
meaning $p\lbul = \{t\}$,
then $\alpha \pminus t$ is feasible from $M_t$ and reaches $M_0$,
in which case we deduce the claim.
Otherwise,
some place $p$ in $\lbul t$ is a shared place,
in which case
we prove in the following that a sequence leading to $M_0$ from $M_t$ also exists.

The transition $t$ belongs necessarily to a marked graph T-subsystem $S_t$ of $S$
corresponding to an edge $e_t=(v,v')$ of $G$ (i.e. the component $\beta(e_t)$), 
where $v$ is associated to $p$ and $v'$ to $p'$ (i.e. $\gamma(e_t)=(p,p')$).
If $S_t$ is reversible, then the initial marking can be reached trivially.
Otherwise, $S_t$ deadlocks at some marking $M_{dt}$ reachable from $M_t$ (since $S_t$ is a strongly connected and well-formed MG,
non-reversibility implies deadlockability).
In this case, $S_t$ necessarily contains at least one synchronization
and at least one unmarked elementary circuit in which $p'$ occurs (by Proposition~\ref{LiveWTS}).
Since $S$ is live, $p'$ is a shared place.
The following algorithm builds a tree $G'$ that is a subgraph of $G$.

Initially, $G'=(V',E')$ contains only the edge $e_t=(v,v')$, i.e. $V'=\{v,v'\}$ and $E'=e_t$.
Until some edge $e_r=(v_r,v_r')$ exists such that $v_r$ belongs to $V'\setminus \{v\}$
and the marked graph T-subsystem $S_{e_r}$ of $(N,M_{dt})$ associated to $e_r$
enables a sequence sending a token to $v_r$,
we iterate the following: 
%
add to $G'$ each non-visited edge $e=(v_1,v_2)$ such that $v_1$ belongs to $V'\setminus \{v\}$, but $v_2$ does not,
the MG T-subsystem $S_e$ of $(N,M_{dt})$ induced by $e$ is not live and becomes live if a token is added to $v_2$.

The loop terminates since each edge is visited at most once,
and the non-existence of the edge $e_r$ above means that the T-subsystem of $S$ corresponding to $G'$ can deadlock as a T-subsystem 
and cannot receive tokens from other transitions of the T-subsystem associated to $G \setminus G'$, contradicting liveness.

Hence, a token can be sent to $v_r$ by firing only in $S_{e_r}$, leading to a new marking $M$ through some sequence $\sigma_r$.
Some elementary sequence of edges $\mu = e_1 \ldots e_k$ exists in $G'$ from $v'$ to $v_r$
and represents a live and reversible T-subsystem $S'_\mu$ of $(N,M)$, by construction.

Walking along this path backwards, i.e. firing successively in 
the marked graph T-subsystems associated to $\beta(e_k) \ldots \beta(e_1)$,
a token is sent to $v'$ through some sequence $\sigma$, leading to $M_\sigma$.
The first firing of $t$ in $S_t$ can thus be canceled, 
leading to $M'$ such that for each place $p'' \neq v,v'$ of $P_{e_t}$, $M'(p'') = M_0(p'')$.
The T-subsystem $S'_\mu$ associated to $\mu$ and marked by $M'$ is live and reversible.
Thus $\sigma$ can be canceled in $S'_\mu$, leading to $M''$.
At $M''$, only places of the marked graph $N_r$ associated to $e_r$ (i.e. $N_r = \beta(e_r)$) might be marked differently from $M_0$.
Either $(N_r,\projection{M_t}{P_{e_r}})$ is live and reversible, in which case it is also live and reversible at $M''$
and $M_0$ is reachable from $M''$,
or $(N_r,\projection{M_t}{P_{e_r}})$ deadlocks, thus deadlocks also from $M''$, meaning that the sequence $\sigma_r$ that sent tokens to $v_r$
using only transitions of $N_r$ did not use any token from any shared place:
consequently,
the initial marking $M_0$ can be reached from $M''$ by firing $\alpha^k \pminus \sigma_r$ for the smallest positive integer $k$
such that $\Parikh(\alpha^k) \ge \Parikh(\sigma_r)$.
\end{proof}

\begin{figure}[!ht]
\vspace*{1cm}
 \centering
 
\begin{minipage}{0.25\linewidth}
\begin{tikzpicture}[scale=0.4,mypetristyle]
\node[ltsNode,label=below:$v_{p_3}$](s0)at(0,0){};
\node[ltsNode,label=above:$v_{p_0}$](s1)at(0,2){};
\node[ltsNode,label=below:$v_{p_2}$](s2)at(2,1){};
\node[ltsNode,label=below:$v_{p_5}$](s3)at(4,1){};
\node[ltsNode,label=above:$v_{p_7}$](s4)at(6,2){};
\node[ltsNode,label=below:$v_{p_9}$](s5)at(6,0){};

\draw[-](s0)to node[]{}(s1);
\draw[-](s1)to node[above]{$e_t$}(s2);
\draw[-](s2)to node[]{}(s0);
\draw[-](s2)to node[above]{$e_1$}(s3);
\draw[-](s3)to node[above]{$e_2$}(s4);
\draw[-](s4)to node[right]{$e_r$}(s5);
\draw[-](s5)to node[]{}(s3);

\end{tikzpicture}

\end{minipage}
$\xrightarrow{\textrm{ref.}}{}{}$
\begin{minipage}{0.65\linewidth}
\raisebox{7mm}{
\begin{tikzpicture}[scale=0.75,mypetristyle]

\node (p0) at (0,4) [petriNode,tokens=1] {};
\node (p1) at (2.25,3.5) [place] {};
\node (p2) at (3,2) [petriNode] {};
\node (p3) at (0,0) [petriNode] {};
\node (p4) at (-0.75,2) [place] {};
\node (p5) at (5,2) [petriNode] {};
\node (p6) at (6,4) [place,tokens=1] {};
\node (p7) at (8,4) [petriNode] {};
\node (p8) at (8.75,2) [place,tokens=1] {};
\node (p9) at (8,0) [petriNode] {};
\node (p10) at (6.2,0.35) [place,tokens=1] {};

\node [anchor=south] at (p0.north west) {$p_0$};
\node [anchor=south west] at (p1.north east) {$p_1$};
\node [anchor=north] at (p2.south) {$p_2$};
\node [anchor=north] at (p3.south) {$p_3$};
\node [anchor=east] at (p4.west) {$p_4$};
\node [anchor=north] at (p5.south) {$p_5$};
\node [anchor=south] at (p6.north) {$p_6$};
\node [anchor=south] at (p7.north) {$p_7$};
\node [anchor=west] at (p8.east) {$p_8$};
\node [anchor=north] at (p9.south) {$p_9$};
\node [anchor=north] at (p10.south) {$p_{10}$};

\node (t1) at (1.25,4) [transition] {};
\node (t2) at (2.25,2.75) [transition] {};
\node (t3) at (1.9,0.25) [transition] {};
\node (t4) at (1.6,1) [transition] {};
\node (t5) at (0,3) [transition] {};
\node (t6) at (-0.75,1) [transition] {};
\node (t7) at (4,2.5) [transition] {};
\node (t8) at (4,1.5) [transition] {};
\node (t9) at (5,3) [transition] {};
\node (t10) at (7,4) [transition] {};
\node (t11) at (8.75,3) [transition] {};
\node (t12) at (8,1) [transition] {};
\node (t13) at (7,0) [transition] {};
\node (t14) at (5.75,1.25) [transition] {};

\node [anchor=south] at (t1.north) {$t_1=t$};
\node [anchor=east] at (t2.west) {$t_2$};
\node [anchor=west] at (t3.east) {$t_3$};
\node [anchor=east] at (t4.west) {$t_4$};
\node [anchor=east] at (t5.west) {$t_5$};
\node [anchor=east] at (t6.west) {$t_6$};
\node [anchor=south] at (t7.north) {$t_7$};
\node [anchor=north] at (t8.south) {$t_8$};
\node [anchor=south] at (t9.north) {$t_9$};
\node [anchor=south] at (t10.north) {$t_{10}$};
\node [anchor=west] at (t11.east) {$t_{11}$};
\node [anchor=west] at (t12.east) {$t_{12}$};
\node [anchor=north] at (t13.south) {$t_{13}$};
\node [anchor=west] at (t14.east) {$t_{14}$};

\draw [->, bend left=0] (p0) to node [] {} (t1);
\draw [->, bend left=0] (t1) to node [] {} (p1);
\draw [->, bend left=0] (p1) to node [] {} (t2);
\draw [->, bend left=0] (t2) to node [] {} (p0);
\draw [->, bend left=20] (t2) to node [] {} (p2);
\draw [->, bend left=20] (p2) to node [] {} (t2);
\draw [->, bend left=0] (p3) to node [] {} (t4);
\draw [->, bend left=0] (t4) to node [] {} (p2);
\draw [->, bend left=0] (p2) to node [] {} (t3);
\draw [->, bend left=0] (t3) to node [] {} (p3);

\draw [->, bend left=20] (p0) to node [] {} (t5);
\draw [->, bend left=20] (t5) to node [] {} (p0);
\draw [->, bend left=0] (p3) to node [] {} (t5);
\draw [->, bend left=0] (t5) to node [] {} (p4);

\draw [->, bend left=0] (p4) to node [] {} (t6);
\draw [->, bend left=0] (t6) to node [] {} (p3);

\draw [->, bend left=0] (p2) to node [] {} (t7);
\draw [->, bend left=0] (t7) to node [] {} (p5);
\draw [->, bend left=0] (p5) to node [] {} (t8);
\draw [->, bend left=0] (t8) to node [] {} (p2);

\draw [->, bend left=0] (p5) to node [] {} (t9);
\draw [->, bend left=0] (t9) to node [] {} (p6);

\draw [->, bend left=0] (p6) to node [] {} (t10);
\draw [->, bend left=0] (t10) to node [] {} (p5);
\draw [->, bend left=20] (p5) to node [] {} (t14);
\draw [->, bend left=20] (t14) to node [] {} (p5);

\draw [->, bend left=0] (p7) to node [] {} (t12);
\draw [->, bend left=0] (t12) to node [] {} (p8);
\draw [->, bend left=0] (p8) to node [] {} (t11);
\draw [->, bend left=0] (t11) to node [] {} (p7);

\draw [->, bend left=20] (t12) to node [] {} (p9);
\draw [->, bend left=20] (p9) to node [] {} (t12);

\draw [->, bend left=20] (p5) to node [] {} (t14);

\draw [->, bend left=20] (p7) to node [] {} (t10);
\draw [->, bend left=20] (t10) to node [] {} (p7);

\draw [->, bend left=0] (p9) to node [] {} (t13);
\draw [->, bend left=0] (t13) to node [] {} (p10);

\draw [->, bend left=0] (p10) to node [] {} (t14);
\draw [->, bend left=0] (t14) to node [] {} (p9);

\end{tikzpicture}
}
\end{minipage}



\caption{
Illustration of the proof of Theorem~\ref{ReverseOneTransition}.
The graph $G$ is pictured on the left, with $\mu = e_1 e_2$.
The \PCMGineq{} $S$ obtained from $G$ is pictured on the right and contains cycles.
}

\label{FigProofReverseOneTransition}
\vspace*{1cm}
\end{figure}
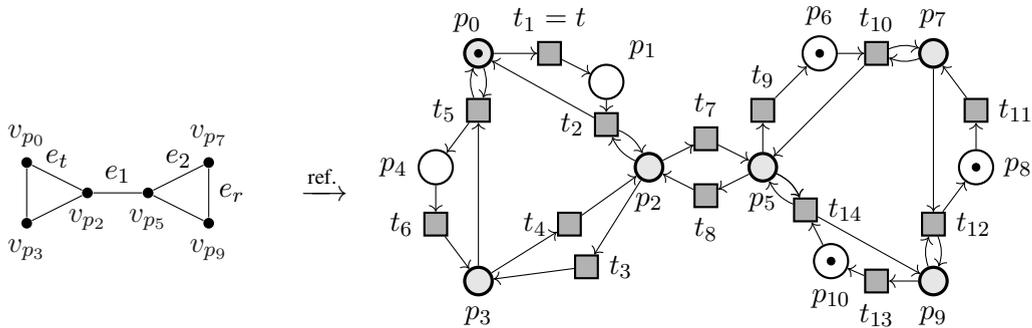

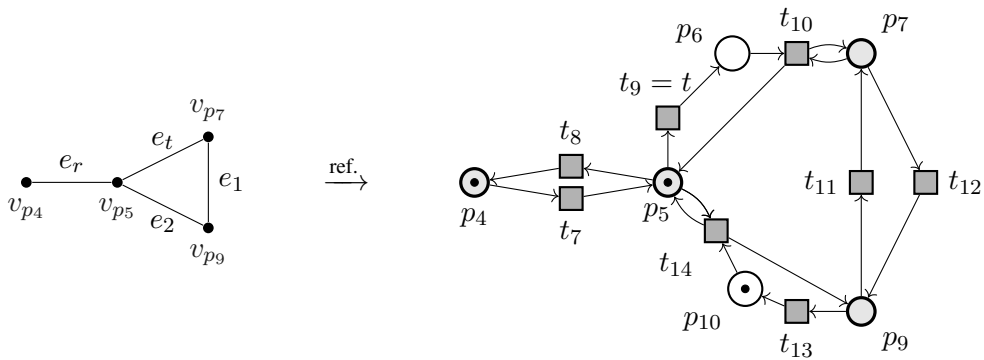
\begin{figure}[!ht]
 \centering
 
\begin{minipage}{0.32\linewidth}
\centering
\begin{tikzpicture}[scale=0.6,mypetristyle]

\node[ltsNode,label=below:$v_{p_4}$](s2)at(2,1){};
\node[ltsNode,label=below:$v_{p_5}$](s3)at(4,1){};
\node[ltsNode,label=above:$v_{p_7}$](s4)at(6,2){};
\node[ltsNode,label=below:$v_{p_9}$](s5)at(6,0){};

\draw[-](s2)to node[above]{$e_r$}(s3);
\draw[-](s3)to node[above]{$e_t$}(s4);
\draw[-](s4)to node[right]{$e_1$}(s5);
\draw[-](s5)to node[below]{$e_2$}(s3);

\end{tikzpicture}

\end{minipage}
$\xrightarrow{\textrm{ref.}}{}{}$
\begin{minipage}{0.57\linewidth}
\centering
\raisebox{0mm}{
\begin{tikzpicture}[scale=0.85,mypetristyle]

\node (p4) at (2,2) [petriNode,tokens=1] {};
\node (p5) at (5,2) [petriNode,tokens=1] {};
\node (p6) at (6,4) [place] {};
\node (p7) at (8,4) [petriNode] {};
\node (p9) at (8,0) [petriNode] {};
\node (p10) at (6.2,0.35) [place,tokens=1] {};

\node [anchor=north] at (p4.south) {$p_4$};
\node [anchor=north] at (p5.south west) {$p_5$};
\node [anchor=south east] at (p6.west) {$p_6$};
\node [anchor=south west] at (p7.north east) {$p_7$};
\node [anchor=north west] at (p9.south east) {$p_9$};
\node [anchor=north east] at (p10.south west) {$p_{10}$};

\node (t7) at (3.5,1.75) [transition] {};
\node (t8) at (3.5,2.25) [transition] {};
\node (t9) at (5,3) [transition] {};
\node (t10) at (7,4) [transition] {};
\node (t11) at (8,2) [transition] {};
\node (t12) at (9,2) [transition] {};
\node (t13) at (7,0) [transition] {};
\node (t14) at (5.75,1.25) [transition] {};

\node [anchor=north] at (t7.south) {$t_7$};
\node [anchor=south] at (t8.north) {$t_8$};
\node [anchor=south] at (t9.north west) {$t_9 = t$};
\node [anchor=south] at (t10.north) {$t_{10}$};
\node [anchor=east] at (t11.west) {$t_{11}$};
\node [anchor=west] at (t12.east) {$t_{12}$};
\node [anchor=north] at (t13.south) {$t_{13}$};
\node [anchor=north east] at (t14.south west) {$t_{14}$};

\draw [->, bend left=0] (p4) to node [] {} (t7);
\draw [->, bend left=0] (t7) to node [] {} (p5);

\draw [->, bend left=0] (p5) to node [] {} (t8);
\draw [->, bend left=0] (t8) to node [] {} (p4);

\draw [->, bend left=0] (p5) to node [] {} (t9);
\draw [->, bend left=0] (t9) to node [] {} (p6);

\draw [->, bend left=0] (p6) to node [] {} (t10);
\draw [->, bend left=0] (t10) to node [] {} (p5);
\draw [->, bend left=20] (p5) to node [] {} (t14);
\draw [->, bend left=20] (t14) to node [] {} (p5);

\draw [->, bend left=0] (p7) to node [] {} (t12);
\draw [->, bend left=0] (t11) to node [] {} (p7);

\draw [->, bend left=0] (t12) to node [] {} (p9);
\draw [->, bend left=0] (p9) to node [] {} (t11);

\draw [->, bend left=20] (p5) to node [] {} (t14);

\draw [->, bend left=20] (p7) to node [] {} (t10);
\draw [->, bend left=20] (t10) to node [] {} (p7);

\draw [->, bend left=0] (p9) to node [] {} (t13);
\draw [->, bend left=0] (t13) to node [] {} (p10);

\draw [->, bend left=0] (p10) to node [] {} (t14);
\draw [->, bend left=0] (t14) to node [] {} (p9);

\end{tikzpicture}
}
\end{minipage}


\vspace*{2mm}

\caption{
Another illustration of the proof of Theorem~\ref{ReverseOneTransition} with $\mu = e_1 e_2$,
where $e_2$, $e_t$ and $e_r$ share a node.
}

\label{FigProofReverseOneTransition2}

\end{figure}

\begin{corollary}[A characterization of reversibility for well-structured, live \PCMGineq{}]\label{CharRevPCMG}
Consider a live and well-structured \PCMGineq{} system $S=(N,M_0)$.
It is reversible iff it enables a T-sequence.
\end{corollary}

\begin{proof}
As recalled earlier, the left to right direction ($\Rightarrow$) is clear.
To prove the other direction ($\Leftarrow$), we reason by induction on the length $n$ of any feasible sequence $\sigma$, as follows.

Base case: $n=0$, $\sigma$ is the empty sequence, hence $M_0$ is trivially reachable.

Inductive case:
$n > 0$ and we suppose the claim to be true for $n-1$.
The sequence is of the form $\sigma = t \sigma_0$ and leads to a marking $M$.
The firing of $t$ leads to the system $(N,M_t)$, which is live and enables a T-sequence by Theorem~\ref{ReverseOneTransition}.
By the induction hypothesis, some sequence $\sigma_0'$ is feasible at $M$ that leads to $M_t$.
Thus, $M_0$ can be reached from $M$.
We proved the base and inductive cases,
hence the claim is true for each length $n$, thus any feasible sequence.
\end{proof}

Figure \ref{FigCamposModif} pictures\footnote{The system is inspired from Figure~21 in~\cite{HDM2016}.} a well-formed, non-reversible system allowing a T-sequence,
obtained by merging three pairs of places in a well-formed marked graph: the system obtained is not a \PCMGineq{}, 
since there does not exist any undirected graph describing its topology.
It has three shared places, whereas in each \PCMGineq{} each $N_i$ contains at most two shared places.
Since a \PCMGineq{} can be seen as the result of place-merging subsets of places in a non-connected marked graph, 
with additional constraints reducing the set of mergeable subsets,
this example can be considered as close to a \PCMGineq{}.
Thus, relaxing the definition of \PCMGineq{} easily leads to examples invalidating Corollary~\ref{CharRevPCMG}.
Another similar counter-example, with an isomorphic reachability graph and only two shared places, is pictured in Figure~\ref{NonRevTwoPlaces}.
Notice that both figures use the same MG, with a different layout. 

\begin{figure}[!ht]
\vspace*{1cm}

	
\centering

\begin{tikzpicture}[scale=1,mypetristyle]

\node (p0') at (0,2.2) [place,draw=black!60,fill=black!15] {};
\node (p0'') at (0,2.8) [place,draw=black!60,fill=black!15] {};
\node (p1') at (2.4,2) [place,draw=black!60,fill=black!15] {};
\node (p1'') at (2.4,2.55) [place,draw=black!60,fill=black!15] {};
\node (p2') at (1,0.5) [place,draw=black!60,fill=black!15] {};
\node (p2'') at (1,0) [place,draw=black!60,fill=black!15] {};
\node (p3) at (3.4,-0.3) [place] {};
\node (p4) at (3.9,-0.45) [place] {};
\node (p5) at (0.5,1.1) [place] {};
\node (p6) at (2.8,1.5) [place] {};

\node [anchor=east] at (p0'.west) {$p_0'$};
\node [anchor=east] at (p0''.west) {$p_0''$};
\node [anchor=north] at (p1'.south west) {$p_1'$};
\node [anchor=east] at (p1''.west) {$p_1''$};
\node [anchor=south west] at (p2'.east) {$p_2'$};
\node [anchor=west] at (p2''.east) {$p_2''$};
\node [anchor=east] at (p3.north west) {$p_3$};
\node [anchor=south] at (p4.north east) {$~p_4$};
\node [anchor=south east] at (p5.north) {$p_5$};
\node [anchor=north] at (p6.south east) {$~~p_6$};

\node (t0) at (3.4,2.8) [transition] {};
\node (t1) at (1,2) [transition] {};
\node (t2) at (0,-0.3) [transition] {};
\node (t3) at (2.4,0.5) [transition] {};

\node [anchor=west] at (t0.east) {$t_0$};
\node [anchor=north west] at (t1.south east) {$t_1$};
\node [anchor=east] at (t2.west) {$t_2$};
\node [anchor=west] at (t3.east) {$t_3$};

\draw [->,thick,bend left=0] (t2.north east) to node {} (p2'');
\draw [->,thick,bend left=0] (p2'') to node {} (t2.north east);

\draw [->,thick] (p2') to node {} (t1);
\draw [->,thick] (t3) to node {} (p2');
\draw [->,thick] (p1') to node {} (t3);
\draw [->,thick] (t1) to node {} (p1');
\draw [->,thick] (p0') to node {} (t1.north west);

\draw [->,thick,bend right=20] (t0.west) to node {} (p1'');
\draw [->,thick,bend left=20] (p1'') to node {} (t0.west);

\draw [->,thick, bend right=0] (t0.north west) to node {} (p0''.north east);
\draw [->,thick] (t2.north west) to node {} (p0'.south west);

\draw [->,thick] (p0'') .. controls  (1.4,2.4) .. (t3.north west);

\draw [->,thick,bend left=0] (t0) to node {} (p3);
\draw [->,thick,bend left=0] (p3) to node {} (t2);

\draw [->,thick,bend right=0] (t2.south east) to node {} (p4.south west);
\draw [->,thick,bend right=0] (p4) to node {} (t0.south east);

\draw [->,thick,bend right=0] (t1.south west) to node {} (p5);
\draw [->,thick,bend right=0] (p5) to node {} (t2.north);

\draw [->,thick,bend right=0] (t3.north east) to node {} (p6);
\draw [->,thick,bend right=0] (p6) to node {} (t0.south west);

\end{tikzpicture}
%
%
\begin{tikzpicture}[scale=1,mypetristyle]

\node (p0) at (0,2.8) [place,tokens=3] {};
\node (p1) at (2.4,2) [place,tokens=1] {};
\node (p2) at (1,0.5) [place,tokens=0] {};
\node (p3) at (3.4,-0.3) [place,tokens=1] {};
\node (p4) at (3.9,-0.45) [place] {};
\node (p5) at (0.5,1.1) [place,tokens=1] {};
\node (p6) at (2.8,1.5) [place] {};

\node [anchor=east] at (p0.west) {$p_0$};
\node [anchor=south east] at (p1.north west) {$p_1$};
\node [anchor=south west] at (p2.north east) {$p_2$};
\node [anchor=east] at (p3.north west) {$p_3$};
\node [anchor=south] at (p4.north east) {$~p_4$};
\node [anchor=south east] at (p5.north) {$p_5$};
\node [anchor=north] at (p6.south east) {$~~p_6$};

\node (t0) at (3.4,2.8) [transition] {};
\node (t1) at (1,2) [transition] {};
\node (t2) at (0,-0.3) [transition] {};
\node (t3) at (2.4,0.5) [transition] {};

\node [anchor=west] at (t0.east) {$t_0$};
\node [anchor=north west] at (t1.south east) {$t_1$};
\node [anchor=east] at (t2.west) {$t_2$};
\node [anchor=west] at (t3.east) {$t_3$};

\draw [->,thick,bend left=0] (t2.north east) to node {} (p2);
\draw [->,thick,bend left=0] (p2) to node {} (t2.north east);

\draw [->,thick] (p2) to node {} (t1);
\draw [->,thick] (t3) to node {} (p2);
\draw [->,thick] (p1) to node {} (t3);
\draw [->,thick] (t1) to node {} (p1);
\draw [->,thick] (p0) to node {} (t1.north west);

\draw [->,thick,bend right=20] (t0.west) to node {} (p1.north);
\draw [->,thick,bend left=20] (p1.north) to node {} (t0.west);

\draw [->,thick, bend right=0] (t0.north west) to node {} (p0.north east);
\draw [->,thick] (t2.north west) to node {} (p0.south west);

\draw [->,thick] (p0) .. controls  (1.3,2.5) .. (t3.north west);

\draw [->,thick,bend left=0] (t0) to node {} (p3);
\draw [->,thick,bend left=0] (p3) to node {} (t2);

\draw [->,thick,bend right=0] (t2.south east) to node {} (p4.south west);
\draw [->,thick,bend right=0] (p4) to node {} (t0.south east);

\draw [->,thick,bend right=0] (t1.south west) to node {} (p5);
\draw [->,thick,bend right=0] (p5) to node {} (t2.north);

\draw [->,thick,bend right=0] (t3.north east) to node {} (p6);
\draw [->,thick,bend right=0] (p6) to node {} (t0.south west);

\end{tikzpicture}
\raisebox{0mm}{
\begin{tikzpicture}[scale=0.95,mypetristyle]
\node[ltsNode,label=above:$s_0$,minimum width=7pt](s0)at(1,2.3){};
\node[ltsNode,label=above:$s_1$](s1)at(2.5,2.3){};
\node[ltsNode,label=right:$s_2$](s2)at(1,1){};
\node[ltsNode,label=left:$s_3$](s3)at(2.5,1){};
\node[ltsNode,label=below:$s_4$](s4)at(1,-0.3){};
\node[ltsNode,label=below:$s_5$](s5)at(2.5,-0.3){};
\node[ltsNode,label=left:$s_6$](s6)at(0.5,1){};
\node[ltsNode,label=right:$s_7$](s7)at(3,1){};
\draw[-{>[scale=2.5,length=2,width=2]}](s0)to node[above]{$t_3$}(s1);
\draw[-{>[scale=2.5,length=2,width=2]}](s1)to node[left]{$t_1$}(s3);
\draw[-{>[scale=2.5,length=2,width=2]}](s3)to node[left]{$t_0$}(s5);
\draw[-{>[scale=2.5,length=2,width=2]}](s5)to node[below right]{$t_3$}(s7);
\draw[-{>[scale=2.5,length=2,width=2]}](s0)to node[right]{$t_0$}(s2);
\draw[-{>[scale=2.5,length=2,width=2]}](s2)to node[right]{$t_3$}(s4);
\draw[-{>[scale=2.5,length=2,width=2]}](s4)to node[above]{$t_1$}(s5);
\draw[-{>[scale=2.5,length=2,width=2]}](s4)to node[below left]{$t_2$}(s6);
\draw[-{>[scale=2.5,length=2,width=2]}](s6)to node[above left]{$t_1$}(s0);
\draw[-{>[scale=2.5,length=2,width=2]}](s7)to node[above right]{$t_2$}(s1);
\end{tikzpicture}
}
	  

\vspace*{3mm}

\caption{
On the left, a well-formed marked graph is pictured with highlighted subsets of places $\{p_0',p_0''\}$, $\{p_1',p_1''\}$ and $\{p_2',p_2''\}$;
place-merging each pair leads to the underlying net of the system in the middle.
The latter is $1$-conservative, consistent, well-formed, live, non-reversible and enables the T-sequence $t_3 \, t_2 \, t_1 \, t_0$. 
Its reachability graph is depicted on the right. 
}

\label{FigCamposModif}

\end{figure}
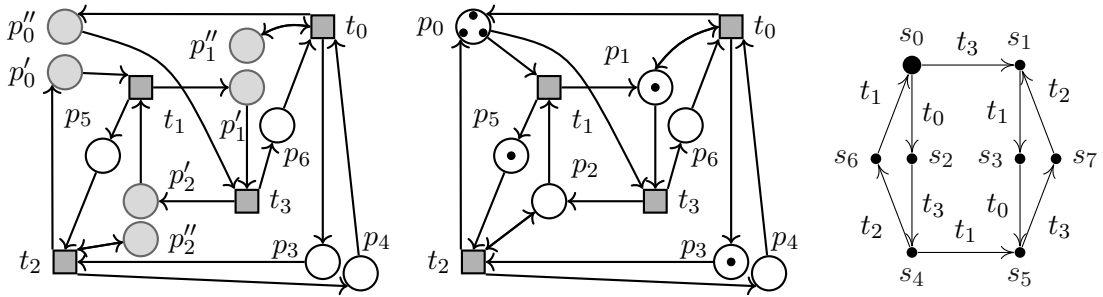

\begin{figure}[!ht]
\vspace*{1cm}
\centering

\hspace*{2mm}
\begin{tikzpicture}[scale=0.9,mypetristyle]

\node (p0) at (-0.5,2.8) [place,tokens=0] {};
\node (p1) at (1.9,2) [place,tokens=0,draw=black!60,fill=black!15] {};
\node (p1') at (2.5,2.15) [place,tokens=0,draw=black!60,fill=black!15] {};
\node (p2) at (1.1,0.5) [place,tokens=0,draw=black!60,fill=black!15] {};
\node (p2') at (0.5,0.35) [place,tokens=0,draw=black!60,fill=black!15] {};
\node (p3) at (3.5,-0.3) [place,tokens=0] {};
\node (p4) at (0.3,1) [place,tokens=0] {};
\node (p5) at (0,2) [place,tokens=0] {};
\node (p6) at (3,0.5) [place,tokens=0] {};
\node (p7) at (2.6,1.5) [place,tokens=0] {};

\node [anchor=south] at (p0.north) {$p_0$};
\node [anchor=south] at (p1.north) {$p_1$};
\node [anchor=south] at (p1'.north) {$p_1'$};
\node [anchor=north] at (p2.south) {$p_2$};
\node [anchor=north] at (p2'.south) {$p_2'$};
\node [anchor=north] at (p3.south) {$p_3$};
\node [anchor=south] at (p4.north) {$p_4$};
\node [anchor=south] at (p5.north) {$p_5$};
\node [anchor=north] at (p6.south) {$p_6$};
\node [anchor=north] at (p7.south) {$~~~p_7$};

\node (t0) at (3.5,2.8) [transition] {};
\node (t1) at (1.1,2) [transition] {};
\node (t2) at (-0.5,-0.3) [transition] {};
\node (t3) at (1.9,0.5) [transition] {};

\node [anchor=south] at (t0.north) {$t_0$};
\node [anchor=south] at (t1.north) {$t_1$};
\node [anchor=north] at (t2.south) {$t_2$};
\node [anchor=north] at (t3.south) {$t_3$};

\draw [->,thick,bend left=0] (t2) to node {} (p2');
\draw [->,thick,bend left=0] (p2') to node {} (t2);

\draw [->,thick] (p2) to node {} (t1);
\draw [->,thick] (t3) to node {} (p2);
\draw [->,thick] (p1) to node {} (t3);
\draw [->,thick] (t1) to node {} (p1);

\draw [->,thick] (t0) to node {} (p7);
\draw [->,thick] (p7) to node {} (t3);

\draw [->,thick,bend right=0] (t0.west) to node {} (p1');
\draw [->,thick,bend left=0] (p1') to node {} (t0.west);

\draw [->,thick, bend right=0] (p0.north east) to node {} (t0.north west);
\draw [->,thick] (t2.north west) to node {} (p0.south west);


\draw [->,thick,bend left=0] (t0.south east) to node {} (p3.north east);
\draw [->,thick,bend left=0] (p3.south west) to node {} (t2.south east);

\draw [->,thick] (t2) to node {} (p4);
\draw [->,thick] (p4) to node {} (t1);

\draw [->,thick,bend right=0] (t1) to node {} (p5);
\draw [->,thick,bend right=0] (p5) to node {} (t2.north);

\draw [->,thick,bend right=0] (t3) to node {} (p6);
\draw [->,thick,bend right=0] (p6) to node {} (t0);

\end{tikzpicture}
\hspace*{4mm}
\begin{tikzpicture}[scale=0.9,mypetristyle]

\node (p0) at (-0.5,2.8) [place,tokens=1] {};
\node (p1) at (1.9,2) [place,tokens=1] {};
\node (p2) at (1.1,0.5) [place,tokens=0] {};
\node (p3) at (3.5,-0.3) [place,tokens=0] {};
\node (p4) at (0.3,1) [place,tokens=1] {};
\node (p5) at (0,2) [place,tokens=1] {};
\node (p6) at (3,0.5) [place,tokens=1] {};
\node (p7) at (2.6,1.5) [place,tokens=1] {};

\node [anchor=south] at (p0.north) {$p_0$};
\node [anchor=south] at (p1.north) {$p_1$};
\node [anchor=north] at (p2.south) {$p_2$};
\node [anchor=north] at (p3.south) {$p_3$};
\node [anchor=south] at (p4.north) {$p_4$};
\node [anchor=south] at (p5.north) {$p_5$};
\node [anchor=north] at (p6.south) {$p_6$};
\node [anchor=north] at (p7.south) {$~~~p_7$};

\node (t0) at (3.5,2.8) [transition] {};
\node (t1) at (1.1,2) [transition] {};
\node (t2) at (-0.5,-0.3) [transition] {};
\node (t3) at (1.9,0.5) [transition] {};

\node [anchor=south] at (t0.north) {$t_0$};
\node [anchor=south] at (t1.north) {$t_1$};
\node [anchor=north] at (t2.south) {$t_2$};
\node [anchor=north] at (t3.south) {$t_3$};

\draw [->,thick,bend left=0] (t2) to node {} (p2);
\draw [->,thick,bend left=0] (p2) to node {} (t2);

\draw [->,thick] (p2) to node {} (t1);
\draw [->,thick] (t3) to node {} (p2);
\draw [->,thick] (p1) to node {} (t3);
\draw [->,thick] (t1) to node {} (p1);

\draw [->,thick] (t0) to node {} (p7);
\draw [->,thick] (p7) to node {} (t3);

\draw [->,thick,bend right=0] (t0.west) to node {} (p1);
\draw [->,thick,bend left=0] (p1) to node {} (t0.west);

\draw [->,thick, bend right=0] (p0.north east) to node {} (t0.north west);
\draw [->,thick] (t2.north west) to node {} (p0.south west);


\draw [->,thick,bend left=0] (t0.south east) to node {} (p3.north east);
\draw [->,thick,bend left=0] (p3.south west) to node {} (t2.south east);

\draw [->,thick] (t2) to node {} (p4);
\draw [->,thick] (p4) to node {} (t1);

\draw [->,thick,bend right=0] (t1) to node {} (p5);
\draw [->,thick,bend right=0] (p5) to node {} (t2.north);

\draw [->,thick,bend right=0] (t3) to node {} (p6);
\draw [->,thick,bend right=0] (p6) to node {} (t0);

\end{tikzpicture}
%
%
\raisebox{3mm}{
\begin{tikzpicture}[scale=0.95,mypetristyle]
\node[ltsNode,label=above:$s_0$,minimum width=7pt](s0)at(1,2.3){};
\node[ltsNode,label=above:$s_1$](s1)at(2.5,2.3){};
\node[ltsNode,label=right:$s_2$](s2)at(1,1){};
\node[ltsNode,label=left:$s_3$](s3)at(2.5,1){};
\node[ltsNode,label=below:$s_4$](s4)at(1,-0.3){};
\node[ltsNode,label=below:$s_5$](s5)at(2.5,-0.3){};
\node[ltsNode,label=left:$s_6$](s6)at(0.5,1){};
\node[ltsNode,label=right:$s_7$](s7)at(3,1){};
\draw[-{>[scale=2.5,length=2,width=2]}](s0)to node[above]{$t_3$}(s1);
\draw[-{>[scale=2.5,length=2,width=2]}](s1)to node[left]{$t_1$}(s3);
\draw[-{>[scale=2.5,length=2,width=2]}](s3)to node[left]{$t_0$}(s5);
\draw[-{>[scale=2.5,length=2,width=2]}](s5)to node[below right]{$t_3$}(s7);
\draw[-{>[scale=2.5,length=2,width=2]}](s0)to node[right]{$t_0$}(s2);
\draw[-{>[scale=2.5,length=2,width=2]}](s2)to node[right]{$t_3$}(s4);
\draw[-{>[scale=2.5,length=2,width=2]}](s4)to node[above]{$t_1$}(s5);
\draw[-{>[scale=2.5,length=2,width=2]}](s4)to node[below left]{$t_2$}(s6);
\draw[-{>[scale=2.5,length=2,width=2]}](s6)to node[above left]{$t_1$}(s0);
\draw[-{>[scale=2.5,length=2,width=2]}](s7)to node[above right]{$t_2$}(s1);
\end{tikzpicture}
}
\vspace*{4mm}

\caption{On the left, a unit-weighted, well-formed MG. Place-merging its subsets $\{p_1,p_1'\}$ and $\{p_2,p_2'\}$
leads to the system in the middle, which is unit-weighted, live, structurally bounded with only two shared places, namely $p_1$ and $p_2$.
The system obtained, with the initial marking pictured, enables the T-sequence $t_0 \, t_3 \, t_2 \, t_1$ but is not reversible.
On the right, its non strongly connected reachability graph is given, with initial state $s_0$.
}
\label{NonRevTwoPlaces}
\vspace*{1cm}
\end{figure}
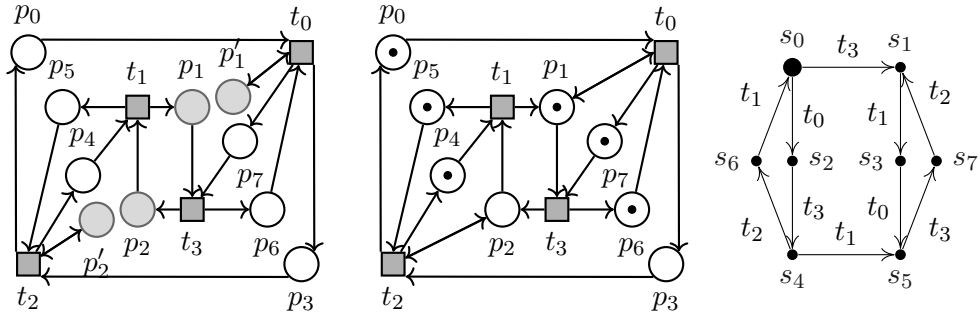

\subsection{Reversibility and the PR-R equality in the acyclic case} 

\vspace*{2mm}

We obtain next theorem.

\begin{theorem}[Directedness and reversibility in a subclass of \PCMGineq{}]\label{DirectedAcyclicPCMG}
Consider a well-structured and live \PCMGineq{} $S=(N,M_0)$ obtained from an acyclic, undirected and connected graph $G$. 
Then $S$ is reversible and fulfills the PR-R equality.
\end{theorem}

\begin{proof}
The proof is illustrated in Figure~\ref{FigProofDirectedAcyclicPCMG}.
We reason by induction on the number $n$ of shared places.\\
$-$ Base case: $n=0$, $S$ is a live and bounded marked graph, 
and Proposition~\ref{RealizableTvectors} applies.\\
$-$ Inductive case: $n>0$. We suppose the claim to be true for $n-1$.
Since $G$ is acyclic, consider an edge $e$ of $G$ having only one shared extremity;
denote by $p$ the shared place of $S$ associated to this extremity.
Denote by $S_e$ the T-subsystem $(N_e,\projection{M_0}{P_e})$ of $S$, where $N_e =(P_e,T_e,W_e)$
is the marked graph associated to $e$, $N_e = \beta(e)$,
and denote by $Y$ a T-vector such that $M = M_0 + I \cdot Y$.
The transition set $T_e$ of $N_e$ is disjoint from the one associated to $G - e$, by definition of \PCMGineq{}.
In the following, for any \PCMGineq{} system $S$ obtained from a graph $G$, for any subgraph $G'$ of $G$,
we denote by $\projection{S}{G'}$ the T-subsystem of $S$ corresponding to $G'$
and
by $\projection{Y}{G'}$ the projection of the T-vector $Y$ on the set of transitions associated to $G'$.

Fire a finite sequence $\sigma_e$ in $S_e$ leading to a marking $M'$ in $R(S)$
that maximizes the amount of tokens in $p$, 
i.e. such that $p$ is $M'(p)$-bounded in $S_e$.    
Since $(N_e,\projection{M'}{P_e})$ cannot produce additional tokens in $p$,
since $S$ is live and since $\projection{(N,M')}{G - e}$ is well-structured,
Theorem~\ref{LiveWellPCMG} applies and
we deduce that $\projection{(N,M')}{G - e}$ is a live, well-structured \PCMGineq{}. 
Thus, the inductive hypothesis applies to the latter:
$\projection{(N,M')}{G - e}$ is also reversible
and every marking that is potentially reachable in $\projection{(N,M')}{G - e}$ is reachable in it.

We show first that some T-sequence is feasible in $S$, allowing to apply Theorem~\ref{DirectedAcyclicPCMG}.
At each marking $M''$ reachable in $S$, if $M''(p) \ge 1$ then $(N_e,\projection{M''}{P_e})$ is live and reversible.
By liveness and reversibility of $\projection{(N,M')}{G - e}$, 
the latter enables a sequence $\tau$ that visits all transitions in $\projection{N}{G - e}$
and reaches a marking $M_p$ such that $M_p(p) \ge 1$;
from the above, the marked graph T-subsystem $(N_e,\projection{M_p}{P_e})$
enables a $\mathrm{T_e}$-sequence $\tau_e$, leading back to $M_p$;
then a sequence $\tau'$ leads back to $M'$.
We deduce that $\alpha = \tau \tau_e \tau'$ is a T-sequence feasible in $(N,M')$.
Now, either $M'(p) > 0$ so that $S_e$ is reversible and $M_0$ can be reached trivially from $M'$,
or $M'(p) = 0$ so that $\alpha$ is also feasible in $S$,
the intermediate marking $M'$ being replaced by $M_0$ in the reasoning above.
In both cases, applying Theorem~\ref{DirectedAcyclicPCMG}, $S$ is reversible.

Now, let us show that $M$ is reachable in $S$. Since $M$ is a marking, we have $M(p) \ge 0$.
By definition of $M'$, $M'(p) \ge M(p)$.
Let us denote by $T_{G-e}$ the set of transitions of $\projection{N}{G-e}$.
We define the T-vector $Y'$ as follows: for each transition $t$, if $t$ belongs to $T_{G-e}$ then $Y'(t) = Y(t)$, otherwise $Y'(t) = 0$.
Then, $M_{Y'} = M' + I \cdot Y'$ is a marking potentially reachable in $(N,M')$, 
and the marking $\projection{M_{Y'}}{G-e}$ is potentially reachable in $\projection{(N,M')}{G - e}$,
thus is reachable in the latter. 

We have two cases: either $M_{Y'}(p) > 0$ or $M_{Y'}(p) = 0$.
Let us define the T-vector $Z = k \cdot \one^{T_e} - \Parikh(\sigma_e) + \projection{Y}{e}$,
where $k$ is the smallest positive integer $k$ such that $k \cdot \one^{T_e} \ge \Parikh(\sigma_e)$.

In the first case, $(N_e, \projection{M_{Y'}}{e})$ is a live, well-formed MG T-subsystem of $(N,M_{Y'})$,
thus
Proposition~\ref{RealizableTvectors} applies: since $M = M' + I \cdot Z$, some sequence with Parikh vector $Z$
is feasible in $(N_e,\projection{M'}{e})$, thus also in $(N,M')$ and leads to $M$, hence the claim.

In the second case, a sequence $\sigma_r$ is feasible in $(N,M_{Y'})$, leading to a marking $M_r$ such that $M_r(p) = 1$.
Since $\projection{(N,M_{Y'})}{G - e}$ is reversible, a sequence $\sigma_r'$ is feasible in $\projection{(N,M_r)}{G - e}$,
thus also in $(N,M_r)$, that leads back to $M_{Y'}$.
Let us define $M_r' = M_r + I \cdot Z$, then $(N_e,\projection{M_r'}{e})$ is live, Proposition~\ref{RealizableTvectors} applies
and $\projection{M_r'}{e}$ is reachable in the latter T-subsystem, hence $M_r'$ is reachable from $(N,M_r)$.
Moreover, $M_r'(p) \ge 1$ since otherwise we would have $M(p) < 0$, which is impossible.
Consequently, since the sequence $\sigma_r'$ is feasible in $(N,M_r)$ 
and 
since $\projection{M_r'}{G-e} \ge \projection{M_r}{G-e}$,
it is also feasible in $(N,M_r')$ and leads to $M$, which is thus reachable from $M_0$.

We proved the property to be true for every number $n$ of shared places. We get the claim.
\end{proof}

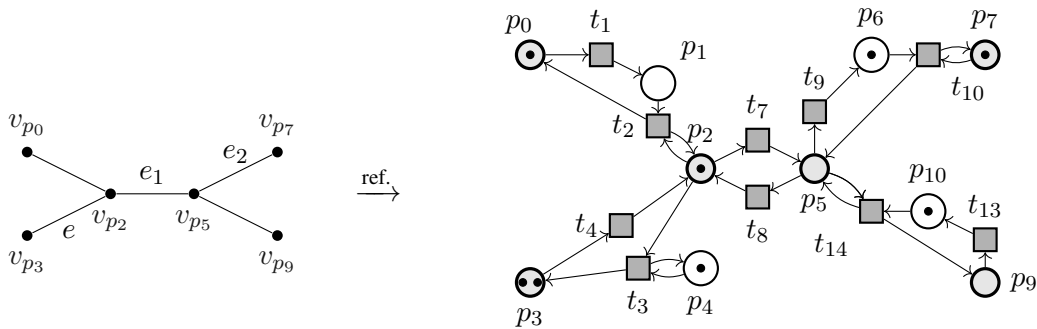
\begin{figure}[!ht]
 \centering
 
\begin{minipage}{0.33\linewidth}
\centering
\begin{tikzpicture}[scale=0.55,mypetristyle]
\node[ltsNode,label=below:$v_{p_3}$](s0)at(0,0){};
\node[ltsNode,label=above:$v_{p_0}$](s1)at(0,2){};
\node[ltsNode,label=below:$v_{p_2}$](s2)at(2,1){};
\node[ltsNode,label=below:$v_{p_5}$](s3)at(4,1){};
\node[ltsNode,label=above:$v_{p_7}$](s4)at(6,2){};
\node[ltsNode,label=below:$v_{p_9}$](s5)at(6,0){};

\draw[-](s1)to node[]{}(s2);
\draw[-](s0)to node[below]{$e$}(s2);
\draw[-](s2)to node[above]{$e_1$}(s3);
\draw[-](s3)to node[above]{$e_2$}(s4);
\draw[-](s5)to node[]{}(s3);

\end{tikzpicture}

\end{minipage}
$\xrightarrow{\textrm{ref.}}{}{}$
\begin{minipage}{0.6\linewidth}
\centering
\raisebox{7mm}{
\begin{tikzpicture}[scale=0.75,mypetristyle]

\node (p0) at (0,4) [petriNode,tokens=1] {};
\node (p1) at (2.25,3.5) [place] {};
\node (p2) at (3,2) [petriNode,tokens=1] {};
\node (p3) at (0,0) [petriNode,tokens=2] {};
\node (p4) at (3,0.25) [place,tokens=1] {};
\node (p5) at (5,2) [petriNode] {};
\node (p6) at (6,4) [place,tokens=1] {};
\node (p7) at (8,4) [petriNode,tokens=1] {};
\node (p9) at (8,0) [petriNode] {};
\node (p10) at (7,1.25) [place,tokens=1] {};

\node [anchor=south] at (p0.north west) {$p_0$};
\node [anchor=south west] at (p1.north east) {$p_1$};
\node [anchor=south] at (p2.north) {$p_2$};
\node [anchor=north] at (p3.south) {$p_3$};
\node [anchor=north] at (p4.south) {$p_4$};
\node [anchor=north] at (p5.south) {$p_5$};
\node [anchor=south] at (p6.north) {$p_6$};
\node [anchor=south] at (p7.north) {$p_7$};
\node [anchor=west] at (p9.east) {$p_9$};
\node [anchor=south] at (p10.north) {$p_{10}$};

\node (t1) at (1.25,4) [transition] {};
\node (t2) at (2.25,2.75) [transition] {};
\node (t3) at (1.9,0.25) [transition] {};
\node (t4) at (1.6,1) [transition] {};
\node (t7) at (4,2.5) [transition] {};
\node (t8) at (4,1.5) [transition] {};
\node (t9) at (5,3) [transition] {};
\node (t10) at (7,4) [transition] {};
\node (t13) at (8,0.75) [transition] {};
\node (t14) at (6,1.25) [transition] {};

\node [anchor=south] at (t1.north) {$t_1$};
\node [anchor=east] at (t2.west) {$t_2$};
\node [anchor=north] at (t3.south) {$t_3$};
\node [anchor=east] at (t4.west) {$t_4$};
\node [anchor=south] at (t7.north) {$t_7$};
\node [anchor=north] at (t8.south) {$t_8$};
\node [anchor=south] at (t9.north) {$t_9$};
\node [anchor=north west] at (t10.south east) {$t_{10}$};
\node [anchor=south] at (t13.north) {$t_{13}$};
\node [anchor=north east] at (t14.south west) {$t_{14}$};

\draw [->, bend left=0] (p0) to node [] {} (t1);
\draw [->, bend left=0] (t1) to node [] {} (p1);
\draw [->, bend left=0] (p1) to node [] {} (t2);
\draw [->, bend left=0] (t2) to node [] {} (p0);
\draw [->, bend left=20] (t2) to node [] {} (p2);
\draw [->, bend left=20] (p2) to node [] {} (t2);
\draw [->, bend left=0] (p3) to node [] {} (t4);
\draw [->, bend left=0] (t4) to node [] {} (p2);
\draw [->, bend left=0] (p2) to node [] {} (t3);
\draw [->, bend left=0] (t3) to node [] {} (p3);



\draw [->, bend left=0] (p2) to node [] {} (t7);
\draw [->, bend left=0] (t7) to node [] {} (p5);
\draw [->, bend left=0] (p5) to node [] {} (t8);
\draw [->, bend left=0] (t8) to node [] {} (p2);

\draw [->, bend left=0] (p5) to node [] {} (t9);
\draw [->, bend left=0] (t9) to node [] {} (p6);

\draw [->, bend left=0] (p6) to node [] {} (t10);
\draw [->, bend left=0] (t10) to node [] {} (p5);
\draw [->, bend left=20] (p5) to node [] {} (t14);
\draw [->, bend left=20] (t14) to node [] {} (p5);



\draw [->, bend left=20] (p5) to node [] {} (t14);

\draw [->, bend left=20] (p7) to node [] {} (t10);
\draw [->, bend left=20] (t10) to node [] {} (p7);

\draw [->, bend left=0] (p9) to node [] {} (t13);
\draw [->, bend left=0] (t13) to node [] {} (p10);

\draw [->, bend left=0] (p10) to node [] {} (t14);
\draw [->, bend left=0] (t14) to node [] {} (p9);

\draw [->, bend left=20] (p4) to node [] {} (t3);
\draw [->, bend left=20] (t3) to node [] {} (p4);

\end{tikzpicture}
}
\end{minipage}



\caption{
Illustration of the proof of Theorem~\ref{DirectedAcyclicPCMG}.
The graph $G$ on the left is acyclic and is labeled with place names.
The \PCMGineq{} on the right is obtained from $G$.
}

\label{FigProofDirectedAcyclicPCMG}
\vspace*{1cm}
\end{figure}

This theorem is no more true if homogeneous weights are allowed, as shown in Figure~\ref{SsystemNonRev}.\\

\begin{figure}[!h]
 \vspace*{1cm}

\centering

\raisebox{6mm}{
\begin{tikzpicture}[scale=1,mypetristyle]
\node[ltsNode](n0)at(0,0){};
\node[ltsNode](n1)at(1.5,0){};
\node[ltsNode](n2)at(3,0){};
\draw[-](n0)to node[above]{}(n1);
\draw[-](n1)to node[above]{}(n2);
\node [anchor=north] at (n0.south) {$p_0$};
\node [anchor=north] at (n1.south) {$p_1$};
\node [anchor=north] at (n2.south) {$p_2$};
\end{tikzpicture}
}
\hspace*{1cm}
\begin{tikzpicture}[mypetristyle,scale=1]

\node (p0) at (0,0) [petriNode,tokens=1] {};
\node (p1) at (2,0) [petriNode] {};
\node (p2) at (4,0) [petriNode,tokens=1] {};

\node [anchor=east] at (p0.west) {$v_{p_0}$};
\node [anchor=north] at (p1.south) {$~~v_{p_1}$};
\node [anchor=west] at (p2.east) {$v_{p_2}$};

\node (t1) at (1,0.5) [transition,thick] {};
\node (t2) at (1,-0.5) [transition,thick] {};
\node (t3) at (3,0.5) [transition,thick] {};
\node (t4) at (3,-0.5) [transition,thick] {};

\node [anchor=south] at (t1.north) {$t_1$};
\node [anchor=north] at (t2.south) {$t_2$};
\node [anchor=south] at (t3.north) {$t_3$};
\node [anchor=north] at (t4.south) {$t_4$};

\draw [->,thick,bend right=20] (p1) to node [above] {$2$} (t1);
\draw [->,thick,bend left=20] (p1) to node [above] {$2$} (t3);
\draw [->,thick,bend right=20] (p0) to node [below] {} (t2);
\draw [->,thick,bend left=20] (p2) to node [below] {} (t4);

\draw [->,thick,bend right=20] (t1) to node [above] {$2$} (p0);
\draw [->,thick,bend left=20] (t3) to node [above] {} (p2);
\draw [->,thick,bend right=20] (t2) to node [below] {} (p1);
\draw [->,thick,bend left=20] (t4) to node [below] {} (p1);

\end{tikzpicture}


\caption{On the right, a weighted, homogeneous state machine that could be seen as a weighted, live and well-structured \PCMGineq{},
obtained from the graph on the left.
It is live, but non reversible; from some reachable marking (e.g. after a single firing of $t_2$), 
the initial marking is potentially reachable but not reachable.
}

\label{SsystemNonRev}

\end{figure}
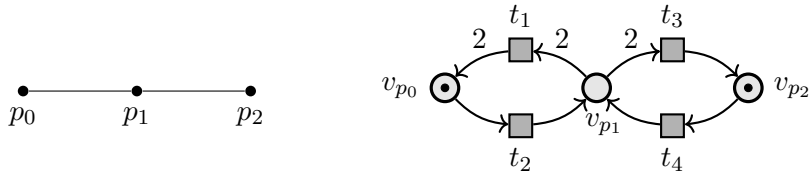

We are now able to derive next corollary.

\begin{corollary}[Property $\mathcal{R}$ and PR-R equality in acyclic, live and well-structured \PCMGineq{}]\label{PCMGacyclicPRR}
Consider a live and well-structured \PCMGineq{} $S$ obtained from an acyclic undirected graph. 
Then $S$ fulfills $\mathcal{R}$ and the PR-R equality.
\end{corollary}

\begin{proof}
We get the claim by combining Theorem~\ref{DirectedAcyclicPCMG} with Lemma~\ref{LRPRrevReverse}.
\end{proof}

By Corollary~\ref{PCMGlivePTIME}, the liveness of a well-structured \PCMGineq{} obtained from an acyclic graph
can be checked in polynomial-time.
Then, Corollary~\ref{PCMGacyclicPRR} derives property $\mathcal{R}$ and the PR-R equality.

\section{Related work}\label{SecRelatedWork}

Results connected to potential reachability in WMG, \WMGineq{} and larger classes, 
together with the behavioral properties investigated in this paper, have been developed notably 
in~\cite{WTS92,STECS,tcs97,DH18,DH19FI}.

As far as we know, the only work dedicated to the H$1$S class is~\cite{ArxivSSP20}.

AMG and their properties, such as liveness and reversibility, together with compositional methods,
have been investigated in \cite{ChuXie97,cheung2008augmented,huang2003property}.

Concerning refinement, synthesis, composition (node merging, fusion...) and abstraction techniques,
the previous works closest to our notion introduced for \PCMGineq{}
are \cite{Stepwise83,WellStruct98,huang2003property,Jiao2004,Surapprox2005,JiaoComposition2005,Knitting06,JiaoSharing2008}.
Numerous other studies provide algorithms for synthetizing, 
from a given labeled transition system, a Petri net with isomorphic reachability graph,
see e.g. \cite{DH18} for the synthesis of \WMGineq{}; such techniques are out of scope for this paper.
%

Other classes with shared places, such as S$^4$PR, PC$^2$R and L-S$^3$PR, which do not contain the H$1$S-\WMGineq{}, the AMG nor the \PCMGineq{}, 
have also been extensively studied in~\cite{WellStruct98,LenderProc2006,LopezGrao2013}
and benefit from strong properties related to reachability, notably to directedness, deadlockability, liveness and reversibility.
%
%
In~\cite{DSSP98,SCECS96}, DSSP and $\{$SC$\}^*$ECS systems, which allow weights and shared places in a restricted fashion,
benefit from structural and behavioral properties.
Generalizations of AC nets have been studied in~\cite{BeyondAC1998}.

\section{Conclusions and Perspectives}\label{SecConclu}

The reachability problem has attracted a lot of attention since the introduction of Petri nets in the~60's by Carl Adam Petri.
It is indeed a central model-checking problem that reduces to numerous other fundamental ones.
Bounds on its complexity have been obtained and refined over the years.
Recently, a non-elementary lower bound has been uncovered.

In this paper, for weighted Petri nets, we provided several sufficient conditions ensuring the PR-R equality,
i.e. the equality of the set of reachable markings and the set of potentially reachable ones,
the latter being described by the solutions of the state equation.
When this equality is fulfilled, asking for the reachability of a marking reduces to solving the state equation for this marking,
allowing to use integer linear programming.
This avoids a brute-force analysis of the state space and trims down the complexity to NP.
The main contributions of this paper are summarized as follows.

First, we developed results connecting some behavioral properties of a system to the same properties in the reverse system.
Using the notion of directedness,
we provided a general condition ensuring the PR-R equality in weighted Petri nets,
based on liveness, reversibility in the system and its reverse.
We deduced a sufficient condition of PR-R equality for homogeneous free-choice nets.

Then, we delineated several subclasses of weighted Petri nets generalizing the marked graphs,
which have been extensively studied and fruitfully used in the modeling and analysis of various real-world applications.
We recalled some use-cases of these larger classes,
extracted from previous studies and the Petri Net Model Checking Contest database.
In these classes, we proposed several sufficient conditions for PR-R equality,
based on the reversibility of the systems and their reverse, and on liveness.

We also recalled and gathered methods for checking the sufficient conditions,
notably relating the structure to liveness and reversibility,
and highlighted cases in which these methods are more efficient.

As a perspective, we believe that our methods may be extended to other classes of Petri nets, notably modular ones.
Also, more efficient methods checking the assumptions of the conditions,
such as liveness, reversibility and properties of the siphons, 
might be uncovered for the classes of our study and more expressive ones.

\bibliographystyle{fundam}
\bibliography{IC19}

\begin{thebibliography}{10}
\providecommand{\url}[1]{\texttt{#1}}
\providecommand{\urlprefix}{URL }
\expandafter\ifx\csname urlstyle\endcsname\relax
  \providecommand{\doi}[1]{doi:\discretionary{}{}{}#1}\else
  \providecommand{\doi}{doi:\discretionary{}{}{}\begingroup
  \urlstyle{rm}\Url}\fi
\providecommand{\eprint}[2][]{\url{#2}}

\bibitem{DesEsp}
Desel J, Esparza J.
\newblock {F}ree {C}hoice {P}etri {N}ets, volume~40 of \emph{Cambridge Tracts
  in Theoretical Computer Science}.
\newblock Cambridge University Press, New York, USA, 1995.

\bibitem{STECS}
Teruel E, Silva M.
\newblock Structure theory of {E}qual {C}onflict systems.
\newblock \emph{Theoretical Computer Science}, 1996.
\newblock \textbf{153}(1{\&}2):271--300.
\newblock \urlprefix\url{https://doi.org/10.1016/0304-3975(95)00124-7}.

\bibitem{LAT98}
Silva M, Teruel E, Colom JM.
\newblock Linear algebraic and linear programming techniques for the analysis
  of place/transition net systems.
\newblock In: Reisig W, Rozenberg G (eds.), Lectures on {P}etri Nets I: Basic
  Models, volume 1491 of \emph{LNCS}. 1998 pp. 309--373.
\newblock \urlprefix\url{https://doi.org/10.1007/3-540-65306-6\_19}.

\bibitem{HDM2016}
Hujsa T, Delosme JM, Munier-Kordon A.
\newblock On Liveness and Reversibility of Equal-Conflict {P}etri Nets.
\newblock \emph{Fundamenta Informaticae}, 2016.
\newblock \textbf{146}(1):83--119.
\newblock \urlprefix\url{https://doi.org/10.3233/FI-2016-1376}.

\bibitem{HD2017}
Hujsa T, Devillers R.
\newblock On Liveness and Deadlockability in Subclasses of Weighted {P}etri
  Nets.
\newblock In: van~der Aalst W, Best E (eds.), Application and Theory of {P}etri
  Nets and Concurrency: 38th International Conference, {P}ETRI NETS 2017,
  Zaragoza, Spain, June 25--30, 2017, Proceedings. Springer International
  Publishing, Cham.
\newblock ISBN 978-3-319-57861-3, 2017 pp. 267--287.
\newblock \urlprefix\url{https://doi.org/10.1007/978-3-319-57861-3\_16}.

\bibitem{HD2018}
Hujsa T, Devillers R.
\newblock On Deadlockability, Liveness and Reversibility in Subclasses of
  Weighted {P}etri Nets.
\newblock \emph{Fundamenta Informaticae}, 2018.
\newblock \textbf{161}(4):383--421.
\newblock \doi{10.3233/FI-2018-1708}.
\newblock \urlprefix\url{https://doi.org/10.3233/FI-2018-1708}.

\bibitem{Lipton76}
Lipton RJ.
\newblock The reachability problem requires exponential space.
\newblock Technical Report~62, Department of Computer Science, Yale University,
  1976.
\newblock
  \urlprefix\url{http://cpsc.yale.edu/sites/default/files/files/tr63.pdf}.

\bibitem{ReachNonElementary2019}
Czerwi\'{n}ski W, Lasota S, Lazi\'{c} R, Leroux J, Mazowiecki F.
\newblock The Reachability Problem for {P}etri Nets is Not Elementary.
\newblock In: Proceedings of the 51st Annual ACM SIGACT Symposium on Theory of
  Computing, STOC 2019. ACM, New York, NY, USA.
\newblock ISBN 978-1-4503-6705-9, 2019 pp. 24--33.
\newblock \doi{10.1145/3313276.3316369}.
\newblock \urlprefix\url{http://doi.acm.org/10.1145/3313276.3316369}.

\bibitem{DH18}
Devillers R, Hujsa T.
\newblock Analysis and Synthesis of Weighted Marked Graph {P}etri Nets.
\newblock In: Khomenko V, Roux OH (eds.), Application and Theory of {P}etri
  Nets and Concurrency. Springer International Publishing, Cham.
\newblock ISBN 978-3-319-91268-4, 2018 pp. 19--39.

\bibitem{ArxivSSP20}
Hujsa T, Berthomieu B, Dal~Zilio S, Le~Botlan D.
\newblock On the {P}etri Nets with a Single Shared Place and Beyond.
\newblock \emph{CoRR}, 2020.
\newblock \textbf{abs/2005.04818}.
\newblock \eprint{2005.04818},
  \urlprefix\url{https://arxiv.org/abs/2005.04818}.

\bibitem{ChuXie97}
Chu F, Xie XL.
\newblock Deadlock analysis of {P}etri nets using siphons and mathematical
  programming.
\newblock \emph{IEEE Transactions on Robotics and Automation}, 1997.
\newblock \textbf{13}(6):793--804.
\newblock \doi{10.1109/70.650158}.

\bibitem{MDG71}
Commoner F, Holt A, Even S, Pnueli A.
\newblock Marked Directed Graphs.
\newblock \emph{J. Comput. Syst. Sci}, 1971.
\newblock \textbf{5}(5):511--523.
\newblock \urlprefix\url{https://doi.org/10.1016/S0022-0000(71)80013-2}.

\bibitem{Sauer2003}
Sauer N.
\newblock Marking Optimization of Weighted Marked Graphs.
\newblock \emph{Discrete Event Dynamic Systems}, 2003.
\newblock \textbf{13}(3):245--262.
\newblock \doi{10.1023/A:1024055724914}.
\newblock \urlprefix\url{https://doi.org/10.1023/A:1024055724914}.

\bibitem{March09}
Marchetti O, Munier-Kordon A.
\newblock {A} sufficient condition for the liveness of {W}eighted {E}vent
  {G}raphs.
\newblock \emph{{E}uropean {J}ournal of {O}perational {R}esearch}, 2009.
\newblock \textbf{197}(2):532--540.
\newblock \urlprefix\url{https://doi.org/10.1016/j.ejor.2008.07.037}.

\bibitem{WTS92}
Teruel E, Chrzastowski-Wachtel P, Colom JM, Silva M.
\newblock On weighted {T}-systems.
\newblock In: Jensen K (ed.), 13th International Conference on Application and
  Theory of {{P}etri} Nets and Concurrency ({ICATPN}), LNCS, volume 616.
  Springer, Berlin, Heidelberg, 1992 pp. 348--367.
\newblock \urlprefix\url{https://doi.org//10.1007/3-540-55676-1\_20}.

\bibitem{ErgoMonoT1991}
Campos J, Chiola G, Su{\'{a}}rez MS.
\newblock Ergodicity and Throughput Bounds of Petri Nets with Unique Consistent
  Firing Count Vector.
\newblock \emph{{IEEE} Trans. Software Eng.}, 1991.
\newblock \textbf{17}(2):117--125.
\newblock \doi{10.1109/32.67593}.
\newblock \urlprefix\url{https://doi.org/10.1109/32.67593}.

\bibitem{LeeM87}
Lee EA, Messerschmitt DG.
\newblock {S}ynchronous {D}ata {F}low.
\newblock \emph{Proceedings of the IEEE}, 1987.
\newblock \textbf{75}(9):1235--1245.

\bibitem{LM87a}
Lee EA, Messerschmidt DG.
\newblock Static scheduling of synchronous data flow programs for digital
  signal processing.
\newblock In: IEEE Transaction on Computers, C-36(1). 1987 pp. 24--35.
\newblock \urlprefix\url{https://doi.org/10.1109/TC.1987.5009446}.

\bibitem{Emb09}
Sriram S, Bhattacharyya SS.
\newblock Embedded multiprocessors: scheduling and synchronization.
\newblock Signal Processing and Communications. CRC Press, 2009.

\bibitem{Pino95}
Pino JL, Bhattacharyya SS, Lee EA.
\newblock A hierarchical multiprocessor scheduling framework for synchronous
  dataflow graphs.
\newblock Technical report, University of California, Berkeley, 1995.

\bibitem{EHW18}
Best E, Hujsa T, Wimmel H.
\newblock Sufficient conditions for the marked graph realisability of labelled
  transition systems.
\newblock \emph{Theor. Comput. Sci.}, 2018.
\newblock \textbf{750}:101--116.
\newblock \doi{10.1016/j.tcs.2017.10.006}.
\newblock \urlprefix\url{https://doi.org/10.1016/j.tcs.2017.10.006}.

\bibitem{ATAEDDEH18}
Devillers R, Erofeev E, Hujsa T.
\newblock Synthesis of Weighted Marked Graphs from Constrained Labelled
  Transition Systems.
\newblock In: Proceedings of the International Workshop on Algorithms {\&}
  Theories for the Analysis of Event Data 2018 Satellite event of the
  conferences: 39th International Conference on Application and Theory of
  {P}etri Nets and Concurrency {P}etri Nets 2018 and 18th International
  Conference on Application of Concurrency to System Design {ACSD} 2018,
  Bratislava, Slovakia, June 25, 2018. 2018 pp. 75--90.
\newblock \urlprefix\url{http://ceur-ws.org/Vol-2115/ATAED2018-75-90.pdf}.

\bibitem{DH19FI}
Devillers R, Hujsa T.
\newblock Analysis and Synthesis of Weighted Marked Graph {P}etri Nets: Exact
  and Approximate Methods.
\newblock \emph{Fundamenta Informaticae}, 2019.
\newblock \textbf{169}(1-2):1--30.
\newblock \urlprefix\url{https://doi.org/10.3233/FI-2019-1837}.

\bibitem{ATAEDDEH19}
Devillers R, Erofeev E, Hujsa T.
\newblock Synthesis of Weighted Marked Graphs from Circular Labelled Transition
  Systems.
\newblock In: Proceedings of the International Workshop on Algorithms {\&}
  Theories for the Analysis of Event Data 2019 Satellite event of the
  conferences: 40th International Conference on Application and Theory of
  {P}etri Nets and Concurrency {P}etri Nets 2019 and 19th International
  Conference on Application of Concurrency to System Design {ACSD} 2019,
  ATAED@{P}etri Nets/ACSD 2019, Aachen, Germany, June 25, 2019. 2019 pp. 6--22.
\newblock \urlprefix\url{http://ceur-ws.org/Vol-2371/ATAED2019-6-22.pdf}.

\bibitem{ToPNoCDEH19}
Devillers R, Erofeev E, Hujsa T.
\newblock Synthesis of Weighted Marked Graphs from Constrained Labelled
  Transition Systems: {A} Geometric Approach.
\newblock \emph{Trans. {P}etri Nets Other Model. Concurr.}, 2019.
\newblock \textbf{14}:172--191.
\newblock \doi{10.1007/978-3-662-60651-3\_7}.
\newblock \urlprefix\url{https://doi.org/10.1007/978-3-662-60651-3\_7}.

\bibitem{ArXivDEHToPNoC1920}
Devillers R, Erofeev E, Hujsa T.
\newblock Efficient Synthesis of Weighted Marked Graphs with Circular
  Reachability Graph, and Beyond.
\newblock \emph{CoRR}, 2019.
\newblock \textbf{abs/1910.14387}.
\newblock \eprint{1910.14387}, \urlprefix\url{http://arxiv.org/abs/1910.14387}.

\bibitem{cheung2008augmented}
Cheung KS.
\newblock Augmented marked graphs and the analysis of shared resource systems.
\newblock In: {P}etri Net, Theory and Applications, chapter~17. IntechOpen,
  2008.

\bibitem{Jiao2004}
Jiao L, Cheung TY, Lu W.
\newblock On liveness and boundedness of asymmetric choice nets.
\newblock \emph{Theoretical Computer Science}, 2004.
\newblock \textbf{311}(1):165 -- 197.
\newblock \doi{https://doi.org/10.1016/S0304-3975(03)00359-1}.

\bibitem{RevCompVSRev2016}
Barylska K, Koutny M, Mikulski L, Piatkowski M.
\newblock Reversible Computation vs. Reversibility in {P}etri Nets.
\newblock In: Reversible Computation - 8th International Conference, {RC} 2016,
  Bologna, Italy, July 7-8, 2016, Proceedings. 2016 pp. 105--118.
\newblock \doi{10.1007/978-3-319-40578-0\_7}.
\newblock \urlprefix\url{https://doi.org/10.1007/978-3-319-40578-0\_7}.

\bibitem{RevCompVSRev2018}
Barylska K, Koutny M, Mikulski L, Piatkowski M.
\newblock Reversible computation vs. reversibility in {P}etri nets.
\newblock \emph{Sci. Comput. Program.}, 2018.
\newblock \textbf{151}:48--60.
\newblock \doi{10.1016/j.scico.2017.10.008}.
\newblock \urlprefix\url{https://doi.org/10.1016/j.scico.2017.10.008}.

\bibitem{huang2003property}
Huang H, Jiao L, Cheung TY.
\newblock Property-preserving composition of augmented marked graphs that share
  common resources.
\newblock In: 2003 IEEE International Conference on Robotics and Automation
  (Cat. No. 03CH37422), volume~1. IEEE, 2003 pp. 1446--1451.

\bibitem{DSSP98}
Recalde L, Teruel E, Silva M.
\newblock Modeling and analysis of sequential processes that cooperate through
  buffers.
\newblock \emph{IEEE Transactions on Robotics and Automation}, 1998.
\newblock \textbf{14}(2):267--277.
\newblock \doi{10.1109/70.681245}.

\bibitem{LenderProc2006}
{Lopez-Grao} J, {Colom} J.
\newblock Lender processes competing for shared resources: Beyond the {S4PR}
  paradigm.
\newblock In: 2006 IEEE International Conference on Systems, Man and
  Cybernetics, volume~4. 2006 pp. 3052--3059.
\newblock \doi{10.1109/ICSMC.2006.384584}.

\bibitem{keller}
Keller RM.
\newblock {A Fundamental Theorem of Asynchronous Parallel Computation}.
\newblock In: Sagamore Computer Conference, August 20-23 1974, LNCS Vol. 24.
  1975 pp. 102--112.
\newblock \urlprefix\url{https://doi.org/10.1007/3-540-07135-0\_113}.

\bibitem{esparza1996decidability}
Esparza J.
\newblock Decidability and complexity of {P}etri net problems--an introduction.
\newblock In: Advanced Course on {P}etri Nets. Springer, 1996 pp. 374--428.

\bibitem{Hujsa2015}
Hujsa T, Delosme JM, Munier-Kordon A.
\newblock On the Reversibility of Live Equal-Conflict {P}etri Nets.
\newblock In: Devillers R, Valmari A (eds.), Application and Theory of {P}etri
  Nets and Concurrency. Springer International Publishing, Cham.
\newblock ISBN 978-3-319-19488-2, 2015 pp. 234--253.

\bibitem{tcs97}
Teruel E, Colom JM, Silva M.
\newblock {Choice-Free {{P}etri} Nets: a Model for Deterministic Concurrent
  Systems with Bulk Services and Arrivals}.
\newblock \emph{{IEEE} Transactions on Systems, Man, and Cybernetics, Part
  {A}}, 1997.
\newblock \textbf{27}(1):73--83.
\newblock \doi{10.1109/3468.553226}.
\newblock \urlprefix\url{http://dx.doi.org/10.1109/3468.553226}.

\bibitem{chu1997deadlock}
Chu F, Xie XL.
\newblock Deadlock analysis of {P}etri nets using siphons and mathematical
  programming.
\newblock \emph{IEEE Transactions on Robotics and Automation}, 1997.
\newblock \textbf{13}(6):793--804.

\bibitem{Stepwise83}
Suzuki I, Murata T.
\newblock A Method for Stepwise Refinement and Abstraction of {P}etri Nets.
\newblock \emph{Journal of Computer and System Sciences}, 1983.
\newblock \textbf{27}:51--76.
\newblock \doi{10.1016/0022-0000(83)90029-6}.

\bibitem{WellStruct98}
Ezpeleta J, Garc{\'i}a-Vall{\'e}s F, Colom JM.
\newblock A Class of Well Structured {P}etri Nets for Flexible Manufacturing
  Systems.
\newblock In: Desel J, Silva M (eds.), Application and Theory of {P}etri Nets
  1998. Springer Berlin Heidelberg, Berlin, Heidelberg.
\newblock ISBN 978-3-540-69108-2, 1998 pp. 64--83.

\bibitem{Surapprox2005}
Peres F, Ribet P, Vernadat F, Berthomieu B.
\newblock V\'erification de propri\'et\'es invariantes par surapproximation.
\newblock In: Journ\'ees "Formalisation des Activit\'es Concurrentes"
  (FAC'2005). 2005 pp. 1--12.

\bibitem{JiaoComposition2005}
Jiao L, Huang H, Cheung TY.
\newblock Property-preserving composition by place merging.
\newblock \emph{Journal of Circuits, Systems and Computers}, 2005.
\newblock \textbf{14}(04):793--812.
\newblock \doi{10.1142/S021812660500260X}.
\newblock \urlprefix\url{https://doi.org/10.1142/S021812660500260X}.

\bibitem{Knitting06}
Chao DY.
\newblock Knitting Technique with {TP-PT} Generations for {P}etri Net
  Synthesis.
\newblock \emph{J. Inf. Sci. Eng.}, 2006.
\newblock \textbf{22}(4):909--923.
\newblock
  \urlprefix\url{http://www.iis.sinica.edu.tw/page/jise/2006/200607\_11.html}.

\bibitem{JiaoSharing2008}
Jiao L, Huang H, Cheung TY.
\newblock Handling resource sharing problem using property-preserving place
  fusions of {P}etri nets.
\newblock \emph{Journal of Circuits, Systems and Computers}, 2008.
\newblock \textbf{17}(03):365--387.
\newblock \doi{10.1142/S021812660800437X}.
\newblock \urlprefix\url{https://doi.org/10.1142/S021812660800437X}.

\bibitem{LopezGrao2013}
L{\'o}pez-Grao JP, Colom JM.
\newblock Structural Methods for the Control of Discrete Event Dynamic Systems
  -- The Case of the Resource Allocation Problem.
\newblock In: Seatzu C, Silva M, van Schuppen JH (eds.), Control of
  Discrete-Event Systems: Automata and {P}etri Net Perspectives, pp. 257--278.
  Springer London, London.
\newblock ISBN 978-1-4471-4276-8, 2013.
\newblock \doi{10.1007/978-1-4471-4276-8_13}.
\newblock \urlprefix\url{https://doi.org/10.1007/978-1-4471-4276-8_13}.

\bibitem{SCECS96}
Recalde L, Teruel E, Silva M.
\newblock {\{}SC{\}}*{ECS}: A class of modular and hierarchical cooperating
  systems.
\newblock In: Billington J, Reisig W (eds.), Application and Theory of {P}etri
  Nets 1996. Springer Berlin Heidelberg, Berlin, Heidelberg.
\newblock ISBN 978-3-540-68505-0, 1996 pp. 440--459.

\bibitem{BeyondAC1998}
Van Der~Aalst W, Kindler E, Desel J.
\newblock Beyond asymmetric choice: A note on some extensions, 1998.

\end{thebibliography}

\end{document}